\documentclass[11pt,a4paper]{article}

\usepackage[utf8]{inputenc}
\usepackage{styling} 
\hypersetup{
  pdftitle={Balanced Allocations with Heterogeneous Bins: The Power of Memory},
  pdfauthor={Dimitrios Los, Thomas Sauerwald, John Sylvester}
  }
\title{Balanced Allocations with Heterogeneous Bins: \\ The Power of Memory\footnote{Full version of a paper appearing in SODA 2023 \cite{LSS23}}}

\author[1]{Dimitrios Los}  
\author[1]{Thomas Sauerwald}
\author[2]{John Sylvester\thanks{J.S. was supported by EPSRC project EP/T004878/1: Multilayer Algorithmics to Leverage Graph Structure.}}
\affil[1]{Department of Computer Science \& Technology, University of Cambridge, UK\\ \texttt{firstname.lastname@cl.cam.ac.uk}} 
\affil[2]{Department of Computer Science, University of Liverpool, UK\\ \texttt{john.sylvester@liverpool.ac.uk}}
\date{\today}

\begin{document}

\maketitle

\begin{abstract}
We consider the allocation of $m$ balls (jobs) into $n$ bins (servers). In the standard \TwoChoice process, at each step $t=1,2,\ldots,m$ we first sample {\em two} bins uniformly at random and place a ball in the least loaded bin. It is well-known that for any $m \geq n$, this results in a gap (difference between the maximum and average load) of $\log_2 \log n + \Theta(1)$ (with high probability). In this work, we consider the \Memory process~\cite{SP02} where instead of two choices, we only sample one bin per step but we have access to a cache which can store the location of one bin. Mitzenmacher, Prabhakar and Shah~\cite{MPS02} showed that in the lightly loaded case ($m = n$), the \textsc{Memory} process achieves a gap of $\mathcal{O}(\log \log n)$. 

Extending the setting of Mitzenmacher et al.~in two ways, we first allow the number of balls $m$ to be arbitrary, which includes the challenging heavily loaded case where $m \geq n$. Secondly, we follow the heterogeneous bins model of Wieder~\cite{W07}, where the sampling distribution of bins can be biased up to some \emph{arbitrary} multiplicative constant. Somewhat surprisingly, we prove that even in this setting, the \textsc{Memory} process still achieves an $\mathcal{O}(\log \log n)$ gap bound. This is in stark contrast with the \TwoChoice (or any \DChoice with $d=\mathcal{O}(1)$) process, where it is known that the gap diverges as $m \rightarrow \infty$~\cite{W07}. 
Further, we show that for any sampling distribution independent of $m$ (but possibly dependent on $n$) the \Memory process has a gap that can be bounded independently of $m$. Finally, we prove a tight gap bound of $\mathcal{O}(\log n)$ for \Memory in another relaxed setting with heterogeneous (weighted) balls and a cache which can only be maintained for two steps.
\end{abstract}

\section{Introduction}

In this work we examine balls-and-bins processes where the goal is to allocate $m$ balls (jobs or tasks) sequentially into $n$ bins (processors or servers). The balls-and-bins framework a.k.a.~balanced allocations~\cite{ABKU99}~is a popular abstraction for various resource allocation and storage problems such as load balancing, scheduling or hashing (see surveys~\cite{MRS01,W17}). In order to allocate the balls in an efficient and decentralized way, randomized strategies are usually employed which are based on sampling a number of bins for each ball, and then allocating the ball into one of those bins. The far-reaching impact of this framework on both theory and practice was recognized by the ``\!\emph{ACM Paris Kanellakis Theory and Practice Award}'' 2020~\cite{award20}.
 
It is well-known that if each ball is placed in a random bin chosen independently and uniformly (called \OneChoice), then the maximum load is $\Theta( \log n / \log \log n)$ \Whp\footnote{In general, with high probability refers to probability of at least $1 - n^{-c}$ for some constant $c > 0$.} for $m=n$, and $m/n + \Theta( \sqrt{ (m/n) \log n})$ \Whp~for $m \geq n \log n$.  Azar, Broder, Karlin and Upfal~\cite{ABKU99} (and implicitly Karp, Luby and Meyer auf der Heide~\cite{KLM96}) proved the remarkable result that if each ball is placed in the lesser loaded of $d \geq 2$ randomly chosen bins, then the maximum load drops to $\log_d \log n + \Oh(1)$ \Whp, if $m=n$. This dramatic improvement from $d=1$ (\OneChoice) to $d=2$ (\TwoChoice) is known as ``power-of-two-choices''.

Later, Berenbrink, Czumaj, Steger and V\"ocking~\cite{BCSV06} extended the analysis of \DChoice to the so-called ``heavily loaded case'', where $m \geq n$ can be arbitrarily large. In particular, for \TwoChoice an upper bound on the gap (the difference between the maximum and average load) of $ \log_2 \log n + \Oh(1)$ \Whp~was shown. Compared to the lightly loaded case, the heavily loaded case is more challenging as for large enough $m$, arbitrarily bad configurations may be encountered\footnote{Technically, for any $f_1:=f_1(n), \ldots, f_{n-1}:=f_{n-1}(n) \in \Z$ (and $f_n = -f_1 -\ldots - f_{n-1}$), there exists $g := g(n, f_1, \ldots, f_n)$ such that \Whp~there exists a step $0 \leq t \leq g$, satisfying $x^t = \big( \frac{t}{n} + f_1, \ldots, \frac{t}{n} + f_{n-1}, \frac{t}{n} + f_{n}\big)$.}. A general technique based on the hyperbolic cosine potential has been used in~\cite{PTW15} to analyze a large family of processes in the heavily loaded case. Talwar and Wieder~\cite{TW14} extended this technique to recover a slightly weaker bound for \DChoice in the heavily loaded case.

\paragraph{Memory Setting.} In this paper, we revisit the \Memory process introduced by Shah and Prabhakar~\cite{SP02}. In $(d, D)$-\Memory, we have the ability to store $D$ of the bins in a cache. Further at each step, $d$ choices are sampled and a ball is allocated in the least loaded of the $d+D$ bins and then the cache is updated to contain the $D$ least loaded of the $d+D$ bins. Mitzenmacher, Prabhakar and Shah~\cite{MPS02} showed that in the lightly loaded case $(d, 1)$-\Memory achieves a gap of $\log_{f(d)} \log n + \Oh(1)$ for $f(d) \in (2d, 2d+1)$, showing that using one cache is roughly like doubling the number of choices you make. Luczak and Norris~\cite{LN13} proved the same bound for the similar queuing setting (referred to as ``supermarket model with memory''). In this paper we focus on the case of $(1,1)$-\Memory (and just refer to this as \Memory), as it turns out that this version already leads to interesting insights.

A related, classical heuristic is \emph{sticky routing}, which means that a route between nodes will be re-used as long as it is not too congested. This has been analyzed and applied to telephone networks~\cite{GKK88} and switch scheduling algorithms~\cite{GPS02}. Similar to balanced allocations, the performance gain of having a cache of capacity one in Cuckoo hashing was demonstrated by Kirsch, Mitzenmacher and Wieder \cite{KMW09}.

\paragraph{Heterogeneous Bins.}
In many applications including P2P systems, distributed hash tables and decentralized storage centers, bins may not always be identical~\cite{BCM04,W07,GS05}. Byers, Considine and Mitzenmacher~\cite{BCM04} introduced a model where the sampling distribution may deviate from uniform by a constant factor (in fact the largest sampling probability may exceed the uniform distribution by a logarithmic factor). They proved that the same asymptotic gap bound as in the uniform case still holds, as long as $m=n$. Later, \cite{W07} generalized the heavily loaded case to the setting where bins are sampled according to an  \textit{$(a,b)$-biased sampling distribution}, which means that for reals $a \geq 1, b \geq 1$ the sampling probabilities $s_i$ for each bin $i \in [n]$ satisfy $\frac{1}{an}\leq s_i \leq \frac{b}{n} $. A tight dichotomy result was proven in the sense that for any constants $a,b \geq 1$, there is a $d'=d'(a,b)$ such that \DChoice with $d \geq (1+\epsilon) d'$ achieves a gap of $\Oh(\log \log n)$, whereas for \protect\DChoice with $d \leq  (1-\epsilon) d'$ there exist sampling distributions where the gap diverges in $m$.

\subsection{Our Results}

In this work we extend the analysis of the \Memory process to the heavily loaded case and show that the process maintains \Whp~a small gap even with $(a,b)$-biased sampling distributions. We use $\Gap(m)$ to denote the difference between the maximum and average load after $m$ balls have been allocated.
\newcommand{\CachingLogLogN}{
Consider the \Memory process with an $(a,b)$-biased sampling distribution, for any constants $a,b \geq 1$. Then there exists a constant $\kappa:=\kappa(a,b) > 0$ such that for every step $m \geq 1$,  
 \[
\Pro{\Gap(m) \leq \kappa \cdot \log \log n} \geq 1 - n^{-3}.
\]
}
\begin{thm}\label{thm:caching_log_log_n}
\CachingLogLogN
\end{thm} 

To the best of our knowledge, even in the case of a uniform sampling distribution (that is, $a = b=1$), an upper bound of $\mathcal{O}(\log \log n)$ was not known in the heavily loaded case $ m \geq n$ (the best result is an $\mathcal{O}(\log n)$ gap bound from \cite{LSS22}).

The ability of \Memory to still attain a doubly-logarithmic gap for \emph{arbitrary} constants $a, b >1$ shows a significant advantage over \TwoChoice (or \DChoice for constant $d$), where Wieder~\cite{W07} showed that only for \emph{some} (sufficiently small) constants $(a, b)$ the maximum load is doubly logarithmic in $n$, otherwise the maximum load may be unbounded. That is, if the constants $(a,b)$ are large, the maximum load diverges as $m \rightarrow \infty$.

We also complement the upper bound with a lower bound, which proves tightness up to multiplicative constants:
\def\cachelower{
Consider the \Memory process with a uniform sampling distribution. Then there is a constant $\kappa > 0$ such that for every step $m \geq n$,
\[
 \Pro{ \Gap(m) \geq \kappa \cdot \log \log n } \geq 1- n^{-1}. 
\]}
\begin{thm}\label{thm:caching_lower}
\cachelower 
\end{thm}

While \cref{thm:caching_log_log_n} and \cref{thm:caching_lower} settle the asymptotic behaviour of \Memory for $(a,b)$-biased sampling distributions with $a,b$ being constants, our next result explores a wider change of sampling distributions, as only a lower bound on $s_{\min}$ independent of $m$ (but possibly dependent on $n$) is needed\footnote{A related result in a queuing setting can be found in \cite[Theorem~1]{SP02}, however, the gap is not quantified.}. 
The cost of this generality is of course the tightness of the bound. However, this is enough to show yet another stark difference to \TwoChoice, whose gap diverges in $m \rightarrow \infty$ already for $(10,10)$-biased sampling distributions (see~\cite{W07}).  

\def\arbdist{
Consider the \Memory process with any sampling distribution $s=(s_i)_{i\in [n]}$ satisfying $s_{\min} := \min_{i \in [n]} s_i > 0$. Then for every step $m \geq 1$ we have \[
\Pro{\Gap(m) \leq \frac{n^8}{s_{\min}^{10}}} \geq 1 - n^{-1}.\]}

\begin{thm}\label{thm:arbdist}
\arbdist
\end{thm}

The first part of the proof of \cref{thm:caching_log_log_n} is based on analyzing a ``leaky'' version of \Memory, which we call \DWeakMemory. In this version, every $d$ steps the cache is reset and only in those steps the load information is updated (see \cref{sec:processes} for a more formal description). For $d = 2$, this can be interpreted as a sample-efficient variant of the $(1+\beta)$-process\footnote{Recall that the $(1+\beta)$-process \cite{PTW15} at each step performs \TwoChoice with probability $\beta$, otherwise \OneChoice.} with $\beta = 1/2$, where the first ball is allocated using \OneChoice and the second one is placed in the least loaded of the last two choices. This process takes exactly one sample per ball, in comparison to the (expected) $1+\beta$ samples made by the original $(1+\beta)$-process in each step.

\begin{thm}[Corollary of \cref{lem:expectation_bound}]
 For any $a,b\geq 1$ there exist  constants $d:=d(a,b)\geq 2$, $\kappa:=\kappa(a,b) > 0$ such that for the \DWeakMemory process with an $(a,b)$-biased sampling distribution, and every step $m \geq 1$,  
  \[
  \Pro{ \Gap(m) \leq \kappa \cdot \log n} \geq 1-n^{-2}.
 \]
\end{thm}
Note that it is important to choose $d$ sufficiently large (depending on $a$ and $b$). For instance, otherwise one could choose $a>d$ and the gap diverges in $m \rightarrow \infty$\footnote{ Consider a sampling distribution $s$ with $s_i = 1/(2 d^2n)$ for some $i \in [n]$. In a run of $d$ steps, this bin $i$ is chosen at least once with probability less than $d \cdot 1/(2 d^2 n)$, and thus the expected number of balls allocated to $i$ in any run of $d$ steps is less than $d \cdot d \cdot 1/(2 d^2 n) = 1/(2n)$. This implies that the difference between the load of bin $i$ and the average diverges as $m \rightarrow \infty$; hence, the gap must diverge also.}.

A tight lower bound of $\Omega(\log n)$ follows quite easily (\cref{lem:d_weak_memory_lower_bound}) by relating \DWeakMemory to the $(1+\beta)$-process~\cite{PTW15}. These two results  demonstrate that a gap of $\Theta(\log n)$ is possible if the constant $d$ is sufficiently large.

Finally, it turns out that the analysis of this relaxed setting even generalizes to a setting where balls are weighted, as long as weights are drawn independently from a distribution with finite MGF (this setting was also studied for different processes in \cite{PTW15}).

\def\weighted{
Consider the $2$-\WeakMemory process with a uniform sampling distribution. Further, assume the weight of each ball is drawn independently from a distribution $W$ satisfying $\ex{W} = 1$ and $\ex{e^{\lambda W} } < \infty$ for a constant $\lambda> 0$.
Then there exists a constant $\kappa:=\kappa(W) > 0$ such that for every step $m \geq 1$, 
\[
\Pro{\Gap(m) \leq \kappa \cdot \log n} \geq 1 - n^{-2}.
\]}
\begin{thm}\label{thm:weighted}
\weighted
\end{thm}
We believe that the result generalizes to \DWeakMemory for any constant $d \geq 2$, but for simplicity we focus on this special case $d=2$ here. 

\subsection{Challenges and Techniques}
One of the main challenges of analyzing the \Memory process are the strong long-term dependencies which are introduced through the cache. This makes \Memory quite different from \TwoChoice, despite their apparent similarity in terms of the gap bound. In fact, the ability of \Memory to store and fill up a light bin is crucial to achieve the $\mathcal{O}(\log \log n)$ gap under biased sampling distributions. In contrast, the decisions between different steps of \TwoChoice are essentially independent (if one disregards the change of the load vector). Formalizing this, \TwoChoice can be expressed by a time-invariant distribution vector $p$, which specifies the probability $p_i$ to allocate a ball to the $i$-th heaviest bin at each step. This type of analysis, paired with a two-sided exponential (a.k.a.~hyperbolic cosine) potential, was pioneered in \cite{PTW15} (and used thereafter in, e.g.,\cite{ABKLN18,TW14,LSS22}). This technique suffices to prove an $\mathcal{O}(\log n)$ gap bound, which is independent of $m$, and provides a useful starting point for tighter gap bounds, as demonstrated in \cite{TW14,LS22Queries}. 

        For \Memory, establishing this starting point, i.e., \emph{base case}, is already quite challenging. 
        To reduce dependencies, we reduce \Memory to  \DWeakMemory, a version of \Memory, which resets the cache every $d=\mathcal{O}(1)$ steps and uses outdated load information. This reduction is formalized by a coupling of \Memory to \DWeakMemory which bounds the hyperbolic potential of the former by the latter.  
        
        Then, we prove that the hyperbolic potential in \DWeakMemory is $\mathcal{O}(n)$ in expectation at any step. While this on its own suffices for an $\mathcal{O}(\log n)$ gap bound, it does not establish an exponentially concentrated load distribution with high probability. 
        To get this stronger conclusion, we study the interplay between two versions of the hyperbolic potential with two smoothing parameters. 
        Apart from the analysis of \Memory, this result for \DWeakMemory may be of independent interest as it gives some insight into the behavior of a version of \Memory where the cache is reset periodically. Finally, we remark that this concentration result is also useful in establishing a lower bound of $\Omega(\log \log n)$ for \Memory, which holds in all steps $m \geq n$ with high probability.
 
Once the base case is established, we begin in earnest to work on the actual \Memory process with all its long-term correlations. We employ a judicious partition of steps separated by rounds (see Figure~\ref{fig:folded_process}), where we define \emph{rounds} inductively through a ``folding'' procedure; a new round can only start if we did not sample a light bin, either at the beginning of the round, or for a sufficiently large number of consecutive steps.  In particular, at the beginning of a round in the folded process we start with a fresh cache, so we are allocating using \OneChoice. To prove an $\mathcal{O}(\log n)$ gap, it is sufficient to reset the cache every two steps (a process we call $2$-\ResetMemory). However, to prove an $o(\log n)$ gap, we need to look at intervals of length $\omega( 1 )$ to counteract this potentially bad allocation.  Essentially, we prove that the time between two rounds is large enough with sufficiently high probability, meaning that there is a strong drift away from the heavy bins so that their contribution to the potential is reduced. In the spirit of layered induction, we analyze increasingly steeper versions of the hyperbolic potentials that grow super-exponentially, and prove that each of them is linear with high probability. After completing $\mathcal{O}(\log \log n)$ steps of the induction, we finally obtain the desired gap bound of $\mathcal{O}(\log \log n)$. 

The main difference compared to the analyses in \cite{PTW15,LSS22,LS22Batched,LS22Noise} is that we need to consider the change of the potential over increasingly longer intervals of a non-constant number of steps, which also requires a slightly different form of super-exponential potentials. Another additional challenge is that the handling of $(a,b)$-biased sampling distributions requires a lot of caution in the base case, in particular, when computing the allocation probabilities in the coupling of \DWeakMemory with \Memory.

\subsection{Further Related Work}

In \cite{G08}, Godfrey introduced a \DChoice model with correlated choices which can be seen as balanced allocations on hypergraphs. Godfrey's results were later improved in \cite{BBFN12,GMP20}, and in~\cite{BBFN14} a related model with bin capacities was studied.

From a different perspective, V\"ocking~\cite{V99} showed that the performance of \DChoice can be improved through a carefully designed asymmetric protocol. In his protocol, $d$ samples are drawn uniformly, one from each of $d$ disjoint groups of bins and ties are broken asymmetrically.
 Multiple potential functions have also been used to analyse various random greedy processes in combinatorics. In particular \cite{BOHMAN2015379} used a hierarchical set of martingales where each martingale is used to control the martingale below it in the set.

\subsection{Road Map}
In \cref{sec:notation} we provide the necessary notation, and define all processes more formally. In \cref{sec:outline} we present a brief outline of the proof of our upper bound on the gap, which is further divided into the base case (\cref{sec:outline_base}) and the layered induction step (\cref{sec:outline_layered}). Correspondingly, the full proofs are deferred to \cref{sec:base_case} for the base case, and \cref{sec:layered_induction} for the layered induction.  In \cref{sec:relaxed_settings} we prove bounds for arbitrarily biased distributions and weights. The lower bounds are proven in \cref{sec:lower}. Finally, in \cref{sec:conclusions}, we summarize the main results and point to some open problems. 

\section{Notation}\label{sec:notation}

\subsection{Basic Definitions}

We consider the allocation of $m$ balls into $n$ bins, which are labeled $[n]:=\{1,2,\ldots,n\}$. For the moment, the $m$ balls are unweighted (or equivalently, all balls have weight $1$). For any step $t \geq 0$, $x^{t}$ is the $n$-dimensional \emph{load vector}, where $x_i^{t}$ is the number of balls allocated into bin $i$ in the first $t$ allocations. In particular, $x_i^{0}=0$ for every $i \in [n]$. The \emph{gap} at step $t$ is defined as
\[
 \Gap(t) = \max_{i \in [n]} x_i^{t} - \frac{t}{n}.
\]
It will be also convenient to keep the load vector $x$ sorted. To this end, 
relabel the $n$ bins such that $y^{t}$ is a permutation of $x^t-(\frac{t}{n}, \dots, \frac{t}{n})$ and $y_1^{t} \geq y_2^{t} \geq \cdots \geq y_n^{t}$. 
Note that $\sum_{i \in [n]} y_i^t=0$ and $\Gap(t)=y_1^t$. Further, we say that a vector $v=(v_1,v_2,\ldots,v_n)$ \emph{majorizes} $u=(u_1,u_2,\ldots,u_n)$ if for all $1 \leq k \leq n$, the prefix sums satisfy: $\sum_{i=1}^k v_i \geq \sum_{i=1}^k u_i$.

Following~\cite{PTW15}, many allocation processes can be described by a time-invariant \emph{probability allocation vector} $p = (p_i)_{i \in [n]}$, such that at each step $t \geq 1$, $p_i$ is the probability for allocating a ball into the $i$-th most heavily loaded bin.  

By $\mathfrak{F}^t$ we denote the filtration of the process until step $t$, which in particular reveals the samples and  allocations of the first $t$ balls including the load vector $x^t$. For random variables $Y,Z$ we say that $Y$ is stochastically smaller than $Z$ (or equivalently, $Y$ is stochastically dominated by $Z$), and write it as $Y \preceq Z$, if $ \Pro{ Y \geq x } \leq \Pro{ Z \geq x }$ for all real $x$.

\subsection{Processes}\label{sec:processes}
 
In this paper we consider balanced allocation processes where the sampling distribution (the distribution of the bin index supplied to the process when it requests a ``random'' bin) will not necessarily be uniform. We let $\mathcal{S}$ be a sampling distribution over the $n$ bins, this is a vector $(s_1, \dots, s_n)$ and use this to define our processes. 
We emphasize that, even though in the analysis we usually work with the load vector sorted decreasingly at each step by relabeling the bins, the sampling distribution assigns a fixed probability to each bin and is not affected by a relabeling. In particular, $(s_1,\ldots,s_n)$ should not be confused with $(p_1,\ldots,p_n)$.
 
We first give a formal description of the \DChoice process, for any integer $d \geq 1$.
\begin{samepage}
\begin{framed}
\vspace{-.45em} \noindent
\underline{\DChoice\!\!($\mathcal{S}$) Process:} \\
\textsf{Iteration:} For each step $t \geq 0$, sample $d$ bins $i_1, \dots, i_d$ independently according to $\mathcal{S}$. Let $i \in \{ i_1, \ldots, i_d \}$ be one bin with $x_{i}^{t} = \min\{ x_{i_1}^t,\dots, x_{i_d}^t\}$, breaking ties randomly. Then update:  
    \begin{equation*}
     x_{i}^{t+1} = x_{i}^{t} + 1.
 \end{equation*}\vspace{-1.5em}
\end{framed}
\end{samepage}

It is immediate that for the uniform sampling distribution the probability allocation vector of \TwoChoice is
\begin{equation*}
    p_{i} = \frac{2i-1}{n^2}, \qquad \mbox{ for all $i \in [n]$.}
\end{equation*}

Our next process is a generalization of the \Memory process to non-uniform bins, which was first introduced by Prabhakar and Shah \cite{SP02} for uniform sampling distributions. 
  
\begin{samepage}
\begin{framed}
\vspace{-.45em} \noindent
\underline{\Memory\!\!($\mathcal{S}$) Process:}\\  
\textsf{Iteration:} Initialize the process by setting the cached bin $b=\emptyset$. For each step $t \geq 0$, sample a bin $i$ according to $\mathcal{S}$, and update:
\begin{equation*}
    \begin{cases}
       x_{i}^{t+1} = x_{i}^{t} + 1  & \mbox{if $x_{i}^{t} <  x_{b}^{t} $ or $b=\emptyset$} \qquad \mbox{(also update cache $b=i$)}, \\
        x_{i}^{t+1} = x_{i}^{t} + 1  & \mbox{if $x_{i}^{t} =  x_{b}^{t} $}, \\
      x_{b}^{t+1} = x_{b}^{t} + 1 & \mbox{if $x_{i}^{t} >  x_{b}^{t}$}.
   \end{cases}
 \end{equation*}\vspace{-1.5em}
\end{framed}
\end{samepage}
 
Next we introduce our own variant of \textsc{Memory} where every $d$ steps the cache is reset and, for allocations at steps between such resets, the process can only make allocation decisions based on an ordering of the bins by load at the time of the last reset. 

\begin{samepage}
\begin{framed}
	\vspace{-.45em} \noindent
	\underline{$d$-\WeakMemory \!\!($\mathcal{S}$) Process:}\\
	\textsf{Iteration:} For each step $t= d\cdot k$ where $k \in \N$, fix any ordering $\sigma:[n]\rightarrow [n]$ of the bins such that for all $i,j\in [n]$,  $\sigma(i)\leq \sigma(j) $ iff $x_{\sigma(i)}^t \geq x_{\sigma(j)}^t $. Then sample a bin $i$ according to $\mathcal{S}$, and update:
	\begin{equation*}
		x_i^{t+1} = x_i^t + 1, \qquad \mbox{(also update cache $b=i$)}.
	\end{equation*}
	For each step $t=d\cdot k+j$, where $k\in \N$ and $1\leq j<d$, sample a bin $i$ according to $\mathcal{S}$, and update:
	\begin{equation*}
		\begin{cases}
			x_{i}^{t+1} = x_{i}^{t} + 1  & \mbox{if $\sigma(i) >  \sigma(b) $} \qquad \mbox{(also update cache $b=i$)}, \\
			x_{b}^{t+1} = x_{b}^{t} + 1 & \mbox{if $\sigma(i) \leq  \sigma(b)$}.
		\end{cases}
	\end{equation*}\vspace{-1em}
\end{framed}
\end{samepage}
We refer to the sequence of steps $d\cdot k, \dots , d \cdot (k+1) -1$, for any $k\geq 0$, as a \textit{run}. For $d=2$, this process can be seen as a sample efficient variant of the $(1+\beta)$ process for $\beta = 1/2$ (for the uniform sampling distribution $\mathcal{U}$).  

\begin{samepage}
\begin{framed}
	\vspace{-.45em} \noindent
	\underline{$2$-\WeakMemory \!\!($\mathcal{U}$) Process:}\\
	\textsf{Iteration:} For each step $2t \geq 0$, sample independently and uniformly two bins $i_1, i_2 \in [n]$, and update:
	\begin{equation*}
        \begin{cases}
		x_{i_1}^{2t+1} & = x_{i_1}^{2t} + 1,\\
        x_{i}^{2t+2}   & = x_{i}^{2t+1} + 1,
        \end{cases}
	\end{equation*}
    where $i \in \{i_1, i_2 \}$ is such that $x_i^{2t} = \min \big\{ x_{i_1}^{2t}, x_{i_2}^{2t} \big\}$.
\end{framed}
\end{samepage}

We emphasize that for $d>2$, the $\DWeakMemory(\mathcal{S})$ process might have $x_i^{t}=x_b^t$, and $\sigma(i) > \sigma(b)$ in some step $t$, which leads to an increment of the load in bin $i$ but also to the bin in the cache being $i$. This is different to \Memory, which more ``smartly'' increments the load of $i$ without updating the cache.

Another related process which may deserve further study, is the following version of the $\DWeakMemory(\mathcal{S})$, where the cache is also reset every $d$ steps, but the load information is updated in every step. In that sense, the process is a half-way house between $\DWeakMemory(\mathcal{S})$ and $\Memory(\mathcal{S})$.
\begin{samepage}
\begin{framed}
	\vspace{-.45em} \noindent
	\underline{$d$-\ResetMemory \!\!($\mathcal{S}$) Process:}\\
	\textsf{Iteration:} For $t= d\cdot k$ where $k\in \N$, sample a bin $i$ according to $\mathcal{S}$ and update:
	\begin{equation*}
		x_i^{t+1} = x_i^t + 1 \qquad \mbox{and update cache $b=i$}.
	\end{equation*}
	For $t=d\cdot k+j$, where $k\in \N$ and $1\leq j<d$, sample a bin $i$ according to $\mathcal{S}$, and update:
	\begin{equation*}
		\begin{cases}
			x_{i}^{t+1} = x_{i}^{t} + 1  & \mbox{if $x_{i}^{t} <  x_{b}^{t} $} \qquad \mbox{(also update cache $b=i$)}, \\
			x_{i}^{t+1} = x_{i}^{t} + 1  & \mbox{if $x_{i}^{t} =  x_{b}^{t} $}, \\
			x_{b}^{t+1} = x_{b}^{t} + 1 & \mbox{if $x_{i}^{t} >  x_{b}^{t}$}.
		\end{cases}
	\end{equation*}\vspace{-1.5em}
\end{framed}
\end{samepage}

 \section{Outline of the Proof of Theorem \ref{thm:caching_log_log_n}} 
 \label{sec:outline}
 
 In this section we outline from a high level the proof of \cref{thm:caching_log_log_n}, as this is by far the most substantial result in the paper and takes up the lion's share of the analysis. 
 
 The rough idea is that, through a layered induction, we prove a series of bounds for ever steeper potential functions. The $j$-th such potential function, where $j\leq j_{\max}$, roughly has the form $\sum_{i\in [n]}e^{f_j\cdot(y_i^m - g_j) }$ for functions $f_j:=f_j(n)$ and $g_j:=g_j(n)$. As $j$ increases the functions $f_j(n)$ and $g_j(n)$ tend to $\infty$ faster with $n$. For the last ``layer'' $j=j_{\max}$ we have $f_j(n)= \Theta(\log n)$ and $g_j(n)= \Theta(\log\log n)$ and so a bound of $\mathcal{O}(n)$ on the potential at level $j_{\max}$ implies that $\Gap(m) = \max_{i\in [n]} y_i^m =  \mathcal{O}(\log\log n)$. The reason for the layered induction is at each layer $j$ a bound on the potential at layer $j$ gives tail bounds on the number of bins with at least a given load; it is this bound which is used to prove a bound on the potential at layer $j+1$. 
 
 To start this induction we use a bound on the hyperbolic cosine potential function which in itself is strong enough to prove an $\mathcal{O}(\log n)$ gap bound on the process (recall that ultimately our target is $\mathcal{O}(\log \log n)$).
 There are many challenges to making this rigorous. We first outline the base case in more detail, followed by the layered induction.

 \subsection{Base Case}\label{sec:outline_base}
For the normalized load vector $(y_i^{t})_{i\in [n]}$ in step $t$,  the \textit{hyperbolic cosine potential}~\cite{PTW15} with smoothing parameter $\alpha > 0$, is defined as \begin{equation*} 
\Gamma^t   := \sum_{i = 1}^n e^{\alpha y_i^{t}} + \sum_{i = 1}^n e^{-\alpha y_i^{t}}.
\end{equation*} 
This potential penalizes over/underloaded bins with an exponential cost in their normalized load. For the base case of the layered induction, we will use that for some sufficiently small constant $\alpha = \Theta(1)$, \Whp~$\Gamma^t = \Oh(n)$ at an arbitrary step. This will allow us to show that the number of bins with normalized load at least $v$ is at most $\Oh(n \cdot e^{-\alpha v})$, which is essential for the layered induction (\cref{sec:outline_layered}).

 \newcommand{\simplebase}{
 
 Consider the \Memory process with any $(a,b)$-biased sampling distribution, for constants $a,b \geq 1$. Then, there exist constants $c:=c(a,b) \geq 1$ and $0<\alpha:=\alpha(a,b)<1$ such that for the potential $\Gamma := \Gamma(\alpha)$, and any step $t \geq 1$, we have  
\[
\Pro{ \bigcap_{u \in [t, t + n \log^8 n]} \left\{ \Gamma^u \leq 6cn \right\} } \geq 1 - n^{-4}.
\]}

{\renewcommand{\thethm}{\ref{simp}}
	\begin{thm}[Restated]
\simplebase
	\end{thm} }
	\addtocounter{thm}{-1}

 To prove this result, we first need to establish that the expectation of the potential has a drift downwards after a certain number of sufficiently many steps.

\newcommand{\ExpectationBoundStatement}{
Consider the \Memory process with any $(a,b)$-biased sampling distribution, for constants $a,b \geq 1$. Then, there exist constants $\alpha':=\alpha'(a,b)>0$,  $c:=c(a,b) \geq 1$ and $d:=d(a,b)\geq 2$, such that for the potential $\Gamma := \Gamma(\alpha)$ with any $\alpha\leq\alpha'$ and for any step $t \geq 1$,
\begin{align*}
\Ex{\left. \Gamma^{t+d} \,\,\right|\,\, \mathfrak{F}^{t}} \leq \Gamma^{t} \cdot \Big( 1 - \frac{\alpha}{c \cdot n}\Big) + c \cdot \alpha.
\end{align*} 
The same bound holds for the \DWeakMemory process.
}

{\renewcommand{\thethm}{\ref{lem:expectation_bound}}
	\begin{thm}[Restated]
\ExpectationBoundStatement
	\end{thm} }
	\addtocounter{thm}{-1}
	 
The difficulty in proving \cref{lem:expectation_bound} directly is that the cache in the \Memory process introduces very long range dependencies between the choices made in each step. We circumvent this problem by making the \Memory process ``forget'' what it has in the cache every $d$ steps and then do a single allocation of \OneChoice to obtain a new cached bin. We call this process $d$-\ResetMemory.

Although $d$-\ResetMemory has only \emph{bounded} range dependencies between allocations, the process is still somewhat opaque, for instance it is very hard to write down the probability allocation vector -  i.e., the probability it allocates a ball to the $i$-th heaviest bin at a given step. For this reason we further restrict the $d$-\ResetMemory so that all its comparisons happen with outdated information (from the last time it forgot the cache). We call this process \DWeakMemory and it has a more tractable probability allocation vector for a family of sampling distributions we call the $(a,b)$-step distributions. Fortunately, we can prove that these are the worst case sampling distributions for \DWeakMemory from the class of all $(a,b)$-biased distributions. Armed with this information we can show that the probability allocation vector of \DWeakMemory has a bias away from heavily loaded bins, which is enough to prove \cref{lem:expectation_bound} for the \DWeakMemory process. Then, by a coupling, we can relate this to a drop for the \Memory process.   

Having established the expectation bound and drift inequalities in \cref{lem:expectation_bound}, we are able to show that for sufficiently small constant smoothing parameter \Whp~the hyperbolic cosine potential is $\Oh(n)$. We achieve this by proving a general concentration result for the hyperbolic cosine potential \cref{thm:gamma_concentration} which is similar to that in \cite[Section 4]{LS22Queries}. The proof relies on an interplay between two instances $\Gamma_1$ and $\Gamma_2$ of the hyperbolic cosine potential with different smoothing parameters $\alpha_1$ and $\alpha_2$ (with $\alpha_1 $ a constant factor larger than $\alpha_2$) such that for any step $t$ with $\Gamma_1^t = \poly(n)$, then it also holds that $|\Gamma_2^{t+1} - \Gamma_2^t| \leq n^{1/3}$. By conditioning on the bad event that $\Gamma_1 = \poly(n)$ for a sufficiently long interval (which follows \Whp~by Markov's inequality and the union bound), we are able to show that $\Gamma_2$ stabilizes at $\Oh(n)$ using an inequality in the spirit of Method of Independent Bounded Differences.
 
A self-contained and full proof of the base case can be found in \cref{sec:base_case}.

 \subsection{Layered Induction}\label{sec:outline_layered}
 
 \paragraph{(Full) Potentials.} We will be using layered induction over super-exponential potential functions, similar to the ones used in \cite[Section 6]{LS22Queries} and \cite[Section 6]{LS22Noise}, but with some differences~(see discussion on page~\pageref{sec:difference_to_previous_approaches}). We will now define the potential functions (to avoid too many technicalities, we leave the exact definitions of constants and other versions of these potentials to \cref{sec:layered_induction}).
The super-exponential potential functions, for $1 \leq j < j_{\max}  =\Theta(\log \log n)$, are given by 
\[
\Phi_j^t := \sum_{i = 1}^n \Phi_{j, i}^t := \sum_{i = 1}^n e^{\alpha_2 \cdot v^j \cdot (y_i^t - z_j)^+},
\]
where $\alpha_2 > 0$ is a constant, $x^+=\max\{x,0\}$, $z_j := \frac{5v}{\alpha_2} \cdot j$, and $v$ is a sufficiently large constant. Our aim will be to prove that $\Phi_{j_{\max}-1}^t = \Oh(n)$, which will imply that $\max_{i \in [n]} y_i^t \leq z_{j_{\max}-1} + \frac{5v^2}{\alpha_2^2} = \Oh(\log \log n)$. 
In order to prove concentration, we also employ a second version of this potential, denoted by $\Psi_j^t$, which is defined in the same way as $\Phi_j^t$, but uses a larger smoothing parameter $\alpha_1 > \alpha_2$. This interplay is similar in spirit to that of $\Gamma_1$ and $\Gamma_2$ in the base case.
 
 \newcommand{\FoldedFigure}
{
\begin{center}
\definecolor{LightGreen}{HTML}{a8d48c}
\definecolor{LightRed}{HTML}{e35454}

\begin{tikzpicture}[scale=1]

\draw[-stealth] (-0.2,0) to node[pos=1.0,right]{steps} (12.4,0);

\draw[-stealth] (0,-0.2) to node[pos=1.0,above]{
\begin{minipage}{0.2\textwidth}
normalized load of allocated bin
\end{minipage}
} (0,7);

\draw[mark=*, mark size=1.5pt,blue] plot  coordinates{ (0,2) (0.2,2.25) (0.4,2.25) (0.6,2.5) (0.8,1.75) (1,2) (1.2,2) (1.4,2.25) (1.6,2.5) (1.8,2.75) (2,2.75) (2.2,3) (2.4,0.75) (2.6,1) (2.8,1.25) (3,1.5) (3.2,1.75) (3.4,2) (3.6,2.25) (3.8,2.25) (4,2.5) (4.2,2.75) (4.4,2.75) (4.6,2.75) (4.8,2.25) (5,1.5) (5.2,1.75) (5.4,2) (5.6,2) (5.8,2.25) (6,2.25) (6.2,2.5) (6.4,2.75) (6.6,3) (6.8,2.5) (7,2.75) (7.2,3) (7.4,3.25) (7.6,3.25) (7.8,3.5) (8,3.75) (8.2,4) (8.4,4) (8.6,4) (8.8,4.25) (9,3.75) (9.2,4) (9.4,4.25) (9.6,4.5) (9.8,4.5) (10,4.75) (10.2,5) (10.4,5) (10.6,5.25) (10.8,1.5) (11,1.75) (11.2,1.75) (11.4,2) (11.6,2.25) }; 

\draw[dashed] (0,5) to node[pos=0,left]{$z_{j-1}+\frac{5v}{\alpha_2}$} (12.1,5);
\draw[dashed] (0,4) to node[pos=0,left]{$z_{j-1}+\frac{4v}{\alpha_2}$} (12.1,4);
\draw[dashed] (0,2) to node[pos=0,left]{$z_{j-1}+\frac{2v}{\alpha_2}$} (12.1,2);

\draw[rectangle,fill=LightGreen,opacity=0.6] (0,0) rectangle (1,0.5);
\draw[rectangle,fill=LightGreen,opacity=0.6] (1,0) rectangle (2,0.5);
\draw[rectangle,fill=LightGreen,opacity=0.6] (2,0) rectangle (3,0.5);

\draw[rectangle,fill=orange,opacity=0.5] (3,0) rectangle (4,0.5);

\draw[rectangle,fill=LightGreen,opacity=0.6] (4,0) rectangle (5,0.5);
\draw[rectangle,fill=LightGreen,opacity=0.6] (5,0) rectangle (6,0.5);
\draw[rectangle,fill=orange,opacity=0.5] (6,0) rectangle (7,0.5);
\draw[rectangle,fill=LightRed,opacity=0.6] (7,0) rectangle (10.6,0.5);
\draw[rectangle,fill=LightGreen,opacity=0.6] (10.6,0) rectangle (11.6,0.5);

\draw [decorate, thick,
    decoration = {calligraphic brace,amplitude=2mm}] (1,-0.7) to node[pos=0.5,below=6pt]{phase}  (0,-0.7);
    \draw [decorate, thick,
    decoration = {calligraphic brace,amplitude=2mm}] (2,-0.7) to node[pos=0.5,below=6pt]{}  (1,-0.7);
        \draw [decorate, thick,
    decoration = {calligraphic brace,amplitude=2mm}] (3,-0.7) to node[pos=0.5,below=6pt]{}  (2,-0.7);
        \draw [decorate, thick,
    decoration = {calligraphic brace,amplitude=2mm}] (4,-0.7) to node[pos=0.5,below=6pt]{}  (3,-0.7);
\draw [decorate, thick,
    decoration = {calligraphic brace,amplitude=2mm}] (4,-1.7) to node[pos=0.5,below=6pt]{round}  (0,-1.7);

\draw[ultra thick,brown,dotted] (0.8,6) to (0.8,-0.5);
\draw[ultra thick,brown,dotted] (1.2,6) to (1.2,-0.5);
\draw[ultra thick,brown,dotted] (2.4,6) to (2.4,-0.5);
\draw[ultra thick,brown,dotted] (4.8,6) to (4.8,-0.5);
\draw[ultra thick,brown,dotted] (5.2,6) to (5.2,-0.5);
\draw[ultra thick,brown,dotted] (10.8,6) to (10.8,-0.5);

\draw[thick,dashed,color=black] (0,0) to node[pos=0.0,left=3pt,sloped]{$r$} (0,6);
\draw[thick,dashed,color=black] (4,0) to node[pos=0.0,left=3pt,sloped]{$r+1$} (4,6);
\draw[thick,dashed,color=black] (7,0) to node[pos=0.0,left=3pt,sloped]{$r+2$} (7,6);
\foreach \x in {7.2,7.4,...,10.4}
{
\draw[thick,dashed,color=black] (\x,0) to node[pos=0.0,left=3pt,sloped]{} (\x,6);
}
\draw[thick,dashed,color=black] (10.6,0) to node[pos=0.0,left=3pt,sloped]{$r+20$} (10.6,6);
\draw[thick,dashed,color=black] (11.6,0) to node[pos=0.0,left=3pt,sloped]{$r+21$} (11.6,6);
\foreach \x in {1,...,7,10.6,11.6}
{
 \draw[thick] (\x,-0.15) to (\x,0.15);
}
\foreach \x in {0,4,7,7.2,...,10.6,11.6}
{
 \draw[ultra thick] (\x,-0.2) to (\x,0.2);
}
\foreach \x in {0,0.2,0.4,...,11.8}
{
 \draw (\x,-0.1) to (\x,0.1);
}

\end{tikzpicture}
\end{center}
 \caption{Illustration of the phases and rounds of the folded process. Brown lines indicate the first substep within a phase in which a light bin was sampled (as can be seen in the second phase of round $r+1$, this does not necessarily mean that this bin is going to be used for the allocation or for the cache). As illustrated, it is only possible to allocate to a bin with normalized load above $z_{j-1}+\frac{5v}{\alpha_2}$ after a long sequence of red rounds.}
}

 Before outlining the proof of the $\Oh(\log \log n)$ bound further, we define the folded process (an illustration is shown in \cref{fig:folded_process}). On a high level, the folded process partitions steps into runs, separated by rounds. Roughly speaking, each run continues as long as we have recently sampled a light bin. Only if this fails, a new round starts. For technical reasons, we also allow some flexibility in where to allocate the ball, as long as the ball does not go into a more loaded bin. 
 
\newcommand{\Folded}{
\paragraph{The Folded Process.} In the $j$-th layer of the layered induction (for $1 \leq j \leq j_{\max} - 1$), we will be analyzing the following folded process of which \Memory is an instance, as we will verify shortly in \cref{lem:memoryfold}. For this, we group the steps into consecutive rounds (of varying lengths), and refer to the $s$-th step within the round as substep $s$. Further, we let $y_i^{r, s}$ be the normalized load of bin $i$ after substep $s$ of round $r$. Then, we define the \textit{folded process} as follows: 
\begin{itemize}
  \item For each round $r \geq 0$, sample bin $i:=i(r) \in [n]$ according to the sampling distribution $\mathcal{S}$:
  \begin{itemize}
      \item \textbf{Case A:} If $y_i^{r,0} \geq z_{j-1} + \frac{2v}{\alpha_2}$, then allocate one ball to an arbitrary bin $\ell$ with $y_{\ell}^{r, 0} \leq y_i^{r, 0}$, and proceed to the next round.
      \item \textbf{Case B:} Otherwise, start a sequence of  consecutive phases each consisting of $\frac{v}{\alpha_2}$ substeps (that is, each phase $k \geq 1$ consists of substeps $s \in [(k - 1) \cdot \frac{v}{\alpha_2}, k \cdot \frac{v}{\alpha_2})$ within the current round $r$.). In each substep $s$, we sample one bin $i=i(r,s)$ according to $\mathcal{S}$ and allocate one ball to an arbitrary bin $\ell$ with $y_{\ell}^{r, s} \leq z_{j-1} + \frac{4v}{\alpha_2}$. 
      At the end of each phase, we also complete the round if either of the following two conditions hold:
        \begin{itemize}
        \item \textbf{Condition 1:} In none of the substeps $s$ of the current phase did we sample a bin $\ell$ with $y_{\ell}^{r, s} < z_{j-1} + \frac{2v}{\alpha_2}$ at the corresponding substep $s$. 
            \item \textbf{Condition 2:} We have completed $k_j := e^{v^{j+1}} \cdot \log^3 n \leq n^{1/7}$ phases.
            
        \end{itemize}
    
  \end{itemize}
\end{itemize}
 
}

\Folded

 \begin{figure}[h]
\FoldedFigure
\label{fig:folded_process}
\end{figure}

Equipped with the definition of rounds, we then proceed to the analysis of the potentials $\Phi_j$ and $\Psi_j$. Our first task is to derive drop inequalities for these potential functions. One challenge compared to previous analyses is that we need to consider longer and longer time-intervals (as $j$ increases), which are the rounds formally defined via the folded process.

\paragraph{Recovery.} Through an interplay between different potentials, some defined for rounds and some defined for steps, we obtain the desired drop inequalities and conclude that, conditioning on the $\Phi_{j-1}$ potential being small, the potential $\Phi_j$ also has to be small in at least one step not too far into the future. 

\paragraph{Stabilization.} In order to prove that $\Phi_j$ remains small, we exploit the flexibility offered by choosing different smoothing parameters, and prove that a weak bound on the other potential $\Psi_j$ implies that $\Phi_j$ (which has a smaller smoothing parameter) can only change by a small sublinear amount. This allows us to apply a martingale concentration inequality with a bad event (\cref{thm:simplified_chung_lu_theorem_8_5}), in the spirit of the Method of Bounded Independent Differences.

\paragraph{Layered Induction.}
Having established the base case in \cref{sec:outline_base}, we can then put all pieces together by performing the induction step ($j \rightarrow j+1$, $j < j_{\max}=\mathcal{O}(\log \log n)$). For each iteration $j$, we consider increasingly steeper potential functions $\Phi_j$. 
Furthermore, at each layer of the induction we shift the time-interval slightly forward, and increase the offset of the potential gently by an additive constant. This leads to the following key lemma, which is formally stated as follows:
 
\newcommand{\InductionStep}{
Consider the \Memory process with any $(a,b)$-biased sampling distribution, for constants $a,b \geq 1$. Then, for any step $t \geq 1$ and $1 \leq j \leq j_{\max} - 1$, define $\beta_{j} := t + 2jn \log^6 n$, and let $C \geq 6$ be as defined in \eqref{eq:fixingconsts}. Then, assuming it holds that
\[
  \Pro{\bigcap_{s \in [\beta_{j-1},t+n \log^8 n]} \{ \Phi_{j-1}^{s} \leq 2Cn\} } \geq 1 - \frac{(\log n)^{11(j-1)}}{n^4},
\]
then the following also holds,
\[
  \Pro{ \bigcap_{s \in [\beta_j, t+n \log^8 n]} \{ \Phi_{j}^{s} \leq 2Cn \} } \geq 1 - \frac{(\log n)^{11j}}{n^4}.
\]
}

	{\renewcommand{\thethm}{\ref{lem:new_inductive_step}}
		\begin{thm}[Restated]
			\InductionStep
	\end{thm} }
	\addtocounter{thm}{-1}

Once this lemma has been established, the desired gap follows by using the base case (\cref{simp}) as a starting point, and applying the union bound over the iterations $1 \leq j \leq j_{\max}$ to conclude that $\Phi_{j_{\max}}^{m} = \mathcal{O}(n)$, which immediately gives the gap bound of $\mathcal{O}(\log \log n)$ with high probability. This concludes the proof outline; a complete and self-contained proof can be found in \cref{sec:layered_induction}.

\section{Proof of the Base Case} \label{sec:base_case}

For ease of reading and to keep the proof self-contained, several definitions and explanations from \cref{sec:outline_base} are repeated.

In \cite{LSS22}, the authors proved that for the \Memory process on a uniform sampling distribution the bound $\Gap(m) = \Oh(\log n)$ holds \Whp~at an arbitrary step $m\geq 1$. Here, we will apply a different analysis to strengthen that result significantly showing that \Whp~the hyperbolic cosine potential is $\Oh(n)$ for at least $n \log^8 n$ steps of the \Memory process on an $(a,b)$-biased sampling distribution from any step $m\geq1$. This will form the base case for the tighter analysis in \cref{sec:layered_induction} and it  also implies the stronger guarantee that \Whp~the underload satisfies $-\min_{i \in [n]} y_i^t = \Oh(\log n)$.

The analysis will be done using $d$-$\WeakMemory(\mathcal{S})$, a process which resets the cache every $d$ steps. So at step $t$ where $d|t $  the ball is allocated using \OneChoice and that bin is added to the cache. Then at steps $t+1, \dots, t+d$ the ball takes a bin sample and allocates to the least loaded bin of the sample and the cache. This process also makes comparisons based on outdated information and we can prove that, in a sense to be formalised, the process makes worse allocations than \Memory (\cref{lem:resettoweak}). By analyzing these $d$ steps together, we are still able to show that the potential drops over $d$ steps and deduce that \Whp~it is $\Oh(n)$. For $d=2$ this process can be interpreted as a variant of $(1+\beta)$ with $\beta = 1/2$, where the bin choices are correlated. That is, the sample at step $2t$ is being re-used at step $2t+1$. Thus, in effect we only need one sample per allocation; nonetheless, the process achieves that the gap is w.h.p.\ $\Oh(\log n)$. 

\subsection{Preliminaries}\label{sec:prelim}

We define the \textit{hyperbolic cosine potential} $\Gamma$ with smoothing parameter $\alpha > 0$, as\begin{equation*} 
\Gamma^t := \Gamma^t(\alpha) := \Phi^t(\alpha) + \Psi^t(\alpha),
\end{equation*}
where $\Phi^t$ is the overload exponential potential \[
\Phi^t := \Phi^t(\alpha) := \sum_{i = 1}^n \Phi_i^t := \sum_{i = 1}^n e^{\alpha y_i^{t}},
\]
and $\Psi^t$ is the underload exponential potential\[
\Psi^t := \Psi^t(\alpha) := \sum_{i = 1}^n \Psi_i^t := \sum_{i = 1}^n e^{-\alpha y_i^{t}}.
\] 
Following \cite{W07}, recall that a sampling distribution $(s_1,\dots , s_n)$ is called $(a,b)$-biased, if for all $i\in [n]$ it holds that $\frac{1}{an}\leq s_i\leq \frac{b}{n}$, where $s_i$ is the probability that bin $i$ is sampled, and $a,b > 1$. Observe that we can always assume that $b\leq n$ as $(s_1,\dots, s_n)$ is a probability distribution. For reals $a,b \geq 1$ such that $M:=\frac{n(a-1)}{ab-1} $ is an integer, we define the $(a,b)$-step distribution to be the vector $\big(\frac{b}{n}, \dots,\frac{b}{n},\frac{1}{an}, \dots,\frac{1}{an}\big)$ where the first $\frac{n(a-1)}{ab-1} $ indices have value $\frac{b}{n}$ and the rest have value $\frac{1}{an} $.

In \cref{sec:expectation_bound_gamma}, we will prove that the hyperbolic cosine potential is $\Oh(n)$ in expectation for the \Memory process with any $(a, b)$-biased sampling distribution (for $a, b \geq 1$ arbitrary constants):

\begin{thm}\label{lem:expectation_bound}
\ExpectationBoundStatement
\end{thm}

In \cref{sec:gamma_concentration} we prove the following general concentration inequality for the hyperbolic potential, a form of which was proved in \cite[Section 4]{LS22Queries}. 

\newcommand{\HyperbolicCosinePotConcentration}{
Consider any process $\mathcal{P}$ where in each step at most $d \in \N_+$ balls are allocated and consider an arbitrary constant $\kappa \geq 6$. Further, assume for this process that for the hyperbolic potential functions $\Gamma_1 := \Gamma_1(\alpha_1)$ and $\Gamma_2 := \Gamma_2(\alpha_2)$ with smoothing parameters $0 < \alpha_1 < 1/(2d)$ and $\alpha_2 \leq \frac{\alpha_1}{12\kappa}$ respectively, there exists an $\eps > 0$ (with $\alpha_2 \eps \geq n^{-1/6}$) and constants $c_1, c_2 > 0$ (with $c_1 \leq c_2$), such that for any step $t \geq 1$,
\begin{align*}
\Ex{\left. \Gamma_1^{t+1} \,\right|\, \mathfrak{F}^t} \leq \Gamma_1^t \cdot \left(1 - \frac{c_1\alpha_1\eps}{n}\right) + c_2\alpha_1\eps,
\end{align*}
and
\begin{align*}
\Ex{\left. \Gamma_2^{t+1} \,\right|\, \mathfrak{F}^t} \leq \Gamma_2^t \cdot \left(1 - \frac{c_1\alpha_2\eps}{n}\right) + c_2\alpha_2\eps.
\end{align*}
Then, for $c := 2 \cdot \frac{c_2}{c_1} \geq 2$, for any step $t \geq 1$,
\[
\Pro{\Gamma_2^t \leq 3cn} \geq 1 - n^{-\kappa}.
\]}

\begin{thm}[Hyperbolic Cosine Potential Concentration]  \label{thm:gamma_concentration}
\HyperbolicCosinePotConcentration
\end{thm}

Next, using \cref{thm:gamma_concentration}, we will deduce that for a sufficiently small smoothing parameter $\alpha > 0$, \Whp~the hyperbolic cosine potential $\Gamma$ is $\Oh(n)$:

\begin{thm} 
\label{simp}
\simplebase
\end{thm}

   \subsection{Expectation Bound on the Hyperbolic Cosine Potential} \label{sec:expectation_bound_gamma}

 Our aim in this section is to prove \cref{lem:expectation_bound}, which gives an expected drop in the hyperbolic cosine potential $\Gamma$ when it is at least a suitably large constant times $n$. To achieve this we will make use of~\cite[Theorem 3.1]{LS22Batched}, which is a generalization of the main theorem in \cite{PTW15}. This will allow us to prove a drop inequality for $\Gamma$ and bound its expectation, when $\alpha > 0$ is a sufficiently small constant. Before we state it, we recall the following condition on a probability vector:
\begin{itemize}\itemsep0pt
  \item \textbf{Condition \hypertarget{c1}{$\mathcal{C}_1$}}: There exist a constant $\delta \in (0, 1)$ and (not necessarily constant) $\eps \in (0, 1)$, such that for any $1 \leq k \leq \delta \cdot n$,
    \[
    \sum_{i=1}^{k} p_{i} \leq (1 - \epsilon) \cdot \frac{k}{n},
    \]
    and similarly for any $\delta \cdot n +1 \leq k \leq n$,
    \[
     \sum_{i=k}^{n} p_i \geq \left(1 + \epsilon \cdot \frac{\delta}{1-\delta} \right) \cdot \frac{n-k+1}{n}.
    \]
 
\end{itemize}

One example of such a vector is the \TwoChoice probability allocation vector $p$, where $p_i = \frac{2i-1}{n^2}$, which satisfies condition $\mathcal{C}_1$ with $\delta=\frac{1}{4}$ and $\epsilon=\frac{1}{2}$ (e.g., \cite[Proposition 2.3]{LS22Batched}).

\begin{thm}[Theorem 3.1 in \cite{LS22Batched}] \label{thm:ls22_thm_3_1}
Consider any probability vector $p$ satisfying condition $\mathcal{C}_1$ for constant $\delta \in (0, 1)$ and $\eps > 0$, and any load vector $x$ with $\Phi :=\Phi(\alpha, x)$, $\Psi :=\Psi(\alpha, x)$ and $\Gamma :=\Gamma(\alpha, x)$. Further for some $K, \kappa > 0$ define, 
\begin{equation}\label{thm:ls22_thm_3_1a}
\Delta\overline{\Phi} := \sum_{i = 1}^n\Phi_i \cdot \Big(\Big(p_i - \frac{1}{n}\Big) \cdot \kappa \cdot \alpha + K \cdot \kappa \cdot \frac{\alpha^2}{n}\Big),
\end{equation}
and
\begin{equation}\label{thm:ls22_thm_3_1b}
\Delta\overline{\Psi} := \sum_{i = 1}^n \Psi_i \cdot \Big(\Big(\frac{1}{n} - p_i\Big) \cdot \kappa \cdot \alpha + K \cdot \kappa \cdot \frac{\alpha^2}{n}\Big).
\end{equation}
Then, there exists a constant $c := c(\delta) > 0$, such that for any $0 < \alpha < \min\{1, \frac{\eps\delta}{8K}\}$,
\[
\Delta\overline{\Gamma} := \Delta\overline{\Phi} + \Delta\overline{\Psi} \leq -\frac{\eps\delta}{8} \cdot \kappa \cdot \frac{\alpha}{n} \cdot \Gamma + c \cdot \kappa \cdot \eps \cdot \alpha.
\]
\end{thm}
 
The main challenge in proving \cref{lem:expectation_bound} is due to the long range correlations present in the \Memory process, in addition, the non-uniform sampling distributions do not help the situation. Instead we analyze the weaker \DWeakMemory process, which is more amenable to an application of \cref{thm:ls22_thm_3_1}, however the non-uniform sampling distributions are still a fly in the ointment. The first step (\cref{lem:resettoweak}) is to show that, from the perspective of our potentials, on an $(a,b)$-biased sampling vector the \DWeakMemory process is indeed weaker than the \Memory process (up to a small additive error term), so that upper bounds for \DWeakMemory also hold for \Memory. We then show that among all the $(a,b)$-biased distributions, the $(a,b)$-step distribution is the worst case with regard to the potential function change of the \DWeakMemory process (\cref{lem:stepmaj}). Since the $(a,b)$-step distribution has a simple form, we can calculate the probability of allocating $j$ balls to the $i$-th most loaded bin (in a fixed ordering) during a run of $d$ allocations (\cref{lem:pijbdds}). Using these allocation probabilities we can then bound the potential drop over a run of \DWeakMemory and relate this back to the potential of \Memory, to prove \cref{lem:expectation_bound}.

\paragraph{Notation} In what follows let $\mathcal{S}$ be a sampling distribution, and $d\geq 1$, $t_0\geq 0$ be integers such that $d\mid t_0$. Recall that at the start of a run the \DWeakMemory process fixes an ordering $\sigma$ of the bins by load and then uses this static ordering $\sigma$ for all allocation decisions during the run. If the run starts at some time $t_0$ we can assume that the ordering $\sigma$ is $\mathfrak{F}^{t_0}$ measurable. 
p
Let $\widehat{p}_{i, j}:=\widehat{p}_{i, j}(\mathcal{S},t_0,d,\sigma)$ be the probability of allocating $j$ balls to bin $\sigma^{-1}(i)$ within a run of $\DWeakMemory(\mathcal{S})$ starting at $t_0$, where $\sigma^{-1}(i)$ is the $i$-th heaviest bin with respect to $\sigma$ (the ordering fixed at time $t_0$). Analogously, let $p_{i, j}:=p_{i, j}(\mathcal{S},t_0,d,\sigma)$ be the probability of \Memory\!\!($\mathcal{S}$) allocating $j$ balls to bin $\sigma^{-1}(i)$ during the steps $t_0, \dots, t_0+d-1$. Let $\widehat{\eta}_i^{\,t}:=\widehat{\eta}_i^{\,t}(\mathcal{S},t_0,d,\sigma)$ denote the number of balls added to bin $\sigma^{-1}(i)$ during the first $t\leq d$ steps of $\DWeakMemory(\mathcal{S})$, and let ${\eta}_i^{\,t}:={\eta}_i^{\,t}(\mathcal{S},t_0,d,\sigma)$ be the analogous quantity for $\Memory(\mathcal{S})$. If $t=d$ we suppress the superscript to give $\widehat{\eta}_i=\widehat{\eta}_i^{\,d}$ and $\eta_i=\eta_i^{d}$.

Note that $\widehat{p}_{i, j}$ and $\widehat{\eta}_i^{\,t}$ only depend on the ordering $\sigma$ whereas $p_{i,j}$ and ${\eta}_i^{\,t}$ also depend on the load configuration $x^t$. The quantities $p_{i,j}$ and ${\eta}_i^{\,t}$ might seem unnatural given that the \Memory process does not follow $\sigma$, however they will be useful to consider when relating the potential functions of the two processes. We often suppress the notational dependence on $\mathcal{S},t_0,d$ and $\sigma$ when this is clear from the context, in particular when conditioning on $\mathfrak{F}^{t_0}$.

Observe that  $\Ex{\widehat{\eta}_i \mid \mathfrak{F}^{t_0} }=\sum_{j=0}^d\widehat{p}_{i,j}\cdot j $ is the expected number of balls the $i$-th heaviest bin  w.r.t.\ $\sigma$ receives during a run of $\DWeakMemory(\mathcal{S})$ starting at time $t_0$. Since $d$ balls are allocated in a run it follows that the vector $(\widehat{p}_{i})_{i\in[n]}$, given by  $\widehat{p}_{i} =\frac{1}{d}\sum_{j=0}^d\widehat{p}_{i,j}\cdot j$, is a probability distribution. It is tempting to call this the allocation vector of \DWeakMemory however this is not quite right in the sense that if so it would allow the allocation of fractional balls, so we call it the \textit{proxy-allocation vector} of \DWeakMemory.

\subsubsection{Couplings Between Processes}
The \Memory process is quite tricky to work with due to long range dependencies between allocations introduced by the cache. The \DWeakMemory process ``forgets'' the cache every $d$ steps also and makes its decision for steps $t+1, \dots, t+d$ using the load vector $x^t$.  
For a positive integer $d$, we can then relate the results for the weaker \DWeakMemory processes back to the \Memory processes via the following coupling over one run of $d$ allocations.

\begin{lem}\label{lem:resettoweak}
Consider any fixed $d \geq 1$ and the potentials $\Phi := \Phi(\alpha)$ and $\Psi := \Psi(\alpha)$ for any $0 < \alpha < 1/d$. Fix a load vector $x^t$ at step $t$ and order the bins decreasingly by their loads. Let $\mu$ be any $(a, b)$-biased sampling vector, and fix $p_{i, j} := p_{i,j}(\mu,d)$ and $\widehat{p}_{i, j} := \widehat{p}_{i,j}(\mu,d)$. Then,
\begin{align*}
\alpha \cdot \sum_{i = 1}^n \Phi_i^t \cdot \sum_{j = 0}^d p_{i, j} \cdot j &\leq \alpha \cdot \sum_{i = 1}^n \Phi_i^t \cdot \sum_{j = 0}^d \widehat{p}_{i, j} \cdot j + \Phi^t \cdot \frac{\alpha^2}{n} \cdot (2d^3b),\intertext{and}
-\alpha \cdot \sum_{i = 1}^n \Psi_i^t \cdot \sum_{j = 0}^d p_{i, j} \cdot j &\leq -\alpha \cdot \sum_{i = 1}^n \Psi_i^t \cdot \sum_{j = 0}^d \widehat{p}_{i, j} \cdot j + \Psi^t \cdot \frac{\alpha^2}{n} \cdot (2d^3b).
\end{align*}
\end{lem}
\begin{proof}
For the first statement, our goal is to upper bound 
\[  
\alpha \cdot \sum_{i = 1}^n \Phi_i^t \cdot \sum_{j = 0}^d p_{i, j} \cdot j = \alpha \cdot \Ex{\sum_{i = 1}^n \Phi_i^t  \cdot \eta_i \;\Bigg| \;\mathfrak{F}^t},
\]where $\eta_i := \eta_i(\mu,t, d,\sigma)$ and $\widehat{\eta}_i := \widehat{\eta}_i(\mu ,t, d,\sigma)$ for each $i \in [n]$ and any ordering $\sigma$ of the bins by load. We consider a coupling between the two processes \Memory and \DWeakMemory by sampling the same bins in steps $t+1,\ldots,t+d$, starting with the normalized load vector $y^t$. Let $c$ be the cache of the \Memory process at step $t$ and let $\mathcal{D}$ be the event that in steps $t+1, \ldots, t + d$, none of the bins is sampled twice and neither bin $c$ is sampled.

\textbf{Case 1 [$\mathcal{D}$ holds]:} 
 In this case, all sampled bins (and initial cache $c$) are different. Hence, \Memory always stores a lesser loaded bin than the cache of \DWeakMemory, and so it allocates to a lesser loaded bin than \DWeakMemory. Hence,
\begin{equation} \label{eq:exp_bound_A}
\alpha \cdot \Ex{\left. \sum_{i=1}^n \eta_i \cdot \Phi_i^t - \sum_{i=1}^n \widehat{\eta}_i \cdot \Phi_i^t \,\right|\, \mathfrak{F}^t,\, \mathcal{D}} \leq 0.
\end{equation}

\textbf{Case 2 [$\mathcal{D}$ does not hold]:} 
On the other hand, if there is a step $s \in [t+1, t + d]$ such that processes do sample a bin for a second time or sample bin $c$, then for any subsequent step $r \in [s, t+d]$, \Memory could allocate to a bin $i$ and \DWeakMemory to a bin $i'$ such that $y_i^{t} \leq y_{i'}^t + d$, since in $d$ allocations the load of a bin can change by at most $d$, and so the cache of \Memory can be at most $d$ balls larger than that of \DWeakMemory.

In order to upper bound the probability that the event $\neg \mathcal{D}$ occurs, we define $\neg \mathcal{D}^r$ for $1 \leq r \leq d$, the event that in step $t + r$, we sampled a bin for the second time or the cache
\begin{align} \label{eq:coupling_prob_bound}
  \Pro{\neg \mathcal{D}\mid \mathfrak{F}^t } 
    &\leq \Pro{\neg \mathcal{D}^1\mid \mathfrak{F}^t} + \ldots + \Pro{\neg \mathcal{D}^d\mid \mathfrak{F}^t}\notag  \\
    &\leq \frac{b}{n} + 2 \cdot \frac{b}{n} + 3 \cdot \frac{b}{n} + \ldots + d \cdot \frac{b}{n}\notag \\ &\leq \frac{d^2b}{n}.
\end{align}
Since $y_i^t \leq y_{i'}^t + d$, the term $\Phi_i^t$ can be upper bounded by the term $\Phi_{i'}^t$ as follows
\[
\Phi_i^t = e^{\alpha y_i^t} \leq e^{\alpha \cdot (y_{i'}^t + d)} \leq \Phi_{i'}^t \cdot e^{\alpha d} \leq \Phi_{i'}^t \cdot (1 + 2\alpha d) = \Phi_{i'}^t + \Phi_{i'}^t \cdot 2\alpha d.
\]
using that $\alpha \leq \frac{1}{d}$ and that $e^z \leq 1 + 2z$ for any $|z| \leq 1$. Since there are at most $d$ different bins $j_1, \ldots, j_d$ allocated in steps $t+1, \ldots, t+d$, we have that 
\begin{equation} \label{eq:exp_bound_B}
\alpha \cdot \Ex{\left. \sum_{i=1}^n \eta_i \cdot \Phi_i^t - \sum_{i=1}^n \widehat{\eta}_i \cdot \Phi_i^t \,\right|\,  \mathfrak{F}^t ,\,  \neg \mathcal{D}} \leq \sum_{k = 1}^d \Phi_{j_k}^t \cdot 2\alpha d \leq \Phi^t \cdot 2\alpha d.
\end{equation}
Hence, combining \cref{eq:exp_bound_A} and \cref{eq:exp_bound_B} we can bound the expectation
\begin{align*}
\alpha \cdot \Ex{\left.\sum_{i=1}^n \eta_i \cdot \Phi_i^t - \sum_{i=1}^n \widehat{\eta}_i \cdot \Phi_i^t \,\right|\,  \mathfrak{F}^t }
 & = \alpha \cdot \Ex{\left. \sum_{i=1}^n \eta_i \cdot \Phi_i^t - \sum_{i=1}^n \widehat{\eta}_i \cdot \Phi_i^t \,\right|\,\mathfrak{F}^t,  \mathcal{D}} \cdot \Pro{\mathcal{D}\mid \mathfrak{F}^t} \\
 & \quad \quad + \alpha \cdot \Ex{\left. \sum_{i=1}^n \eta_i \cdot \Phi_i^t - \sum_{i=1}^n \widehat{\eta}_i \cdot \Phi_i^t \,\right|\, \mathfrak{F}^t, \, \neg \mathcal{D}} \cdot \Pro{ \neg \mathcal{D}\mid \mathfrak{F}^t} \\
 & \leq 0 + \Phi^t \cdot 2 \alpha^2 d  \cdot \frac{d^2 b}{n} = \Phi^t \cdot \frac{2\alpha^2 d^3 b}{n}.
\end{align*}

Now we proceed similarly for $\Psi$, by upper-bounding the term $-\Psi_i^t$ by $-\Psi_{i'}^t$ \[
-\Psi_{i}^t = -e^{-\alpha y_i^t} \leq -e^{-\alpha \cdot (y_{i'}^t + d)} \leq -\Psi_{i'}^t \cdot e^{-\alpha d} \leq -\Psi_{i'}^t + \Psi_{i'}^t \cdot 2\alpha d,
\]

using that $\alpha \leq \frac{1}{d}$ and that $e^z \leq 1 + 2z$ for any $|z| \leq 1$. Hence,
\[
-\alpha \cdot \Ex{\left. \sum_{i=1}^n \eta_i \cdot \Psi_i^t - \sum_{i=1}^n \widehat{\eta}_i \cdot \Psi_i^t\,\right|\,  \mathfrak{F}^t } \leq \Psi^t \cdot \frac{2\alpha^2 d^3 b}{n}.\qedhere
\]\end{proof}
 
 \subsubsection{Properties of the Allocation Probabilities in \texorpdfstring{$d$}{d}-\WeakMemory}
 
Our result bounding the expected potential drop for the \Memory process (\cref{lem:expectation_bound}) follows from a bound on the expected potential drop for the \DWeakMemory process. This is proved by applying \cref{thm:ls22_thm_3_1}, which requires showing that the probability allocation vector of the process (or a proxy for it) induces a bias away from heavily loaded bins. 

 In this sub-section we gather several results on the proxy-allocation vector of  \DWeakMemory that will allow us to apply \cref{thm:ls22_thm_3_1} to prove a drop in the expected potential over one run. The first two results allow us to analyze the proxy-allocation vector of  \DWeakMemory on a $(a,b)$-step distribution rather than an arbitrary $(a,b)$-biased distribution.

Recall that we say that a vector $v=(v_1,v_2,\ldots,v_n)$ \emph{majorizes} $u=(u_1,u_2,\ldots,u_n)$ if for all $1 \leq k \leq n$, the prefix sums satisfy: $\sum_{i=1}^k v_i \geq \sum_{i=1}^k u_i$.

\begin{lem}\label{lem:stepmaj}Let $d\geq 1$, $t_0\geq 0$, and $a,b> 1$ be such that $M= n\cdot \frac{a-1}{ab-1}$ is an integer. Let $\nu$ be the $(a,b)$-step distribution and $\mu$ be an any $(a,b)$-biased distribution.  Then, $\sum_{k=1}^i\widehat{\eta}_{k}(\mu)\preceq \sum_{k=1}^i \widehat{\eta}_{k}(\nu)$ holds for any $i \in [n]$. Thus, $\left(\Ex{\widehat{\eta}_{j}(\nu)\mid \mathfrak{F}^{t_0}} \right)_{j\in [n]} $ majorizes $\left(\Ex{\widehat{\eta}_{j}(\mu)\mid \mathfrak{F}^{t_0}} \right)_{j\in [n]}$.   
\end{lem}

\begin{proof} We couple $\mathcal{N}$, an instance of \DWeakMemory on a $(a,b)$-step distribution $\nu=(\nu_1, \dots, \nu_n)$, to $\mathcal{M}$, an instance of \DWeakMemory on the $(a,b)$-biased distribution $\mu=(\mu_1,\cdots, \mu_n)$. Both processes will start from the same initial load configuration $x^{t_0}$ and ordering $\sigma$ at time $t_0$. 
	
We give a coupling $(s_\nu,s_\mu)$ of the sampling distributions $s_\nu\sim \nu $ and $s_\mu\sim \mu $ as follows:
	
	\medskip 
	
	For each $k\in [n]$ set $q_k=\min\{\nu_k,\mu_k\}$ and  $q= \sum_{k=1}^nq_k$. Let $X\sim \operatorname{Ber}(q)$. 
	\begin{itemize}
		\item If $X=\mathsf{heads}$  sample a bin $k\in [n]$ with probability $q_k/q$ and set $(s_\nu, s_\mu)=(k,k)$. \item If $X=\mathsf{tails}$ sample $U$ uniformly from $[0,1]$. Set $s_\nu\in [n]$ and $s_\mu \in [n]$ to be the smallest integers satisfying $\sum_{k=1}^{s_\nu}(\nu_k-q_k)/(1-q)\geq U$ and $\sum_{k=1}^{s_\mu}(\mu_k-q_k)/(1-q)\geq U$ respectively.   
	\end{itemize}
	To see that this is a valid coupling of $\nu$ and $\mu$ observe that for any $j\in [n]$,
	\[\Pro{s_\nu = j} = \Pro{X=\mathsf{heads} }\cdot \frac{q_j}{q} + \Pro{X=\mathsf{tails}}\cdot  \frac{\nu_j-q_j}{1-q} = \nu_j,  \] and similarly $\Pro{s_\mu = j} = \mu_j $. We then couple $\mathcal{N}$ and $\mathcal{M}$ ball by ball by, in each step, sampling $(s_\nu,s_\mu)$ according to the coupling and giving bin sample $s_\nu$ to $\mathcal{N}$ and $s_\mu$ to $\mathcal{M}$. We claim that, under the coupling above, at each step $t\leq d$ of the run the following invariant holds: 
\[(\mathcal{J}):\text{ $\mathcal{N}$ samples a bin that is higher in the ordering than the bin sampled by $\mathcal{M}$, i.e., $s_\nu \leq s_\mu$.}\]  
It is clear that invariant $\mathcal{J}$ holds at every step where $X=\mathsf{heads}$ since then $s_\nu = s_\mu$, i.e., the same bin sample is given to both processes under the coupling. 

Otherwise, if $X=\mathsf{tails}$, then for any $j\leq M$, \begin{equation}\label{eq:domsmallj} \sum_{k=1}^{j}\frac{\nu_k-q_k}{1-q} = \sum_{k=1}^{j}\frac{b/n-q_k}{1-q} \geq \sum_{k=1}^{j}\frac{\mu_k-q_k}{1-q}.\end{equation} Likewise for any $j> M$, as in this case $q_j=\min\{\nu_j,\mu_j\} =1/(an)$, we have \[ \sum_{k=j}^{n}\frac{\nu_k-q_k}{1-q} = \sum_{k=j}^{n}\frac{1/(an)-q_k}{1-q}=0 \leq \sum_{k=j}^{n}\frac{\mu_k-q_k}{1-q},  \] and so since $\sum_{k=j}^{n}\frac{\mu_k-q_k}{1-q}=1$ and $\sum_{k=j}^{n}\frac{\mu_k-q_k}{1-q}=1$, for any $j>M$ we obtain \begin{equation}\label{eq:dombigj} \sum_{k=1}^{j}\frac{\nu_k-q_k}{1-q} = 1-\sum_{k=j+1}^{n}\frac{\nu_k-q_k}{1-q} \geq  1-  \sum_{k=j+1}^{n}\frac{\mu_k-q_k}{1-q} = \sum_{k=1}^{j}\frac{\mu_k-q_k}{1-q}, \end{equation}where we take empty sums to be zero. Thus, by \eqref{eq:domsmallj} and \eqref{eq:dombigj}, if $X=\mathsf{tails}$ then $s_\nu \leq s_\mu$.

Observe that, for any $i\in [n]$, if either process is given a bin sample $j\geq i$ at some point then no further balls are allocated to bins $k<i$ as there will be a bin in the cache that was lighter in the initial ordering. Thus  $\sum_{k=1}^i\widehat{\eta}_{k}(\mu)\preceq \sum_{k=1}^i \widehat{\eta}_{k}(\nu) $ holds for any $i\in [n]$, as $\mathcal{J}$ holds at every step of the run. The second claim in the statement then follows from linearity of expectation and the definition of majorization.
\end{proof}

We now show that the proxy-allocation vector for \DWeakMemory on a step distribution is piece-wise non-decreasing. 

 \begin{lem}\label{lem:nondec}Let $d\geq 1$, $t_0\geq 0$, and $a,b> 1$ be such that $M= n\cdot \frac{a-1}{ab-1}$ is an integer. Let $\nu$ be the $(a,b)$-step distribution and $\widehat{\eta}_{i}^{\,t}:=\widehat{\eta}_{i}^{\,t}(\nu)$, where $t\in [d]$. Then, for any $i,i'\in [n]$ satisfying either $i<i' \leq M$ or $M<i<i'$ we have $\widehat{\eta}_{i}^{\,t} \preceq \widehat{\eta}_{i'}^{\,t} $ for any $t\in[d]$.
\end{lem}

\begin{proof}We shall couple two instances $\mathcal{P}$ and $\mathcal{P}'$ of \DWeakMemory with the same $(a,b)$-step sampling distribution, initial load configuration, and ordering $\sigma$, for one run as follows: 
	\begin{quote}
		Given a sequence $(b_j)_{j\in[d]}$ of bins, let $(b_j')_{j\in[d]}$ be such that if $b_j=i$ then $b_j'=i'$, if $b_j=i'$ then $b_j'=i$, and otherwise $b_j=b_j'$.  Then if $\mathcal{P}$ is given $(b_j)_{j\in[d]}$, we couple by giving the sequence $(b_j')_{j\in[d]}$ as input to $\mathcal{P}'$. 
	\end{quote} 
Observe that since either $i<i' \leq M$ or $M<i<i'$, both bin $i$ and $i'$ are sampled with the same probability. Thus the function taking $(b_j)_{j\in[d]}$ to $(b_j')_{j\in[d]}$ is a measure preserving bijection from the set of inputs to a single run in \DWeakMemory to itself, giving a coupling of $\mathcal{P}$ to $\mathcal{P}'$.

 Recall that \DWeakMemory does not update the loads within one run, so the ordering $\sigma$ is fixed. Observe that if the cache contains the $\ell$-th bin, then no ball can be allocated to any bin $b<\ell$ in any subsequent step in the same run. Consider the first time $\tau$ at which $b_\tau\in \{i,i'\}$. If $b_t>i'$ for any $t<\tau$ then no balls are allocated to $i$ or $i'$ during the whole run so the claim in the statement holds. If $i<\min_{t<\tau} b_t<i'$ then a ball is added to bin $i'$ in $\mathcal{P}'$ at time $\tau$ but not to bin $i$ in $\mathcal{P}$ (as it has a lower bin in the cache) so $\widehat{\eta}_{i}^{\,\tau} < \widehat{\eta}_{i'}^{\,\tau} $ and the claim in the statement holds as no further balls can be added to $i$. Otherwise, a single ball is added to bin $i$ in $\mathcal{P}$ and $i'$ in $\mathcal{P}'$ at time $\tau$, thus $\widehat{\eta}_{i}^{\,\tau} = \widehat{\eta}_{i'}^{\,\tau} $, and these bins occupy their respective caches. Going forward these bins are displaced from the cache if and only if a bin $b\notin \{i,i'\}$ strictly higher in the ordering is selected. Since $i'>i$ it follows that any time a ball is allocated to $i$ it is also allocated to $i'$, proving the result.   	\end{proof} 
 	
The final result in this section determines the allocation probabilities of \DWeakMemory on a step distribution exactly. 

 \begin{lem}\label{lem:pijbdds}Let $t_0\geq 0$, $d\geq 1$ and $a,b> 1$ be such that $M= n\cdot \frac{a-1}{ab-1}$ is an integer. Let $\nu$ be the $(a,b)$-step distribution,  $\sigma$ be any ordering of the bins by load at time $t_0$, and $\widehat{p}_{i, j}:=\widehat{p}_{i, j}(\nu,t_0,d,\sigma)$. Then, for any $1\leq j\leq d$, we have
\[\widehat{p}_{i,j} =    \begin{cases}  \frac{b}{n}\cdot \left(\frac{ib}{n}\right)^{j-1}\cdot\left[1- \frac{b }{n-b(i-1)  }\cdot  \left(1 -\left(\frac{b(i-1)}{n}\right)^{d-j}\right) \right] &\text{ if } i\leq M,\\
 \frac{a}{n} \cdot \left(1-\frac{a(n-i)}{n} \right)^{j-1}\cdot\left[1- \frac{a}{a(n-i+1)}\cdot  \left(1 -\left(1- \frac{a(n-i+1)}{n}\right)^{d-j}\right) \right]&\text{ if } M<i\leq n. \end{cases}\]  
\end{lem}

	\begin{proof}Recall that under $\sigma$ the bins are ordered by load so that a bin with heaviest load is at position $1$ and lightest load is at position $n$. The \DWeakMemory process can only compare the loads at time $t$, in particular it makes all comparisons based on an ordering $\sigma$ that is fixed throughout the run. Thus, all references to `most loaded' etc in what follows are w.r.t.\ $\sigma$. 
	
	Observe that if the $i'$-th most loaded bin, where $i'>i$, is sampled then we can never place another ball in the $i$-th most loaded bin within the current run. Also note that if the $i$-th most loaded bin is in the cache then it will remain in the cache until a bin which is higher in the ordering is sampled. Let $e_i$, $f_i$, $g_i$ and $h_i=1-g_i$ be the probabilities that an element equal to, strictly lower than, lower than or equal to, or strictly greater than $i$, respectively, in the ordering is chosen. Then, we claim that for $1\leq j\leq d$ we have
\begin{equation}\label{eq:pijbddraw}\widehat{p}_{i,j} = \sum_{\ell =0 }^{d-j-1} f_i^{\ell}\cdot e_i\cdot g_i^{j-1}\cdot h_i + f_i^{d-j}\cdot e_i\cdot g_i^{j-1}.\end{equation}
		 To see this holds we first consider the second term; this is the case where only the last $j$ balls are allocated to the $i$-th most loaded bin and none before this. Thus, the first $d-j$ balls must be allocated to the $i-1$ most loaded bins, which happens with probability $f_i^{d-j}$. Then we allocate to $i$ followed by $j-1$ samples of bins that are at least as high as $i$ in the ordering, this occurs with probability $e_i\cdot g_i^{j-1}$.   Turning to the sum; each summand corresponds the case where the block of $j$ balls allocated to the $i$-th bin starts $\ell$ places into the run. The first three terms in these  probabilities are similar to before and then the $ h_i$ term is the probability we sample a bin higher than $i$ in the ordering in the $(\ell+j+1)$-th step. 
		 
		  We can simplify \eqref{eq:pijbddraw} using a geometric series, 
		 \begin{equation*}\widehat{p}_{i,j} = e_i\cdot g_i^{j-1} \cdot \left[h_i\cdot  \sum_{\ell =0 }^{d-j-1} f_i^{\ell}  + f_i^{d-j}\right] = e_i\cdot g_i^{j-1}\cdot\left[h_i\cdot   \frac{1 -f_i^{d-j}}{1-f_i }  + f_i^{d-j}\right] . \end{equation*} 
		 Now, observing that $h_i = 1-g_i =1-f_i - e_i  $, we have
		   \begin{equation}\label{eq:pijbddsimp} \widehat{p}_{i,j} 
		 = e_i\cdot g_i^{j-1}\cdot\left[(1-f_i -e_i)\cdot   \frac{1 -f_i^{d-j}}{1-f_i }  + f_i^{d-j}\right] = e_i\cdot g_i^{j-1}\cdot\left[1- \frac{e_i}{1-f_i }\cdot  (1 -f_i^{d-j})\right]. \end{equation}

		 Recall $M:=\frac{n(a-1)}{ab-1} $ from the   definition of the $(a,b)$-step vector and observe that
		  \begin{equation}\label{eq:stepbiastable}    
		 \begin{array}{|l||c|c|c|} 
		 \hline  
		 i & e_i & f_i & g_i\\	\hline  \hline 
		 i\leq M & b/n  &b(i-1)/n  & bi/n\\ 
		 \hline 
		 M<i<n& a/n	& 1- a(n-i+1)/n& 1- a(n-i)/n\\\hline 
		\end{array} 
		 \end{equation}

		 For $i\leq M$, by \eqref{eq:stepbiastable} and \eqref{eq:pijbddsimp}, we have 
		   \begin{equation*}  \widehat{p}_{i,j}  
		 = \frac{b}{n} \cdot \left(\frac{bi}{n} \right)^{j-1}\cdot\left[1- \frac{b/n}{1-b(i-1)/n }\cdot  \left(1 -\left(\frac{b(i-1)}{n}\right)^{d-j}\right) \right]. \end{equation*}Thus simplifying gives the claimed result. Similarly for the case $M<i\leq n$,  
	    \begin{equation*}  \widehat{p}_{i,j}  
	 	= \frac{a}{n} \cdot \left(1-\frac{a(n-i)}{n} \right)^{j-1}\cdot\left[1- \frac{a}{a(n-i+1)}\cdot  \left(1 -\left(1- \frac{a(n-i+1)}{n}\right)^{d-j}\right) \right], \end{equation*} as claimed. 
	\end{proof}
 
\subsubsection{Proof of the Expected Potential Drop Inequality}

We are now ready to prove the first theorem of this section. 

{\renewcommand{\thethm}{\ref{lem:expectation_bound}}
	\begin{thm}[Restated]
\ExpectationBoundStatement
	\end{thm} }
	\addtocounter{thm}{-1}

\begin{proof}[Proof of \cref{lem:expectation_bound}]Let $\mu$ be our given $(a,b)$-biased distribution and observe that we can assume $a,b>1$ since the only $(1,1)$-biased distribution is the uniform distribution, which is also an $(a,b)$-biased distribution for any $a,b>1$. Our aim is to bound the expected change of the $\Gamma$ potential for the \Memory\!\!($\mu$) process over $d$ steps,  where $\Gamma^t=\Phi^t+\Psi^t$. We will begin by bounding the expected change of the $\widehat{\Gamma}$ potential for the $\DWeakMemory(\mu)$ process over $d$ steps (or one run). We then relate the drops in potential for $\DWeakMemory(\mu)$ process to that of $\Memory(\mu)$. For the coupling we start both processes from the same configuration at time $t$, thus, we have $\Phi_i^t= \widehat{\Phi}_i^t$ and $\Psi_i^t= \widehat{\Psi}_i^t$ for all $i\in [n]$.

To begin, let $\widehat{p}_{i,j}:=\widehat{p}_{i,j}(\mu)$. Then, the expected change of the overload potential $\widehat{\Phi}$ for the $i$-th most loaded bin over one run of \DWeakMemory\!\!($\mu$), is given by
\begin{align}\label{eq:basedrop-1}
	\Ex{\left. \widehat{\Phi}_i^{t+d} \,\right|\, \mathfrak{F}^{t}}  & =  \widehat{\Phi}_i^{t} \cdot e^{-d\alpha/n} \cdot  \sum_{j=0}^d\widehat{p}_{i,j}\cdot e^{j\cdot \alpha }    =     \Phi_i^{t} \cdot e^{-d\alpha/n} \Bigg(  1 +  \sum_{j=1}^d\widehat{p}_{i,j}\cdot (e^{j\cdot \alpha }-1)\Bigg) \notag.
	\intertext{Using the inequality $e^{z} \leq 1 + z + z^2$ for $|z| < 1$, since $\alpha < 1/2$, we have}	\Ex{\left. \widehat{\Phi}_i^{t+d} \,\right|\, \mathfrak{F}^{t}}  &\leq  \Phi_i^{t} \cdot \Big( 1 - \frac{d\alpha}{n} + \frac{(d\alpha)^2}{n^2}\Big) \cdot\Bigg(  1 +  \sum_{j=1}^d\widehat{p}_{i,j}\cdot \Big(   j\alpha + (j\alpha)^2\Big) \Bigg)\end{align}
Now, observe that for any $i\in [n]$ bin $i$ must be chosen for it to receive a ball, thus $\widehat{p}_{i,j}\leq 1-(1-b/n)^d \leq bd/n  $ for any $1\leq j\leq d$ by Bernoulli's inequality. Thus, for any $i\in [n]$,  
\begin{equation}\label{eq:error}\sum_{j=1}^d\widehat{p}_{i,j}\cdot j^2 \leq \frac{bd}{n}\cdot \frac{d(d+1)(2d+1)}{6} \leq \frac{bd^4}{n} .\end{equation}  Thus, as $\sum_{j=1}^d\widehat{p}_{i,j}\cdot j \leq \sum_{j=1}^d\widehat{p}_{i,j}\cdot j^2 \leq \frac{bd^4}{n}$, by \eqref{eq:basedrop-1} and \eqref{eq:error} we have 
 \begin{align}\label{eq:basedrop1}\Ex{\left. \widehat{\Phi}_i^{t+d} \,\right|\, \mathfrak{F}^{t}}	 &\leq  \Phi_i^{t} \cdot \Bigg[ 1 + \alpha d\Bigg(\frac{1}{d} \sum_{j=1}^d \widehat{p}_{i,j}\cdot j - \frac{1}{n} \Bigg)+ \alpha^2\Bigg( \sum_{j=1}^d \widehat{p}_{i,j}\cdot  j^2 + \mathcal{O}\Big(\frac{bd^5}{n^2} \Big)  \Bigg)   \Bigg].
	\end{align}
Similarly, for the underload potential $\Ex{\left. \widehat{\Psi}_i^{t+d} \,\right|\, \mathfrak{F}^{t}}  =  \Psi_i^{t} \cdot e^{d\alpha/n} \cdot  \sum_{j=0}^d\widehat{p}_{i,j}\cdot e^{-j\cdot \alpha }$, thus 
\begin{align}\label{eq:basedrop2}
	\Ex{\left. \widehat{\Psi}_i^{t+d} \,\right|\, \mathfrak{F}^{t}}  & \leq \Psi_i^{t} \cdot \Big( 1 +\frac{d\alpha}{n} + \frac{(d\alpha)^2}{n^2}\Big) \cdot\Bigg(  1 +  \sum_{j=1}^d\widehat{p}_{i,j}\cdot \Big(   - j\alpha + (j\alpha)^2\Big) \Bigg)\notag \notag\\
	&=\Psi_i^{t} \cdot \Bigg[ 1 + \alpha d\Bigg(\frac{1}{n}-\frac{1}{d} \sum_{j=1}^d \widehat{p}_{i,j}\cdot j  \Bigg)+ \alpha^2\Bigg( \sum_{j=1}^d \widehat{p}_{i,j}\cdot  j^2 + \mathcal{O}\Big(\frac{bd^5}{n^2} \Big)  \Bigg)  \Bigg].
	\end{align}
We aim to apply \cref{thm:ls22_thm_3_1} to bound $\mathbf{E}[\widehat{\Gamma}^{t+d} \mid \mathfrak{F}^{t}]$. It follows from \eqref{eq:error} that the bounds \eqref{eq:basedrop1} and \eqref{eq:basedrop2} on the overload and underload potentials are in the correct form for comparison with the terms $\Delta\overline{\Phi} $ and $\Delta\overline{\Psi} $ from \cref{thm:ls22_thm_3_1}. However, we must check that the probability allocation vector $(\widehat{p}_i)_{i\in[n]}$, given by $\widehat{p}_i = \frac{1}{d} \sum_{j=1}^d \widehat{p}_{i,j}\cdot j$ for $i\in [n]$, satisfies Condition \COne. We first bound $\widehat{p}_i$ over the heaviest bins. 

We can assume $a,b> 1 $ are such that $M=n\cdot \frac{a-1}{ab-1}$ is an integer, as otherwise we could prove the theorem for any pair $a',b'> 1$ of constants such that $ M=n\cdot \frac{a'-1}{a'b'-1}$ is an integer and $a'\leq a$ and $b'\geq b$. The result holds since the class of $(a',b')$-biased distributions contains the $(a,b)$-biased distributions.   
Secondly, let $\nu$ be the $(a,b)$-step distribution. Then, by \cref{lem:stepmaj}, the vector $\widehat{p}_{i}(\nu)$ induced by $\DWeakMemory(\nu)$ majorizes the corresponding vector $\widehat{p}_{i}(\mu)$ induced by any other $\DWeakMemory(\mu)$ process, where $\mu$ is any $(a,b)$-biased distribution. Thus, by applying \cref{lem:sum_of_convex_is_schur_convex},  $\Phi$ and $\Psi$ are Schur-convex functions, so \[\mathbf{E}\big[\widehat{\Phi}^{t+d}(\mu) \,\big|\, \mathfrak{F}^{t}\big]\leq \mathbf{E}\big[\widehat{\Phi}^{t+d}(\nu) \,\big|\, \mathfrak{F}^{t}\big] \quad \text{and}\quad  \mathbf{E}\big[\widehat{\Psi}^{t+d}(\mu) \,\big|\, \mathfrak{F}^{t}\big]\leq \mathbf{E}\big[\widehat{\Psi}^{t+d}(\nu) \,\big|\, \mathfrak{F}^{t}\big].\] Hence, it suffices to bound $\widehat{p}_{i}(\nu)$ for the $\DWeakMemory(\nu)$ process. Thus let $\widehat{p}_{i,j}:= \widehat{p}_{i,j}(\nu)$ and suppose that $\widehat{p}_{i,j} \leq  \kappa\cdot \gamma^{j-1}$, for some $\kappa:=\kappa(i)$ and $\gamma:=\gamma(i)<1$. Then, by \cref{lem:arthgeo} 
\begin{equation}\label{eq:nonexplicit}
	 \sum_{j=1}^d\widehat{p}_{i,j}\cdot   j \leq  \kappa\cdot  \sum_{j=0}^{d-1} \gamma^{j} \cdot (j+1) \leq  \kappa \left(\frac{1}{1-\gamma}+ \frac{\gamma}{(1-\gamma)^2}\right)  .
\end{equation} 
By \cref{lem:pijbdds}, for any $i\leq M=\frac{n(a-1)}{ab-1} $ and $0\leq j \leq d$,  we have 
\begin{equation}\label{eq:hatpij}\widehat{p}_{i,j} =  \left(\frac{ib}{n}\right)^{j-1}\cdot\frac{b}{n} \left[1- \frac{b }{n-b(i-1)  }\cdot  \left(1 -\left(\frac{b(i-1)}{n}\right)^{d-j}\right)\right].  \end{equation} Observe that $b\cdot M  = n \cdot \frac{ab-b}{ab-1} = n - n\cdot \frac{b-1}{ab-1} <n$. It follows that, for all $0\leq i\leq M$, we have $\frac{b }{n-b(i-1)  }>0$ and $\frac{b(i-1)}{n}<1$. Consequently, for any $i\leq M$ and $0\leq j\leq d$,    \[  \frac{b }{n-b(i-1)  }\cdot  \left(1 -\left(\frac{b(i-1)}{n}\right)^{d-j}\right) \geq 0. \] Thus $\widehat{p}_{i,j} \leq \left(\frac{ib}{n}\right)^{j-1}\cdot\frac{b}{n} $ by \eqref{eq:hatpij} and we can take $ \kappa(i) = \frac{b}{n}  $ and  $\gamma(i) = \frac{bi}{n}$. Hence, by \eqref{eq:nonexplicit}, for $i\leq M$ 
\begin{equation}\label{eq:lowvalp_i}
 \widehat{p}_i = \frac{1}{d}\cdot \sum_{j=1}^d\widehat{p}_{i,j}\cdot   j  \leq \frac{1}{d}\cdot  \frac{b}{n} \left(\frac{n}{n-bi}+ \frac{n\cdot bi}{(n-bi)^2}\right)  =    \frac{bn}{d(n-bi)^2} .
	\end{equation}

	We now prove that the conditions of \cref{thm:ls22_thm_3_1} are met by $(\widehat{p}_i)_{i\in[n]}$.   \begin{clm}\label{clm:c1satisfied} For any $\eps<1$ there exists some $d:=d(a,b)\geq 2$ such that the probability allocation vector $(\widehat{p}_i)_{i\in[n]}$ of the $\DWeakMemory(\nu)$ process satisfies Condition $\mathcal{C}_1$ for $\delta=\frac{a-1}{ab-1}$ and  $\eps<1$. 
\end{clm}
\begin{pocd}{clm:c1satisfied} As $\widehat{p}_i = \frac{1}{d}\sum_{j=0}^dj\cdot \widehat{p}_{i,j}$ is proportional to the expected number of balls allocated to $i$ in one run, we see that $\widehat{p}_i$ is non-decreasing in $i$ for any $i\leq M$ by \cref{lem:nondec}. 
Additionally, since $\widehat{p}_i \leq \frac{bn}{d(n-bi)^2}$ for $i\leq M=\frac{n(a-1)}{ab-1}$ by \eqref{eq:lowvalp_i},  we have 
\begin{align*}\sum_{i=1}^k\widehat{p}_i &\leq \frac{bn}{d}\int_1^{k+1}\frac{1}{(n-bx)^2}\,\mathrm{d}x\\ 
&= \frac{bn}{d}\left[ \frac{1}{b(n-bx)} \right]_1^{k+1} \\
&= \frac{n}{d}\left[  \frac{ 1}{n-b(k+1)} -  \frac{1}{n-b}  \right] \\
&=   \frac{ bnk }{d(n-b)(n-b(k+1))} .				
\end{align*}

		Thus for any $k\leq M-1=\frac{n(a-1)}{ab-1} -1$, if we ensure that $n \geq  2b$, then  
	\begin{equation}\label{eq:upperbddpartsum}\sum_{i=1}^k\widehat{p}_i\leq  \frac{ nb}{dn} \cdot    \frac{ k }{(n/2)\cdot \left( 1-b\cdot \frac{a-1}{ab-1}\right) } =  \frac{2b(ab-1) }{d(b-1)}\cdot \frac{k}{n} . \end{equation} Observe that \[\widehat{p}_M = \frac{bn}{d(n-bM)^2} = \frac{b}{dn}\cdot\frac{(ab-1)^2}{(b-1)^2 }  \] Hence, adding on the value of $\widehat{p}_M$ to the bound on the sum upto $M-1$ gives
	\begin{equation}\label{eq:upperbddsum}\sum_{i=1}^M\widehat{p}_i \leq \frac{2b(ab-1) }{d(b-1)}\cdot \frac{M-1}{n} +\frac{b}{dn}\cdot\frac{(ab-1)^2}{(b-1)^2 }   \leq \frac{2b(ab-1)^2 }{d(b-1)^2}\cdot \frac{M}{n}.  \end{equation}   Thus if we take $\delta=M/n=\frac{a-1}{ab-1}$ then the first part of condition \COne is satisfied as, for any $\eps<1$, if we choose an integer $d \geq \frac{2b(ab-1)^2 }{(b-1)^2(1-\eps)}\geq \frac{2b(ab-1) }{(b-1)(1-\eps)}$ then by \eqref{eq:upperbddpartsum} and \eqref{eq:upperbddsum} we have     \[\sum_{i=1}^{k} \widehat{p}_{i} \leq (1 - \epsilon) \cdot \frac{k}{n}, \]for any $k\leq M$. Thus, once we have fixed an $\eps<1$, we have $\sum_{i=1}^{M} \widehat{p}_{i} \leq (1 - \epsilon) \cdot \delta$. Since the $\widehat{p}_i$'s form a probability vector we must have 
	\[\sum_{i=M+1}^{n}\widehat{p}_i = 1-\sum_{i=1}^{M}\widehat{p}_i \geq 1-  (1 - \epsilon) \cdot \delta  .  \]  Now, since $\widehat{p}_i$ is non-decreasing in $i>M$ by \cref{lem:nondec}, for any $M+1\leq k\leq n$ we have 
	\[\sum_{i=k}^{n}\widehat{p}_i \geq \frac{n-k+1}{(1-\delta)n}\cdot (1-  (1 - \epsilon) \cdot \delta ) = \frac{n-k+1}{n}\cdot \left(1+  \epsilon\cdot\frac{ \delta }{1-\delta} \right), \]  
	and so the $\widehat{p}_i$'s also satisfy the second part of Condition \COne.	\end{pocd}
Now, by \cref{clm:c1satisfied}, for any $a,b>1$ there exists a constant integer $d:=d(a,b)\geq 2$ such that $(\widehat{p}_i)_{i\in[n]}$ satisfies Condition \COne for $\delta=\frac{a-1}{ab-1}$ and $\eps=1/2$. By  \eqref{eq:basedrop1}, \eqref{eq:basedrop2} and \eqref{eq:error} there exists some constant $K:=K(a,b,d)=bd^4 + \mathcal{O}(bd^5/n) \leq 2bd^4$, since $n$ is large, such that $\mathbf{E}[ \widehat{\Phi}_i^{t+d} \,\big|\, \mathfrak{F}^{t}] $,  $\mathbf{E}[\widehat{\Psi}_i^{t+d} \,\big|\, \mathfrak{F}^{t}] $,  $(\widehat{p}_i)_{i\in[n]}$ and $K$ satisfy Conditions \eqref{thm:ls22_thm_3_1a} and \eqref{thm:ls22_thm_3_1b}. Thus, if we fix \begin{equation}\label{eq:ParametersforExpDrop} \delta=\frac{a-1}{ab-1}, \qquad \eps =\frac{1}{2}, \qquad \kappa =d, \qquad \text{and} \qquad K =2bd^4, 
\end{equation} then by \cref{thm:ls22_thm_3_1}, there are constants $c_1,c_2 \geq 1$, such that for any $0<\alpha \leq  \min\{\frac{1}{d}, \frac{\eps \delta }{d\cdot 8K}\}$,
\begin{equation}\label{eq:hatdrop}
\Ex{\left. \widehat{\Gamma}^{t+d} \,\right|\, \mathfrak{F}^{t}}=\Ex{\left. \widehat{\Phi}^{t+d} \,\right|\, \mathfrak{F}^{t}} +\Ex{\left. \widehat{\Psi}^{t+d} \,\right|\, \mathfrak{F}^{t}}  \leq \Gamma^{t} \cdot \Big(1 - \frac{\alpha}{c_1\cdot n}\Big) + c_2 \cdot \alpha.
\end{equation}

Having established a drop in potential for the $\DWeakMemory(\mu)$ process it remains to relate this to the potential of the $\Memory(\mu)$ process, our original goal. We will compare the change in overload potential  of the \Memory process $\Phi$ with that of the $\DWeakMemory(\mu)$ process $\widehat{\Phi}$. Let $p_{i,j}$ be the probability that under the $\Memory(\mu)$ process, bin $i$ receives $j$ balls. If we condition on the value of the cache (known to $\mathfrak{F}^t$ these probabilities are fixed. Observe that, by the same steps, one can derive analogous bound to \eqref{eq:basedrop1} but for the $\Memory(\mu)$ process. Hence, by applying the bound on parts of the expression from \cref{lem:resettoweak}, we have 
\begin{align*} 
	\Ex{\left. {\Phi}^{t+d} \,\right|\, \mathfrak{F}^{t}}  &\leq \sum_{i=1}^n\Phi_i^{t} \cdot \Bigg[ 1 + \alpha d\Bigg(\frac{1}{d} \sum_{j=1}^d {p}_{i,j}\cdot j - \frac{1}{n} \Bigg)+ \alpha^2\Bigg( d\cdot \sum_{j=1}^d {p}_{i,j}\cdot  j + \mathcal{O}\Big(\frac{bd^5}{n^2} \Big) \Bigg)  \Bigg]\\
	&\leq \sum_{i=1}^n\Phi_i^{t} \cdot \Bigg[ 1 + \alpha d\Bigg(\frac{1}{d} \sum_{j=1}^d \widehat{p}_{i,j}\cdot j - \frac{1}{n} \Bigg)+ \frac{2d^3\alpha^2}{n} +\alpha^2\Bigg( d\cdot \sum_{j=1}^d \widehat{p}_{i,j}\cdot  j + \mathcal{O}\Big(\frac{bd^4}{n} \Big)  \Bigg)  \Bigg]\\
	&\leq 	\Ex{\left. \widehat{\Phi}^{t+d} \,\right|\, \mathfrak{F}^{t}}  + \Phi^t\cdot \frac{C\alpha^2}{n} , 
\end{align*}
for some $C:=C(a,b,d)$, where in the last inequality we have used that $ \sum_{j=1}^d \widehat{p}_{i,j}\cdot  j\leq  \sum_{j=1}^d \widehat{p}_{i,j}\cdot  j^2 \leq bd^4/n$ by \eqref{eq:error}. Similarly we have $\mathbf{E}\big[  {\Psi}^{t+d} \,\big|\, \mathfrak{F}^{t}\big] \leq 	\mathbf{E}\big[  \widehat{\Psi}^{t+d} \,\big|\, \mathfrak{F}^{t}\big]  + \Psi^t\cdot \frac{C\alpha^2}{n}$ and thus, \begin{equation}\label{eq:tilhat}	\Ex{\left. {\Gamma}^{t+d} \,\right|\, \mathfrak{F}^{t}} \leq 	\Ex{\left. \widehat{\Gamma}^{t+d} \,\right|\, \mathfrak{F}^{t}}  + \Gamma^t\cdot \frac{2C\alpha^2}{n}.\end{equation}
 Thus, by \eqref{eq:tilhat} and \eqref{eq:hatdrop}, there exists a constant $0<\alpha:= \alpha(a,b)$ such that \[
\Ex{\left. \Gamma^{t+d} \,\right|\, x^{t}} \leq \Gamma^{t} \cdot \Big(1 - \frac{ \alpha}{c_1\cdot n}+\frac{\alpha^2\cdot 2C}{n}\Big) + c_2 \cdot \alpha \leq \Gamma^{t} \cdot \Big(1 - \frac{ \alpha}{c\cdot n}\Big) + c \cdot \alpha , 
\] for some $c:=c(a,b)\geq 1$, giving the claim. 
 \end{proof}

\subsection{Hyperbolic Cosine Potential Concentration} \label{sec:gamma_concentration}

In this section, we will prove a quite general theorem for the concentration of the hyperbolic cosine potential. A version of this theorem appeared in~\cite{LS22Queries}.

{\renewcommand{\thethm}{\ref{thm:gamma_concentration}}
	\begin{thm}[Restated]
\HyperbolicCosinePotConcentration
	\end{thm} }
	\addtocounter{thm}{-1}

We will first show how to obtain \cref{simp} as an application of this theorem. 

{\renewcommand{\thethm}{\ref{simp}}
	\begin{thm}[Restated]
\simplebase
	\end{thm} }
	\addtocounter{thm}{-1}

\begin{proof}
Let $\alpha' := \alpha'(a, b)$, $c := c(a, b)$ and $d := d(a, b)$ be the constants in \cref{lem:expectation_bound}. Let $\alpha_1 := \alpha'$ and let $\alpha_2 := \frac{\alpha'}{12\cdot 6}$, then by \cref{lem:expectation_bound} we have that the potentials $\Gamma_1 := \Gamma_1(\alpha_1), \Gamma_2 := \Gamma_2(\alpha_2)$ satisfy the drop inequalities for any step $t \geq 1$
\[
\Ex{\left. \Gamma_1^{t+d} \,\,\right|\,\, \mathfrak{F}^{t}} \leq \Gamma_1^{t} \cdot \Big( 1 - \frac{\alpha_1}{c \cdot n}\Big) + c \cdot \alpha_1,
\]
and 
\[
\Ex{\left. \Gamma_2^{t+d} \,\,\right|\,\, \mathfrak{F}^{t}} \leq \Gamma_2^{t} \cdot \Big( 1 - \frac{\alpha_2}{c \cdot n}\Big) + c \cdot \alpha_2.
\]
Hence, by \cref{thm:gamma_concentration} with $\kappa := 6$, $c_1 := 1/c$ and $c_2 := c$, we obtain that for $c := 2 \cdot \frac{c_2}{c_1} \geq 2$ and for any $k \in \N$,
\[
\Pro{\Gamma_2^{k \cdot d} \leq 3cn} \geq 1 - n^{-6}.
\]
Let $k_0 := \lfloor t/d \rfloor$ and $k_1 := \lceil (t + n \log^8 n)/d \rceil$, and note that $k_1 - k_0 + 1\leq 2 n \log^8 n$. Hence, by the union bound over $k_1 - k_0 + 1$ steps, we have that 
\[
\Pro{\bigcap_{k \in [k_0, k_1]} \left\{ \Gamma_2^{k \cdot d} \leq 3cn \right\}} \geq 1 - n^{-6} \cdot (2n \log^8 n) \geq 1 - n^{-4}.
\]
In any $d$ steps, the contribution of a single bin $i \in [n]$ to $\Gamma_2$ can change by at most a factor of $e^{2\alpha d}$, i.e., for any step $s$, $\Gamma_{2,i}^{s+d} \leq e^{2\alpha d} \cdot \Gamma_{2,i}^{s}$ and by aggregating $\Gamma_2^{s+d} \leq e^{2\alpha d} \cdot \Gamma_2^{s}$. Hence, for any $k \in \N$ and $0 \leq j < d$, for the in-between step $k \cdot d + j$, we have that $\Gamma_2^{k\cdot d + j} \leq e^{2\alpha_2 d} \cdot \Gamma_2^{k\cdot d} \leq 2 \Gamma_2^{k\cdot d}$, using that $\alpha_2 \leq 1/(4d)$. Hence, we conclude that
\[
\Pro{\bigcap_{u \in [t, t + n \log^8 n]} \left\{ \Gamma_2^{u} \leq 6cn \right\}} \geq 1 - n^{-4}. \qedhere
\]
\end{proof}

\subsubsection{Proof outline} \label{sec:concentration_proof_outline}

In this section, we will outline the proof of \cref{thm:gamma_concentration}, giving some intuition for the requirement/choice of the two potential functions $\Gamma_1$ and $\Gamma_2$.

Our goal is to show that \Whp~$\Gamma_2^t \leq 3cn$, for any given $t \geq 0$. We will do this by analyzing the steps in the interval $[t - T_r, t]$, where $T_r := \big\lceil 2 \cdot \frac{4/3 + 2\kappa}{c_1 \alpha_2 \eps} \cdot n \log n \big\rceil$. In particular, in this interval, which we call the \textit{recovery interval}, we will show that \Whp~$\Gamma_2^r \leq cn$ for at least \textit{one} step $r \in [t - T_r, t]$ and then we will show that it \textit{stabilises}, i.e., remains small, for \textit{all} steps in $[r, t]$.

Now, we will give a few more details for the steps in the proof (see \cref{fig:hyperbolic_cosine_potential_proof_outline}). By the expectation bound, we have that for any step $t \geq 0$, $\Ex{\Gamma_1^t} \leq cn$. So, by Markov's inequality \Whp~$\Gamma_1^s \leq cn^{2\kappa+1}$ for all $s \in [t - T_r, t]$.

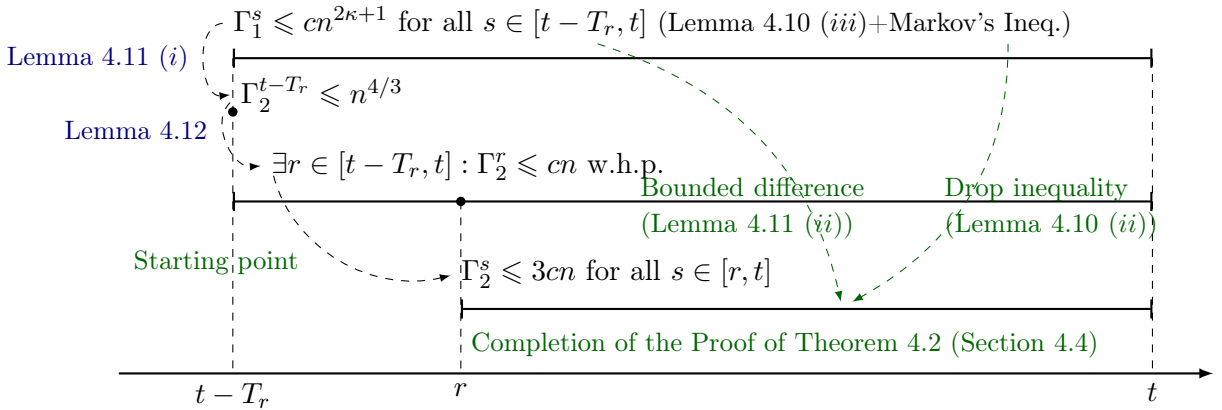
\begin{figure}[H]
\begin{tikzpicture}[
  txt/.style={anchor=west,inner sep=0pt},
  Dot/.style={circle,fill,inner sep=1.25pt},
  implies/.style={-latex,dashed},yscale=0.95]

\def\betaO{0}
\def\betaA{1}
\def\betaB{4}
\def\betaC{7}
\def\End{13.1}
\def\yO{-4.4}

\draw[dashed] (\betaA, 0) -- (\betaA, \yO);
\node[anchor=north] at (\betaA, \yO) {$t - T_r$};
\draw[dashed] (\betaB, -2) -- (\betaB, \yO);
\node[anchor=north] at (\betaB, \yO) {$r$};
\draw[dashed] (\End, 0) -- (\End, \yO);
\node[anchor=north] at (\End, \yO) {$t$};

\node (indstep) at (\betaA,0.5) {};
\node[txt] at (indstep) {$\Gamma_1^{s} \leq cn^{2\kappa + 1}$ for all $s \in [t - T_r, t]$ {\small (\cref{cor:gamma_1_2_expected_change}~$(iii)$+Markov's Ineq.)}};
\draw[|-|, thick] (\betaA, 0) -- (\End, 0) ;

\node (PsiSmall) at (\betaA+0.1,-0.5) {};
\node[txt] at (PsiSmall) {$\Gamma_2^{t - T_r} \leq n^{4/3}$};
\node[Dot] at (\betaA, -0.75){};

\node (ExistsPsiLinear) at (\betaA+0.5,-1.5) {};
\node[txt] at (ExistsPsiLinear) {$\exists r \in [t - T_r, t] : \Gamma_2^{r} \leq cn$ \Whp};
\node[Dot] at (\betaB, -2){};
\draw[|-|, thick] (\betaA, -2) -- (\End, -2);

\node (ExPsiLinear) at (\betaB, -3) {};
\node[txt] at (\betaB, -3) {$\Gamma_2^{s} \leq 3cn$ for all $s \in [r, t]$};
\draw[|-|, thick] (\betaB, -3.5) -- (\End, -3.5);

\draw[-latex, thick] (\betaO - 0.5, \yO) -- (\End + .8, \yO);

\draw[implies] (indstep) edge[bend right=90] (PsiSmall);
\node[anchor=east, black!50!blue] at (\betaA - 0.45, -0) {\small \cref{lem:gamma_1_poly_implies}~$(i)$};

\draw[implies] (PsiSmall) edge[bend right=70] (ExistsPsiLinear);
\node[anchor=east, black!50!blue] at (\betaA - 0.2, -1) {\small \cref{lem:gamma_recovery}};

\draw[implies] (ExistsPsiLinear) edge[bend right=45] (ExPsiLinear);
\node[anchor=east, black!60!green] at (\betaA + 1, -2.85) {\small Starting point};

\draw[implies,black!60!green] (\betaC - 1.2, 0.2) to[bend left=30] (\betaC + 2.0, -3.4);
\node[text width=4cm, anchor=east,black!60!green] at (\betaC + 3.5, -2.1) {\small Bounded difference (\cref{lem:gamma_1_poly_implies}~$(ii)$)};

\draw[implies,black!60!green] (\betaC + 4.2, 0.2) to[bend left=30] (\betaC + 2.15, -3.4);
\node[anchor=east,black!60!green, text width=3cm] at (\betaC + 6.5, -2.1) {\small Drop inequality \\ (\cref{cor:gamma_1_2_expected_change}~$(ii)$)};

\node[black!60!green,anchor=west] at (\betaB, -4) {\small Completion of the Proof of \cref{thm:gamma_concentration} (\cref{sec:completing_high_prob_beta_potential_proof})};

\end{tikzpicture}
\caption{Outline for the proof of \cref{thm:gamma_concentration}. Results in green are used in the application of Azuma's concentration inequality for super-martingales (\cref{lem:azuma})~in \cref{thm:gamma_concentration}.} 
\label{fig:hyperbolic_cosine_potential_proof_outline}
\end{figure}

By the choice of $\alpha_2 \leq \frac{\alpha_1}{12\kappa}$, we will show that when $\Gamma_1^s \leq cn^{2\kappa+1}$, then we also have $(i)$ that $\Gamma_2^s \leq n^{4/3}$ and $(ii)$ that $|\Gamma_2^{s+1} - \Gamma_2^s| \leq n^{1/3}$ (\cref{lem:gamma_1_poly_implies}). The first condition will be useful for proving the \textit{recovery}, i.e., that $\Gamma_2^r \leq cn$ for at least one step $r \in [t - T_r, t]$ (\cref{lem:gamma_recovery}). Then, starting from this step $r$ and using the second condition allows us to use a concentration inequality to deduce that $\Gamma_2$ \textit{stabilises}, i.e., that $\Gamma_2^s \leq 3cn$ for all steps $s \in [r, t]$ (\cref{lem:gamma_1_2_stabilization}).

\subsubsection{Auxiliary lemmas}

In this section, we will prove some auxiliary lemmas for the potentials $\Gamma_1$ and $\Gamma_2$ as defined in \cref{thm:gamma_concentration}.

\begin{lem} \label{cor:gamma_1_2_expected_change}
Consider any process $\mathcal{P}$ satisfying the preconditions of \cref{thm:gamma_concentration}. Then for any step $t \geq 0$,
\begin{align*}
(i) & \quad \Ex{\left. \Gamma_1^{t+1} \,\right|\, \mathfrak{F}^t, \Gamma_1^t > cn} \leq \Gamma_1^t \cdot \left( 1 - \frac{c_1 \alpha_1 \eps}{2n}\right), \\
(ii) & \quad \Ex{\left. \Gamma_2^{t+1} \,\right|\, \mathfrak{F}^t, \Gamma_2^t > cn} \leq \Gamma_2^t \cdot \left( 1 - \frac{c_1 \alpha_2 \eps}{2n}\right), \\
(iii) & \quad \Ex{\Gamma_1^t} \leq cn.
\end{align*}
\end{lem}
\begin{proof}
\textit{First Statement.} Recall that $c = 2 \cdot \frac{c_2}{c_1} \geq 2$. For the first statement, by the assumptions
\begin{align*}
\Ex{\left. \Gamma_1^{t+1} \,\right|\, \mathfrak{F}^t, \Gamma_1^t > cn} 
 & \leq \Gamma_1^t \cdot \left( 1 - \frac{c_1 \alpha_1\eps}{n}\right) + c_2 \alpha_1 \eps\\
 & = \Gamma_1^t \cdot \left( 1 - \frac{c_1 \alpha_1\eps}{2n}\right) - \Gamma_1^t \cdot \frac{c_1 \alpha_1 \eps}{2n} + c_2 \alpha_1 \eps \\
 & \leq \Gamma_1^t \cdot \left( 1 - \frac{c_1 \alpha_1\eps}{2n}\right) - 2 \cdot \frac{c_2}{c_1} \cdot n \cdot \frac{c_1 \alpha_1 \eps}{2n} + c_2 \alpha_1 \eps
 = \Gamma_1^t \cdot \left( 1 - \frac{c_1 \alpha_1\eps}{2n}\right).
\end{align*}

\textit{Second Statement.} Similarly, we obtain the second statement for $\Gamma_2$. 

\textit{Third statement.} By \cref{lem:recurrence_inequality_expectation} for $a = \frac{c_1\alpha_1\eps}{n}$ and $b = c_2\alpha_1\eps$, since $\Gamma_1^0 = 2n \leq 2 \cdot \frac{c_2}{c_1} \cdot n = cn$, it follows that $\ex{\Gamma^t} \leq cn$, for any step $t \geq 0$.
\end{proof}

\begin{lem} \label{lem:gamma_1_poly_implies}
Consider any process $\mathcal{P}$ satisfying the preconditions of \cref{thm:gamma_concentration}. For any step $t \geq 0$ where $\Gamma_1^{t} \leq cn^{2\kappa + 1}$, we have that
\begin{align*}
(i) & \qquad \Gamma_2^{t} \leq n^{4/3}, \\
(ii) & \qquad | \Gamma_2^{t+1} - \Gamma_2^{t} | \leq n^{1/3}.
\end{align*}
\end{lem}

\begin{proof}
Consider an arbitrary step $t$ where $\Gamma_1^t \leq cn^{2\kappa + 1}$. We start by proving the following bound on the normalised load $y_i^t$ for any bin $i \in [n]$,
\[
\Gamma_1^{t} \leq cn^{2\kappa + 1} \Rightarrow e^{\alpha_1\cdot y_i^t} + e^{-\alpha_1\cdot y_i^t} \leq cn^{2\kappa + 1} \Rightarrow 
y_i^t \leq \frac{3\kappa}{\alpha_1} \cdot \log n \, \wedge \,
 -y_i^t \leq \frac{3\kappa}{\alpha_1} \cdot \log n,
\]
where in the second implication we used $\log c + \frac{2 \kappa + 1}{\alpha_1} \log n \leq \frac{3\kappa}{\alpha_1} \log n$, for sufficiently large $n$ as $c$ is a constant and $\kappa \geq 6 \geq 1$.

\emph{First Statement.} Recall that $\alpha_2 \leq \alpha_1/12\kappa$. By the definition of $\Gamma_2^{t}$ and the bound on each normalised bin load, we get that
\[
 \Gamma_2^{t} \leq 2 \cdot \sum_{i=1}^n \exp\Big( \alpha_2 \cdot \frac{3\kappa}{\alpha_1} \cdot \log n \Big)
 = 2n \cdot n^{1/4} \leq n^{4/3}.
\]

\emph{Second Statement.} Consider $\Gamma_2^{t+1}$ as a sum over $2n$ exponentials, which is obtained from $\Gamma_2^{t}$ by slightly changing the values of the $2n$ exponents. The total $\ell_1$-change in the exponents is upper bounded by $4d$, as we will increment $d$ entries in the load vector $x^{t}$ (and each of these entries appear twice), and we will also increment the average load by $\frac{d}{n}$ in all $2n$ exponents. Since $\exp(\cdot)$ is convex, the largest change is upper bounded by the (hypothetical) scenario in which the largest exponent increases by $4d$ and all others remain the same,
\begin{align*}
 \left| \Gamma_2^{t+d} - \Gamma_2^{t} \right| 
 & \leq \exp \Big( \alpha_2 \cdot \big( 4d + \max_{i \in [n]} |y_i^t| \big) \Big) 
   \leq e^{4d \alpha_2} \cdot \exp\Big( \alpha_2 \cdot \frac{3\kappa}{\alpha_1} \cdot \log n \Big) = e^{4d \alpha_2} \cdot n^{1/4} \leq n^{1/3},
\end{align*}
using that $\alpha_2 \leq \frac{\alpha_1}{12\kappa}$ and that $\alpha_2 \leq \alpha_1 \leq 1/(2d)$.
\end{proof}
 
\subsubsection{Recovery and stabilization}

Using the second and third statements in \cref{cor:gamma_1_2_expected_change}, we will now prove a weaker statement of \cref{thm:gamma_concentration}, showing that $\Gamma_2^r \leq cn$ for \textit{at least one} step $r \in [t - T_r, t]$, where $T_r$ is the length of the recovery interval\begin{align}
\label{eq:gamma_concentration_recovery_interval}
T_r := \left\lceil 2 \cdot \frac{4/3 + 2\kappa}{c_1 \alpha_2 \eps} \cdot n \log n \right\rceil.
\end{align}
Before we do this, we proceed by defining an auxiliary process.

\paragraph{Auxiliary process.} Let $\mathcal{P}$ be the process satisfying the preconditions of \cref{thm:gamma_concentration}. We want to condition that $\mathcal{P}$ has $\Gamma^s \leq cn^{2\kappa + 1}$ for every step $s$ in an interval of $\poly(n)$ length, so that we can deduce it satisfies the bounded difference condition (\cref{cor:gamma_1_2_expected_change}) and then apply Azuma's inequality (\cref{lem:azuma}).

To this end, we will define an auxiliary process $\tilde{\mathcal{P}}_{t_0} := \tilde{\mathcal{P}}_{t_0}(\mathcal{P})$ for some arbitrary step $t_0 \geq 0$. Let $\sigma := \inf\left\{ s \geq t_0 : \Gamma^s > \frac{1}{2} cn^{2\kappa + 1} \right\}$. Then, we define $\tilde{\mathcal{P}}$ so that
\begin{itemize}
  \item in steps $[0, \sigma)$ it makes the same allocations as $\tilde{\mathcal{P}}$, and
  \item in steps $[\sigma, \infty)$ it allocates to the currently least loaded bin, i.e., it uses the probability allocation vector $q^s = (0, \ldots, 0, 1)$.
\end{itemize}

\label{sec:auxiliary_process_p_and_tilde_p}
Let $y_{\tilde{\mathcal{P}}}^s$ be the normalised load vector of $\tilde{\mathcal{P}}_{t_0}$ at step $s \geq 0$. By \cref{cor:gamma_1_2_expected_change}~$(iii)$, Markov's inequality and the union bound, it follows that for any interval $[t_0, m]$ with $m - t_0 \leq T_r$, with high probability the two processes agree\begin{align} \label{eq:modified_process_agrees_with_p_gamma_concentration}
\Pro{\bigcap_{s \in [t_0, m]} \left\{ y_{\tilde{\mathcal{P}}}^s = y^s \right\} } 
 & \geq \Pro{\bigcap_{s \in [t_0, m]} \left\{ \Gamma_1^s \leq cn^{2\kappa + 1} \right\} } \geq 1 - n^{-2\kappa} \cdot T_r.
\end{align}
The process $\tilde{\mathcal{P}}_{t_0}$ is defined in a way to satisfy the following property:
\begin{itemize}
  \item (\textbf{Property 1}) The $\tilde{\mathcal{P}}_{t_0}$ process satisfies the drop inequalities for the potential functions $\Gamma_{1, \tilde{\mathcal{P}}}$ and $\Gamma_{2, \tilde{\mathcal{P}}}$ (first and second preconditions) for any step $s \geq 0$.  This holds because for any step $s < \sigma$, the process follows $\mathcal{P}$. For any step $s \geq \sigma$, the process allocates to the currently least loaded bin and therefore minimises the potential $\Gamma_{1, \tilde{\mathcal{P}}}^{s+1}$ given any $\mathfrak{F}^s$, which means that $\Gamma_{1, \tilde{\mathcal{P}}}^{s+1} \leq \Ex{\Gamma_1^{s+1} \mid \mathfrak{F}^s}$ and so it trivially satisfies any drop inequality (and similarly for $\Gamma_{2, \tilde{\mathcal{P}}}$).
\end{itemize}
Further, we define the event that the potential $\Gamma_1$ is small at step $t_0$, as\begin{align} \label{eq:gamma_concentration_z_event_def}
\mathcal{Z}^{t_0} := \left\{  \Gamma_{1, \tilde{\mathcal{P}}}^{t_0} \leq \frac{1}{2}cn^{2\kappa + 1} \right\},
\end{align}
where $c \geq 1$ is the constant defined in \cref{thm:gamma_concentration}. 
When the event $\mathcal{Z}^{t_0}$ holds, then the process $\tilde{\mathcal{P}}_{t_0}$ also satisfies the following property (which ``implements'' the conditioning that $\Gamma_{1, \tilde{\mathcal{P}}}^{s} \leq cn^{2\kappa + 1}$):
\begin{itemize}
  \item (\textbf{Property 2}) For any step $s \geq t_0$, it follows that \[
    \Gamma_{1, \tilde{\mathcal{P}}}^{t_0} \leq cn^{2\kappa + 1},
  \]
  At any step $s \in [t_0, \sigma)$, this holds by the definition of $\sigma$. For any step $s \geq \sigma$, a ball will never be allocated to a bin with $y_i^s > 0$ and in every $n$ steps the at most $n$ bins with load equal to the minimum load (at step $s$) will receive at least one ball each. Hence, over any $n$ steps the maximum absolute normalised load does not increase and in the steps in between this can be larger by at most $1$ and hence, 
  \[
  \Gamma_{1, \tilde{\mathcal{P}}}^{t_0} \leq e^{\alpha_1} \cdot \Gamma_{1, \tilde{\mathcal{P}}}^{\sigma} \leq e^{\alpha_1} \cdot \frac{1}{2}cn^{2\kappa + 1} \leq cn^{2\kappa + 1}.
  \]
 
\end{itemize}

\begin{lem}[Recovery] \label{lem:gamma_recovery}
Consider any step $t \geq 0$ and the auxiliary process $\tilde{\mathcal{P}}_{t - T_r} := \tilde{\mathcal{P}}_{t - T_r}(\mathcal{P})$ for any $\mathcal{P}$ satisfying the preconditions of \cref{thm:gamma_concentration} and with $\mathcal{Z}^{t - T_r}$ being the event defined in \cref{eq:gamma_concentration_z_event_def}. For any step $t \geq 0$, we have that
\[
\Pro{ \left. \bigcup_{r \in [t - T_r, t]} \{ \Gamma_2^{r} \leq cn \} \,\right|\, \mathfrak{F}^{t - T_r}, \mathcal{Z}^{t - T_r} } \geq 1 - 2n^{-2\kappa-1}.
\]
\end{lem}
\begin{proof}
If $t < T_r$, then the statement holds trivially since for $r = 0$, deterministically $\Gamma_2^r = 2n \leq cn$. Otherwise, by the condition $\mathcal{Z}^{t - T_r}$, we have that $\{ \Gamma_1^{t - T_r} \leq cn^{2\kappa + 1} \}$ holds. By \cref{lem:gamma_1_poly_implies}~$(i)$, this implies that $\{ \Gamma_2^{t - T_r} \leq n^{4/3} \}$ also holds. 

By \cref{cor:gamma_1_2_expected_change}~$(ii)$, for any step $s \geq 0$, \begin{equation} \label{eq:conditional_drop_inequality}
	\Ex{\left. \Gamma_2^{s+1} \,\right|\, \mathfrak{F}^{s}, \Gamma_2^s > cn} \leq \Gamma_2^{s} \cdot \left( 1 - \frac{c_1\alpha_2\eps}{2n}\right).
\end{equation}
Next, we define the ``killed'' potential function at steps $s \geq t - T_r$ as
\[
\widehat{\Gamma}_2^{s} := \Gamma_2^{s} \cdot \mathbf{1}_{\bigcap_{ r \in [t - T_r, s]} \left\{ \Gamma_2^{r} > cn \right\} }.
\]
Note that when $\left\{ \Gamma_1^{s} \leq cn \right\}$ then also $\left\{ \widehat{\Gamma}_1^{s} = 0 \right\}$ and $\left\{ \widehat{\Gamma}_1^{s+1} = 0 \right\}$. Therefore, the $\widehat{\Gamma}$ potential unconditionally satisfies the inequality of \cref{eq:conditional_drop_inequality}, that is for any $s \geq t-T_r$
\[
\Ex{\left. \widehat{\Gamma}_2^{s+1} \,\right|\, \mathfrak{F}^{s}} \leq \widehat{\Gamma}_2^{s} \cdot \left(1 - \frac{c_1\alpha_2\eps}{2n}\right).
\]
Inductively applying this for $T_r$ steps, starting with $\widehat{\Gamma}_2^{t - T_r} \leq \Gamma_2^{t - T_r} \leq n^{4/3}$, we get
\begin{align*}
\Ex{\left. \widehat{\Gamma}_2^{t} \,\,\right|\,\, \mathfrak{F}^{t - T_r}, \mathcal{Z}^{t - T_r}} 
 & \leq \Ex{\left. \widehat{\Gamma}_2^{t} \,\,\right|\,\, \mathfrak{F}^{t - T_r}, \Gamma_2^{t - T_r} \leq n^{4/3}} 
 \leq \Big(1 - \frac{c_1\alpha_2\eps}{2n}\Big)^{T_r} \cdot n^{4/3} \\
 & \stackrel{(a)}{\leq} e^{- \frac{1}{2} c_1\alpha_2\eps \cdot \frac{T_r}{n}} \cdot n^{4/3} 
 \stackrel{(b)}{\leq} n^{-2\kappa},
\end{align*}
using in $(a)$ that $1 + u \leq e^u$ (for any $u$) and in $(b)$ that $T_r = \big\lceil 2 \cdot \frac{4/3 + 2\kappa}{c_1 \alpha_2 \eps} \cdot n \log n \big\rceil$. So, by Markov's inequality, \[
 \Pro{\left. \widehat{\Gamma}_2^{t} \leq n  \, \right| \, \mathfrak{F}^{t - T_r}, \mathcal{Z}^{t - T_r}} \geq 1 - n^{-2\kappa-1}.
\]
Since at any step $s$, we have deterministically that $\Gamma_2^s \geq 2n$, we conclude that when $\left\{ \widehat{\Gamma}_2^{t} \leq n \right\}$, then also $\left\{ \widehat{\Gamma}_2^{t} = 0 \right\}$ and so 
\begin{center}
$\mathbf{1}_{\bigcap_{ r \in [t - T_r, t]} \left\{ \Gamma_2^{r} > cn \right\} } = 0$, implying that $\neg \bigcap_{ r \in [t - T_r, t]} \left\{ \Gamma_2^{r} > cn \right\}$ holds
\end{center}
with probability at least $1 - 2n^{-2\kappa}$, concluding the claim.\end{proof}

We will now show that whenever $\Gamma_2^r \in [cn, 2cn]$ holds in for some step $r \in [t - T_r, t]$, then with high probability it $(i)$ remains \textit{small} until step $t$, i.e., $\Gamma_2^s \leq 3cn$ for \textit{all} $s \in [r, t]$ or $(ii)$ it remains small until some step $s \leq t$ where it becomes \textit{very small}, i.e., $\Gamma_2^s \leq cn$.
 
\begin{lem}[Stabilization] \label{lem:gamma_1_2_stabilization}
Consider any step $t \geq 0$ and the auxiliary process $\tilde{\mathcal{P}}_{t - T_r} := \tilde{\mathcal{P}}_{t - T_r}(\mathcal{P})$ for any $\mathcal{P}$ satisfying the preconditions of \cref{thm:gamma_concentration} and with $\mathcal{Z}^{t - T_r}$ being the event defined in \cref{eq:gamma_concentration_z_event_def}. Then, for any step $r \in [t - T_r, t]$ for $T_r$ as defined in \cref{eq:gamma_concentration_recovery_interval}, we have that 
\begin{align*}
& \Pro{\left. \bigcap_{s \in [r, t]} \left\{ \Gamma_2^s \leq 3cn \right\} \cup \bigcup_{s \in (r, t]} \left( \bigcap_{u \in [r, s]} \left\{ \Gamma_2^u \leq 3cn\right\} \cap \left\{ \Gamma_2^s \leq cn \right\} \right) \,\,\right|\,\, \mathcal{Z}^{t - T_r}, \mathfrak{F}^r, \Gamma_2^r \in [cn, 2cn] }
\\ 
& \qquad \geq 1 - n^{-\frac{4}{3} \kappa}.
\end{align*}
\end{lem}
\begin{proof}
Consider an arbitrary step $r \in [t - T_r, t]$ such that $\Gamma_2^r \in [cn, 2cn]$. We define the stopping time \[
\tau := \inf \{ \tilde{r} > r \colon \Gamma_2^{\tilde{r}} \leq cn \},
\]
and for any step $s \in [r, t]$,
\[
  X_{r}^{s} := \Gamma_2^{s \wedge \tau}.
\]

The idea behind this definition is that $X_{r}^{s}$ forms a super-martingale. To see this, note that by \cref{cor:gamma_1_2_expected_change}~$(ii)$, for any step $s < \tau$ we have that \begin{align}
\Ex{\left. X_r^{s+1} \, \right| \, \mathfrak{F}^s, s < \tau} \leq \Ex{\left. \Gamma_2^{s+1} \, \right| \, \mathfrak{F}^s, \Gamma_2^{s} \geq cn} \leq \Gamma_2^{s}, \label{eq:gamma_2_supermartingale_A}
\end{align}
and for any step $s \geq \tau$,
\begin{align}
\Ex{\left. X_r^{s+1} \, \right| \, \mathfrak{F}^s, s \geq \tau} = X_r^{s}. \label{eq:gamma_2_supermartingale_B}
\end{align}

Recall that when $\mathcal{Z}^{t - T_r}$ holds, then by Property 2 (see \cref{sec:auxiliary_process_p_and_tilde_p}), it follows that for every step $s \geq t - T_r$ it holds that $\Gamma_1^s \leq cn^{2\kappa + 1}$. Hence, by \cref{lem:gamma_1_poly_implies}~$(ii)$ it also holds that $|\Gamma_2^{s+1} - \Gamma_2^{s}| \leq n^{1/3}$. Thus, applying \cref{lem:azuma} for any $s \in [r,t]$ gives
\[
\Pro{X_{r}^{s} \geq X_{r}^{r} + cn \,\left|\, \mathcal{Z}^{t - T_r}, \mathfrak{F}^r, \Gamma_2^r \in [cn, 2cn] \right.} \leq \exp\left( - \frac{c^2n^2}{10 \cdot T_r \cdot (n^{1/3})^2 } \right) \leq 2 \cdot T_r \cdot n^{-2\kappa},
\]
using that $T_r = \Oh(n \cdot n^{1/6} \cdot \log n)$. Also recall that for the starting point $r$, $X_{r}^{r} = \Gamma_2^r \leq 2cn$. Therefore, we can conclude that 
\[
\Pro{X_r^s > 3cn \,\left|\, \mathcal{Z}^{t - T_r}, \mathfrak{F}^r, \Gamma_2^r \in [cn, 2cn] \right.} \leq 2 \cdot T_r \cdot n^{-2\kappa}.
\]
By taking the union bound over all $s \in [r, t]$, we get 
\[
\Pro{\left. \bigcap_{s = r}^{t} \left\{ X_r^s \leq 3cn \right\} \,\right|\, \mathcal{Z}^{t - T_r}, \mathfrak{F}^r, \Gamma_2^r \in [cn, 2cn] } \geq 1 - 3 \cdot T_r^2 \cdot n^{-2\kappa} \geq 1 - n^{-\kappa},
\]
using that $\kappa \geq 6$. Now, assuming that $\bigcap_{s = r}^{t} \left\{ X_r^s \leq 3cn \right\}$ holds, we consider the following cases based on the stopping time $\tau$:
\begin{itemize}
  \item \textbf{Case 1 [$\tau > t$]:} Then for all steps $u \in [r, t]$, we have that $\Gamma_2^u = X_r^u \leq 3cn$.
  \item \textbf{Case 2 [$\tau \leq t$]:} Then for all steps $u \in [r, \tau]$, we have that $\Gamma_2^u = X_r^u \leq 3cn$ and $\Gamma_2^{\tau} \leq cn$. So the following event holds for $s = \tau > r$,
  \[
  \bigcup_{s \in (r, t]} \left( \bigcap_{u \in [r, s]} \left\{ \Gamma_2^u \leq 3cn\right\} \cap \left\{ \Gamma_2^s \leq cn \right\} \right).
  \]
\end{itemize}
Hence, this concludes the claim.\end{proof}

\subsection{Completing the proof of Theorem~\ref{thm:gamma_concentration}} \label{sec:completing_high_prob_beta_potential_proof}

To complete the proof of \cref{thm:gamma_concentration}, we will first prove the equivalent statement for the auxiliary process $\tilde{\mathcal{P}}_{t - T_r}$.
 
\begin{lem} \label{lem:auxiliary_process_concentration}
Consider any step $t \geq 0$ and the auxiliary process $\tilde{\mathcal{P}}_{t - T_r} := \tilde{\mathcal{P}}_{t - T_r}(\mathcal{P})$ for any $\mathcal{P}$ satisfying the preconditions of \cref{thm:gamma_concentration} and with $\mathcal{Z}^{t - T_r}$ being the event defined in \cref{eq:gamma_concentration_z_event_def}. Then, for $c := 2 \cdot \frac{c_2}{c_1} \geq 2$, for any step $t \geq 0$,
\[
\Pro{\left. \Gamma_2^t \leq 3cn \,\right|\, \mathcal{Z}^{t - T_r}} \geq 1 - n^{-\kappa}.
\]
\end{lem}

\begin{proof}
The proof will be concerned with steps $\in [t - T_r, t]$. First, by applying
 \cref{lem:gamma_recovery}, 
 \begin{equation} \label{eq:base_case_starting_point}
\Pro{\left. \bigcup_{r_0 \in [t - T_r, t]} \left\{ \Gamma_2^{r_0} \leq cn \right\} \,\right|\, \mathfrak{F}^{t - T_r}, \mathcal{Z}^{t - T_r} } \geq 1 - 2n^{-2\kappa-1}.
 \end{equation}
Consider now an arbitrary step $r_0 \in [t - T_r,t]$ and assume that $\Gamma_2^{r_0} \leq cn$. We partition the time-steps $s \in [r_0, t]$ into red and green phases (see \cref{fig:gamma_2_small_large_regions}):
\begin{enumerate}\itemsep-4pt
    \item \textbf{Red Phase:} The step $s$ is in a red phase if $\Gamma_2^{s} > cn$.
    \item \textbf{Green Phase:} Otherwise, the process is in a green phase.
\end{enumerate}

Note that by the choice of $r_0$, the process is at a green phase at time $r_0$. Then each green phase may be preceded by a red phase. Obviously, for each step $s$ in a green phase, we have $\Gamma_2^{s} \leq cn$. When $s$ is the first step of a red phase after a green phase, it follows that $\Gamma_2^{s} \leq e^{d\alpha_2} \cdot \Gamma_2^{s-1} \leq 2 \cdot \Gamma_2^{s-1} \leq 2cn$, since $0<\alpha_2 <1/(2d)$.

Let $\mathcal{R}^s$ be the event that step $s$ is the first step of a red phase and let $\mathcal{A}^s$ be the event that all steps $u \geq s$ in the same phase as $s$, satisfy $\Gamma_2^u \leq 3cn$. By \cref{lem:gamma_1_2_stabilization}, we have that
\[
\Pro{\mathcal{A}^s \,\left|\, \mathcal{Z}^{t - T_r}, \mathfrak{F}^s, \mathcal{R}^s \right.} \geq 1 - n^{-\frac{4}{3}\kappa}.
\]
For any events $\mathcal{E}_1 \neq \emptyset$ and $\mathcal{E}_2$, we have that $
\Pro{\mathcal{E}_2 \cup \neg \mathcal{E}_1} \geq 1 - \Pro{\neg \mathcal{E}_2 \mid \mathcal{E}_1}
$
and hence\[
\Pro{\mathcal{A}^s \cup \neg \mathcal{R}^s \,\left|\, \mathcal{Z}^{t - T_r}, \mathfrak{F}^s \right.} \geq 1 - n^{-\frac{4}{3}\kappa}.
\]
By taking the union-bound over all steps $s$ in $[r_0, t]$, we have that\[
\Pro{\left. \bigcap_{s \in [r_0, t]} \left(\mathcal{A}^s \cup \neg \mathcal{R}^s \right)\,\,\right|\,\, \mathcal{Z}^{t - T_r}, \mathfrak{F}^{r_0}, \Gamma_2^{r_0} \leq cn} \geq 1 - n^{-\frac{4}{3}\kappa} \cdot T_r \geq 1 - \frac{1}{2}n^{-\kappa}.
\]
When $\cap_{s \in [r_0, t]} \left( \mathcal{A}^s \cup \neg \mathcal{R}^s \right)$ holds, all steps $u$ in all red phases satisfy $\Gamma_2^u \leq 3cn$. Thus, since steps in green phases are good by definition, we have that
\[
\Pro{\left. \bigcap_{s \in [r_0, t]} \left\{ \Gamma_2^s \leq 3cn \right\} \,\,\right|\,\, \mathcal{Z}^{t - T_r}, \mathfrak{F}^{r_0}, \Gamma_2^{r_0} \leq cn} \geq 1 - n^{-\frac{4}{3} \kappa} \cdot T_r \geq 1 - \frac{1}{2} n^{-\kappa}.
\]

Hence, by defining the stopping time $\rho := \inf\{ r_0 \geq t - T_r  : \Gamma_2^{r_0} \leq cn \}$, we have
\begin{align*}
\Pro{\Gamma_2^{t} \leq 3cn}
 & \geq \sum_{r_0 = t - T_r}^t \Pro{\left.\Gamma_2^{t} \leq 3cn \, \right| \, \mathcal{Z}^{t - T_r}, \rho = r_0} \cdot \Pro{\rho = r_0 \,\left|\, \mathcal{Z}^{t - T_r} \right.} \\
 & \geq \sum_{r_0 = t - T_r}^t \Pro{\left.\bigcap_{s \in [r_0, t]} \left\{ \Gamma_2^s \leq 3cn \right\} \, \right| \, \mathcal{Z}^{t - T_r}, \mathfrak{F}^{r_0}, \Gamma_2^{r_0} \leq cn} \cdot \Pro{\rho = r_0 \,\left|\, \mathcal{Z}^{t - T_r} \right.} \\
 & \geq \left( 1 - \frac{1}{2}n^{-\kappa}\right) \cdot \sum_{r_0 = t - T_r}^t \Pro{\rho = r_0 \,\left|\, \mathcal{Z}^{t - T_r} \right.} \\
 & = \left( 1 - \frac{1}{2}n^{-\kappa}\right) \cdot \Pro{\rho \leq t} \\
 & \stackrel{(a)}{\geq} \left( 1 - \frac{1}{2}n^{-\kappa}\right) \cdot \left(1 - 2n^{-2\kappa}\right) \geq 1 - n^{-\kappa},
\end{align*}
using \cref{eq:base_case_starting_point} in $(a)$ and so, the claim follows.\end{proof}

\begin{figure}[t]
    \centering
\begin{tikzpicture}[
  IntersectionPoint/.style={circle, draw=black, very thick, fill=black!35!white, inner sep=0.05cm}
]

\definecolor{MyBlue}{HTML}{9DC3E6}
\definecolor{MyYellow}{HTML}{FFE699}
\definecolor{MyGreen}{HTML}{E2F0D9}
\definecolor{MyRed}{HTML}{FF9F9F}
\definecolor{MyDarkRed}{HTML}{C00000}

\node[anchor=south west] (plt) at (-0.1, 0) {\includegraphics[width=15.40cm]{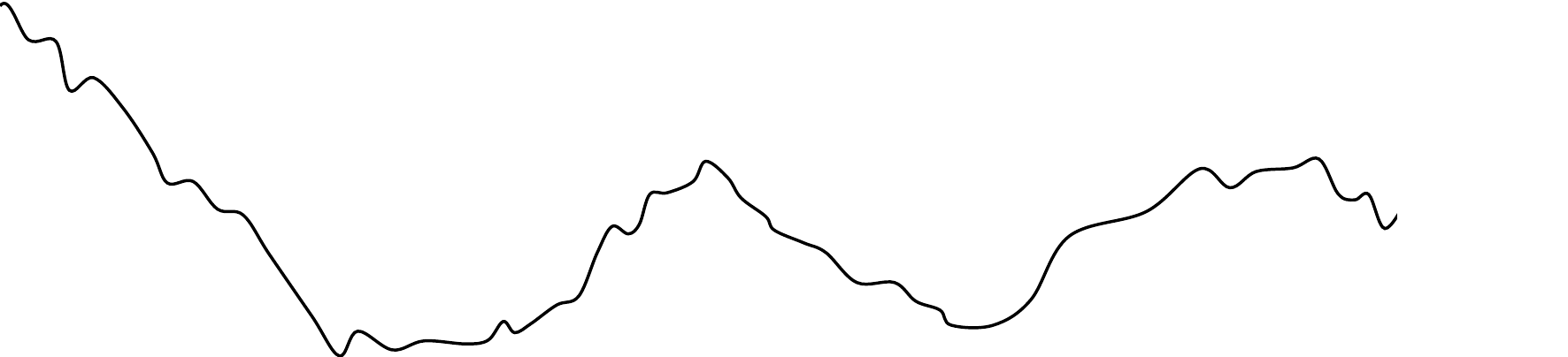}};

\def\xEnd{14}
\def\xLast{13.82}
\def\yLast{4.3}
\def\cn{1}
\def\yBottom{-0.8}

\draw[dashed, thick] (0,1) -- (\xLast, 1);
\draw[dashed, thick] (0, 1.8) -- (\xLast, 1.8);

\node[anchor=east] at (0, \cn) {$cn$};
\node[anchor=east] at (0, 1.8) {$2cn$};
\node[anchor=east] at (0, 3.7) {$\Gamma_2^r$};

\node[anchor=west] at (\xEnd, \yBottom) {$r$};

\def\tA{2.75}
\def\tB{5.82}
\def\tC{8.3}
\def\tD{10.3}
\def\tM{13.51}
\def\tE{14.36}

\newcommand{\drawLine}[3]{
\draw[dashed, very thick, #3] (#1, \yBottom) -- (#1, \yLast);
\draw[very thick] (#1, \yBottom) -- (#1, \yBottom -0.2);
\node[anchor=north] at (#1, \yBottom -0.3) {#2};}

\newcommand{\drawPoint}[3]{
\drawLine{#1}{#2}{#3}
\node[IntersectionPoint] at (#1, \cn) {};}

\draw[very thick] (0, \yBottom) -- (0, \yBottom -0.2);
\node[anchor=north] at (0, \yBottom -0.3) {$m - T_r$};

\drawPoint{\tA}{$s_0$}{black!30!white};
\drawPoint{\tB}{$\tau_1$}{black!30!white};
\drawPoint{\tC}{$s_1$}{black!30!white};
\drawPoint{\tD}{$\tau_2$}{black!30!white}; 
\drawLine{\tM}{\textcolor{MyDarkRed}{$m$}}{MyDarkRed};

\newcommand{\drawRegionRect}[3]{
\node[anchor=south west, rectangle, fill=#3, minimum width=#2 cm- #1 cm, minimum height=0.3cm] at (#1, \yBottom) {};}

\drawRegionRect{\tA}{\tB}{MyGreen}
\drawRegionRect{\tB}{\tC}{MyRed}
\drawRegionRect{\tC}{\tD}{MyGreen}
\drawRegionRect{\tD}{13.70}{MyRed} 

\newcommand{\drawBrace}[4]{
\draw [
    thick,
    decoration={
        brace,
        raise=0.5cm,
        amplitude=0.2cm
    },
    decorate
] (#2, \yBottom - 0.3) -- (#1, \yBottom - 0.3) 
node [anchor=north,yshift=-0.7cm,#4] {#3}; }

\drawBrace{0}{\tA}{Recovery by Lemma 5.9}{pos=0.5};

\draw[->, ultra thick] (0,\yBottom) -- (0, 4.5);
\draw[->, ultra thick] (0,\yBottom) -- (\xEnd, \yBottom);

\end{tikzpicture}
    \caption{Green phases indicate steps where $\Gamma_2^{r}$ is small and red phases indicate steps for which the potential is large and drops (in expectation). In \cref{lem:gamma_1_2_stabilization}, we show that $\Gamma_2 \leq 3cn$  at every point within a red phase using a concentration inequality.}
    \label{fig:gamma_2_small_large_regions}
\end{figure}

Now, we complete the proof of \cref{thm:gamma_concentration}.

\begin{proof}[Proof of \cref{thm:gamma_concentration}]
Consider the auxiliary process $\tilde{\mathcal{P}}_{t - T_r}$ and let $\Gamma_{2, \tilde{\mathcal{P}}}$ be its $\Gamma_2$ potential. Then, by \cref{lem:auxiliary_process_concentration} we have that\begin{align} \label{eq:aux_process_gamma_2_concentration}
\Pro{\left. \Gamma_{2, \tilde{\mathcal{P}}}^t \leq 3cn \,\right|\, \mathcal{Z}^{t - T_r}} \geq 1 - n^{-\kappa}.
\end{align}
By \cref{cor:gamma_1_2_expected_change} and Markov's inequality, since $\tilde{\mathcal{P}}_{t - T_r}$ and $\mathcal{P}$ agree for every step $s \leq t - T_r$, we have that \begin{align}\label{eq:aux_process_z_event_holds}
\Pro{\mathcal{Z}^{t - T_r}} = \Pro{\Gamma_1^{t - T_r} \leq \frac{1}{2}cn^{2\kappa + 1}} \geq 1 - 2n^{-2\kappa}.
\end{align}
Hence, by combining \cref{eq:aux_process_gamma_2_concentration} and \cref{eq:aux_process_z_event_holds}, we have that\[
\Pro{\Gamma_{2, \tilde{\mathcal{P}}}^t \leq 3cn} \geq 1 - n^{-\kappa}.
\]
As mentioned in \cref{eq:modified_process_agrees_with_p_gamma_concentration}, \Whp~the process $\tilde{\mathcal{P}}$ agrees with $\mathcal{P}$ in $[t - T_r, t]$, and hence\[
\Pro{\Gamma_2^t \leq 3cn} 
  \geq \Pro{ \left\{ \Gamma_{2, \tilde{\mathcal{P}}}^t \leq 3cn \right\} \cap \bigcap_{s \in [t - T_r, t]} \left\{ y^s = y_{\tilde{\mathcal{P}}}^s\right\} }
  \geq 1 - n^{-2\kappa} - \frac{1}{2} n^{-2\kappa} \cdot T_r \geq 1 - n^{-\kappa},
\]
using that $T_r = \Oh(n \cdot n^{1/6} \cdot \log n)$.
\end{proof}

\section{Proof of the Layered Induction Step} \label{sec:layered_induction}

To keep the proof self-contained, some definitions and explanations from \cref{sec:outline_layered} are repeated. Our goal is to prove the following theorem.

{\renewcommand{\thethm}{\ref{thm:caching_log_log_n}}
	\begin{thm}[Restated]
\CachingLogLogN
	\end{thm} }
	\addtocounter{thm}{-1}
 
In what follows all potentials are defined in relation to the \Memory process with an $(a, b)$-biased sampling distribution, for constants $a, b \geq 1$. Furthermore all results in this section hold under this assumption.
\paragraph{Full Potentials:} We will be using layered induction over super-exponential potential functions, similar to the one used in \cite[Section 6]{LS22Queries} and \cite[Section 6]{LS22Noise}, but with some differences (see discussion on page~\pageref{sec:difference_to_previous_approaches}). We now define the super-exponential potential functions for $1 \leq j \leq j_{\max} - 1$,
\[
\Psi_j^t := \sum_{i = 1}^n \Psi_{j, i}^t := \sum_{i = 1}^n e^{\alpha_1 \cdot v^j \cdot (y_i^t - z_j)^+},
\]
and
\[
\Phi_j^t := \sum_{i = 1}^n \Phi_{j, i}^t := \sum_{i = 1}^n e^{\alpha_2 \cdot v^j \cdot (y_i^t - z_j)^+},
\]
where $x^+=\max\{x,0\}$ and throughout the remainder of this paper we set \begin{equation}\label{eq:fixingconsts} z_j := \frac{5v}{\alpha_2} \cdot j,\quad  v := \max\{\log (2Cb), 36b\}, \quad  C := \max\{ 6c, 6 \}, \quad j_{\max} = \log_v (\frac{\alpha_2}{2v}\log n),\end{equation} and $\alpha_1, \alpha_2, c > 0$, where $\alpha_1 = 6 \cdot 14 \cdot \alpha_2$, are as defined in Theorems \ref{thm:gamma_concentration} and \ref{simp} from \cref{sec:base_case}. Our aim will be to prove that $\Phi_{j_{\max}-1}^t = \Oh(n)$, which will imply that \[\max_{i \in [n]} y_i^t \leq z_{j_{\max}-1} + \frac{4v^2}{\alpha_2^2} = \Oh(\log \log n).\]

\Folded
\label{def:k_j}
 
The folded process has a structure consisting of (semi) independent blocks of a fixed length, this helps with the analysis when proving the upper bound in the the layered induction. The flexibility in choosing which bin is allocated to will allow us to show that the \Memory process is an instance of the folded process for some series of allocation choices. See \cref{fig:folded_process_2} for an illustration of the folded process. 
\begin{lem}\label{lem:memoryfold}
\Memory is an instance of the folded process defined above.
\end{lem}
\begin{proof}
Recall that any step $t \geq 1$ of the original \Memory step belongs to a unique round $r \geq 0$ and substep $s \geq 1$.

In Case $A$ (i.e., $s=1$), \Memory samples a bin $i=i(r)$ and allocates either to the bin $b$ or $i=i(r)$, i.e., $\ell \in \{ i, b \}$ and thus $y_{\ell}^{r,0} \leq y_{i}^{r,0}$, as needed.

For Case $B$, we have $s \geq 1$ and substep $s$ can only be part of a phase in round $r$, 
if in the previous (or current) phase we sampled a light bin, i.e,
there is a substep $u \in [s-2 \cdot \frac{v}{\alpha_2},s]$ and bin sample $i=i(r,u)$  such that
\[
 y_{i}^{r,u} < z_{j-1} + \frac{2v}{\alpha_2}.
\]
Hence, as we store the least loaded of the sampled bin and cached bin in each round, we have access to a bin with load at most that of $i$ (plus however many balls have been places since sampling $i$) in each substep of the phase. Thus,
\[
 y_{\ell}^{r,s} \leq y_{i}^{r,s} \leq y_{i}^{r,u} + s-u \leq z_{j-1} + \frac{2v}{\alpha_2} + 2 \cdot \frac{v}{\alpha_2} = z_{j-1} + \frac{4v}{\alpha_2}. \qedhere
\] 
\end{proof}

\begin{figure}
\FoldedFigure
 \label{fig:folded_process_2}
\end{figure}

\paragraph{Partial Potentials:} For the \textit{recovery phase}, i.e., showing that \Whp~at some \textit{step} $s$ in an interval of $n \cdot \polylog(n)$ length we have $\Phi_j^s \leq Cn$ (\cref{lem:exists_s_st_phi_linear_whp}), we will need faster drop rates for the potentials, so we will be using the following ``partial'' potentials defined only over the heavy bins
\[
\dot{\Psi}_j^t := \sum_{i = 1}^n \dot{\Psi}_{j, i}^t := \sum_{i : y_i^t \geq z_j} e^{\alpha_1 \cdot v^j \cdot (y_i^t - z_j)},
\quad \text{ and } \quad
\dot{\Phi}_j^t := \sum_{i = 1}^n \dot{\Phi}_{j, i}^t := \sum_{i : y_i^t \geq z_j} e^{\alpha_2 \cdot v^j \cdot (y_i^t - z_j)},
\]
with parameters $\alpha_1, \alpha_2, z_j > 0$ as defined above. In contrast to $\Phi_j$ (and $\Psi_j$) which always have value at least $n$ (since each bin contributes at least $1$), $\dot{\Phi}_j$ (and $\dot{\Psi}_j$) could be as small as $0$. Also, $\dot{\Phi}_j^t \leq \Phi_j^t \leq \dot{\Phi}_j^t + n$ (and $\dot{\Psi}_j^t \leq \Psi_j^t \leq \dot{\Psi}_j^t + n$).

\paragraph{Potentials over Rounds:} We also define versions of the $\Phi_j$ and $\Psi_j$ potentials indexed by a round $r \geq 0$ (note the starting step of the first round may not be equal to $0$):
\[
\overline{\Psi}_j^r := \sum_{i = 1}^n e^{\alpha_1 \cdot v^j \cdot (y_i^{r, 0} - z_j)^+},
\quad \text{ and } \quad
\overline{\Phi}_j^r := \sum_{i = 1}^n e^{\alpha_2 \cdot v^j \cdot (y_i^{r, 0} - z_j)^+}.
\]
Similarly, we define the partial potential functions over rounds
\[
\ddot{\Psi}_j^r := \sum_{i : y_i^t \geq z_j} e^{\alpha_1 \cdot v^j \cdot (y_i^{r, 0} - z_j)},
\quad \text{ and } \quad
\ddot{\Phi}_j^r := \sum_{i : y_i^t \geq z_j} e^{\alpha_2 \cdot v^j \cdot (y_i^{r, 0} - z_j)}.
\]

\paragraph{Differences to Previous Applications:}\label{sec:difference_to_previous_approaches} These potentials are similar in form to the ones used in \cite[Section 6]{LS22Queries} for $\Theta(\log \log n)$ quantiles and in \cite[Section 6]{LS22Noise} for $\Theta(1)$ additive noise. However, the analysis is different as the potentials drop in expectation only when considering a sufficiently large number of steps (e.g., the folded version of the process). For example, starting from a state where the cache has load at least $z_j + 1$, the potential $\Phi_j$ will increase in expectation over one step. Considering rounds consisting of several balls, introduces several challenges:
\begin{itemize}
  \item \textbf{Issue 1:} In each round we could allocate as many as $k_j$ balls, which could be $\Omega(n^{\eps})$. This would mean, that starting from a round $r_0$ with $\Phi_j^{r_0} = \Oh(e^{\frac{1}{2} \log^3 n})$ and having a drop inequality similar to that in \cite{LS22Queries,LS22Noise}, e.g., \[
  \Ex{\left. \overline{\Phi}_j^{r+1} \,\right|\, \mathfrak{F}^r, \overline{\Phi}_{j-1}^r \leq Cn} \leq \overline{\Phi}_j^r \cdot \Big(1 - \frac{1}{n} \Big) + 2,
  \]
  we may need $\Omega(n \log^3 n)$ rounds to prove that the potential becomes $\Oh(n)$ in expectation. In these rounds, we could allocate $\Omega(n^{1 + \eps} \cdot \log^3 n)$ balls ($k_j \cdot \frac{v}{\alpha}$ in each round) and so the length of the interval of the entire analysis would need to be $\omega(n \cdot \polylog(n))$. However, it would not be possible to tolerate a $\poly(n)$ probability decrease in each layer of the layered induction (\cref{lem:new_inductive_step}), as we have $j_{\max} = \Theta(\log \log n)$ layers.
  
  \textbf{Solution:} Define the potential function $\ddot{\Phi}_j$ over just the bins with normalized load at least $z_j$. For this potential function, we can show that:
  \[
  \Ex{\left. \ddot{\Phi}_j^{r+1} \,\right|\, \mathfrak{F}^r, \overline{\Phi}_{j-1}^r \leq 2Cn} \leq \ddot{\Phi}_j^r \cdot \Big(1 - \frac{e^{v^{j+1}}}{n} \Big) + e^{-v^j}.
  \]
  This means that starting from a round $r_0$ with $\Phi_j^{r_0} = \Oh(e^{\frac{1}{2} \log^3 n})$, we need to wait only for $n \cdot e^{-v^{j+1}} \cdot \log^3 n$ rounds, so at most $n \cdot \log^6 n$ steps (which is at most $\mathcal{O}(n \cdot \polylog(n))$), for the potential to become $\Oh(n)$.
  
  \item \textbf{Issue 2:} Unfortunately, for stabilization, i.e., showing that \Whp~$\Phi_j^s = \Oh(n)$ for $n \cdot \polylog(n)$ steps, we cannot use just the partial potential function $\ddot{\Phi}_j$. The reason is that in a single round, the potential could change by $\Omega(n)$, so the bounded difference inequality cannot be applied. 
  
  Consider the case where there are $n \cdot e^{-v^j}$ bins (for $j = 1$), whose load is $z_j + \frac{k_j}{n}$. Then in a single round, we could allocate $k_j$ balls only in light bins, so that the potential becomes $0$. Hence, the potential decreases by $n \cdot e^{-v^j} \cdot e^{\alpha_2 k_j / n} = \Omega( n )$, for $j = 1$. This means that we can no longer apply the concentration inequality, as the bounded difference condition is not strong enough.
  
  \textbf{Solution:} For this part of the analysis, we use the full potential $\overline{\Phi}_j$ and a stopping time to guarantee that the number of balls allocated at every application of the concentration inequality is at most $n/\log^2 n$. This allows us to apply the smoothness argument (\cref{clm:phi_j_does_not_drop_quickly}) to argue that the potential is $\Oh(n)$ in every step of the interval.
\end{itemize}

\subsection{Proof Outline of Lemma~\ref{lem:new_inductive_step}.}

We will now give a summary of the main technical steps in the proof of \cref{lem:new_inductive_step} (an illustration of the key steps is shown in \cref{fig:proof_outline}). 

First, fix any $1 \leq j \leq j_{\max} - 1$. Then the induction hypothesis ensures that $\Phi_{j-1}^{u} = \Oh(n)$ for all steps $u \in [\beta_{j-1}, t + n \log^8 n]$, where $\beta_{j-1} := t + 2(j-1)n \log^6 n$. 

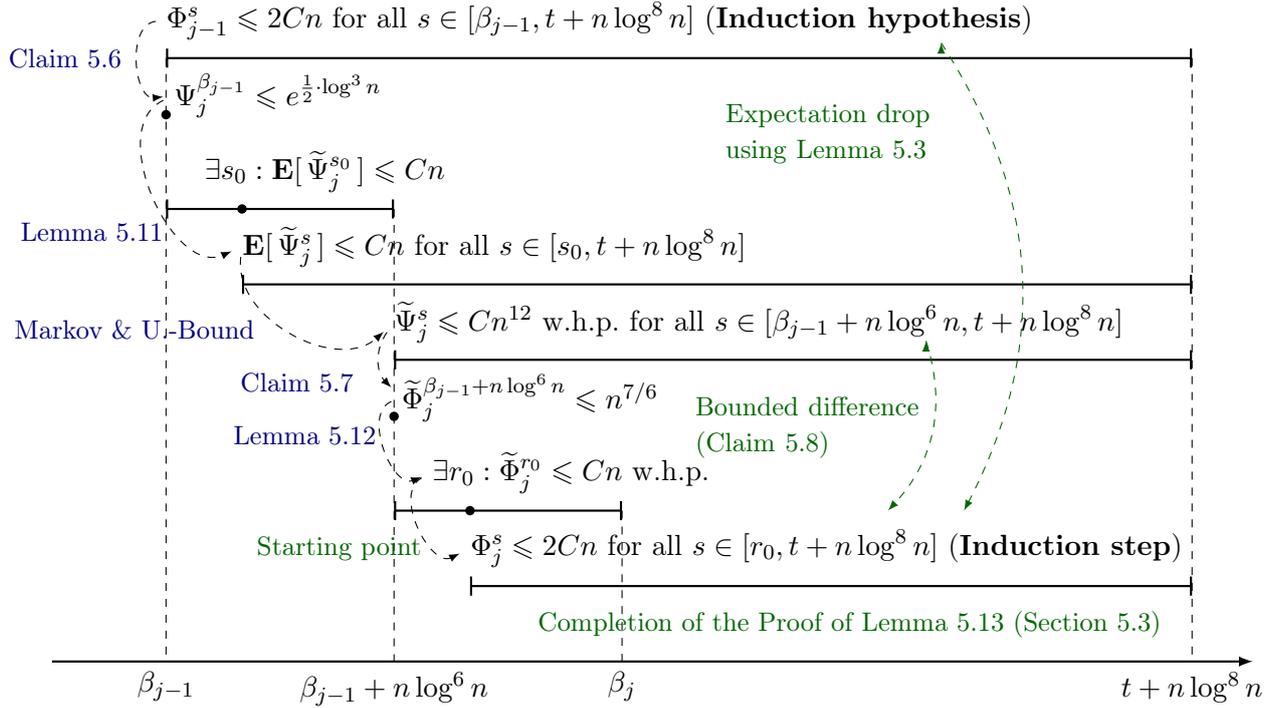
\begin{figure}[H]
\makebox[\textwidth][c]{
\begin{tikzpicture}[
  txt/.style={anchor=west,inner sep=0pt},
  Dot/.style={circle,fill,inner sep=1.25pt},
  implies/.style={-latex,dashed}]

\def\betaO{0}
\def\betaA{1}
\def\betaB{4}
\def\betaC{7}
\def\End{14.5}
\def\yO{-8}

\draw[dashed] (\betaA, 0) -- (\betaA, \yO);
\node[anchor=north] at (\betaA, \yO) {$\beta_{j-1}$};
\draw[dashed] (\betaB, -2) -- (\betaB, \yO);
\node[anchor=north] at (\betaB, \yO) {$\beta_{j-1} + n \log^6 n$};
\draw[dashed] (\betaC, -6) -- (\betaC, \yO);
\node[anchor=north] at (\betaC, \yO) {$\beta_j$};
\draw[dashed] (\End, 0) -- (\End, \yO);
\node[anchor=north] at (\End, \yO) {$t + n \log^8 n$};

\node (indstep) at (\betaA,0.5) {};
\node[txt] at (indstep) {$\Phi_{j-1}^{s} \leq 2Cn$ for all $s \in [\beta_{j-1}, t + n \log^8 n]$ (\textbf{Induction hypothesis})};
\draw[|-|, thick] (\betaA, 0) -- (\End, 0) ;

\node (PsiSmall) at (\betaA+0.1,-0.5) {};
\node[txt] at (PsiSmall) {$\Psi_{j}^{\beta_{j-1}} \leq e^{\frac{1}{2} \cdot \log^3 n}$};
\node[Dot] at (\betaA, -0.75){};

\node (ExistsPsiLinear) at (\betaA+0.5,-1.5) {};
\node[txt] at (ExistsPsiLinear) {$\exists s_0 : \ex{\tilde{\Psi}_{j}^{s_0}} \leq Cn$};
\node[Dot] at (\betaA+1, -2){};
\draw[|-|, thick] (\betaA, -2) -- (\betaB, -2);

\node (ExPsiLinear) at (\betaA + 1, -2.5) {};
\node[txt] at (\betaA + 1, -2.5) {$\ex{\tilde{\Psi}_j^{s}} \leq Cn$ for all $s \in [s_0, t + n \log^8 n]$};
\draw[|-|, thick] (\betaA+1, -3) -- (\End, -3);

\node (PsiPoly) at (\betaB, -3.5) {};
\node[txt] at (PsiPoly) {$\tilde{\Psi}_j^{s} \leq Cn^{12}$ \Whp~for all $s \in [\beta_{j-1}+n \log^6 n, t + n \log^8 n]$};
\draw[|-|, thick] (\betaB, -4) -- (\End, -4);

\node (PhiPoly) at (\betaB+0.1,-4.5) {};
\node[txt] at (\betaB+0.1,-4.5) {$\tilde{\Phi}_{j}^{\beta_{j-1} + n \log^6 n} \leq n^{7/6}$};
\node[Dot] at (\betaB, -4.75){};

\node (ExistsPhiLinear) at (\betaB+0.5,-5.5) {};
\node[txt] at (ExistsPhiLinear) {$\exists r_0 : \tilde{\Phi}_{j}^{r_0} \leq Cn$ \Whp};
\node[Dot] at (\betaB+1, -6){};
\draw[|-|, thick] (\betaB, -6) -- (\betaC, -6);

\node (PhiLinear) at (\betaB+1,-6.5) {};
\node[txt] at (\betaB+1,-6.5) {$\Phi_{j}^{s} \leq 2Cn$ for all $s \in [r_0, t + n \log^8 n]$ (\textbf{Induction step})};
\draw[|-|, thick] (\betaB+1, -7) -- (\End, -7);

\draw[-latex, thick] (\betaO - 0.5, \yO) -- (\End + .8, \yO);

\draw[implies] (indstep) edge[bend right=90] (PsiSmall);
\node[anchor=east, black!50!blue] at (\betaA - 0.45, -0) {\small \cref{clm:phi_small_implies_psi_plus_one_small}};
 
\draw[implies] (PsiSmall) edge[bend right=90] (ExPsiLinear);
\node[anchor=east, black!50!blue] at (\betaA + 0.1, -2.3) {\small \cref{lem:newexists_s_st_ex_psi_linear}};

\draw[implies] (ExPsiLinear) edge[bend right=70] (PsiPoly);
\node[anchor=east, black!50!blue] at (\betaA + 1.3, -3.6) {\small Markov \& U.-Bound};

\draw[implies] (PsiPoly) edge[bend right=40] (PhiPoly);
\node[anchor=east, black!50!blue] at (\betaB - 0.4, -4.3) {\small \cref{clm:psi_potential_poly_implies_phi_small}};

\draw[implies] (PhiPoly) edge[bend right=90] (ExistsPhiLinear);
\node[anchor=east, black!50!blue] at (\betaB - 0.1, -5) {\small \cref{lem:exists_s_st_phi_linear_whp}};

\draw[implies] (ExistsPhiLinear) edge[bend right=90] (PhiLinear);
\node[anchor=east, black!60!green] at (\betaB +0.5, -6.5) {\small Starting point};

\draw[implies,black!60!green] (\betaC + 4, -3.75) edge[bend left=30] (\betaC + 3.5, -6);
\node[text width=4cm, anchor=east,black!60!green] at (\betaC + 5.1, -4.9) {\small Bounded difference (\cref{clm:psi_potential_poly_implies_bounded_diff})};

\draw[implies,black!60!green] (\betaC + 4.2, 0.2) edge[bend left=30] (\betaC + 4.5, -6);
\node[anchor=east,black!60!green, text width=3cm] at (\betaC + 4.5, -1) {\small Expectation drop \\ using \cref{lem:rec_inequality_phi_psi}};

\node[black!60!green] at (\betaC + 3, -7.5) {\small  
Completion of the Proof of \cref{lem:new_inductive_step} (\cref{sec:proof_of_key})};

\end{tikzpicture}}
\caption{Outline for the proof of \cref{lem:new_inductive_step}. Results in blue are given in \cref{sec:auxiliary}, while results in green are used in the application of the concentration inequality (\cref{thm:simplified_chung_lu_theorem_8_5})~in \cref{sec:proof_of_key}. }
\label{fig:proof_outline}
\end{figure}

\paragraph{Recovery.} When $\Phi_{j-1}^{\beta_{j-1}} = \Oh(n)$, it follows by a simple estimate that  $\Psi_j^{\beta_{j-1}} \leq e^{\frac{1}{2} \cdot \log^3 n}$ (\cref{clm:phi_small_implies_psi_plus_one_small}). Using a drop inequality for the partial potential function $\ddot{\Psi}_j$ (\cref{lem:rec_inequality_phi_psi}), it follows that $\ex{\Psi_j^{u}} \leq Cn$, for any step $u \in [\beta_{j-1} + n \log^6 n, t + n \log^8 n]$ (\cref{lem:newexists_s_st_ex_psi_linear}).
By simply using Markov's inequality and a union bound, we can deduce that $\Psi_j^{u} \leq Cn^{12}$ for all steps $u \in [\beta_{j-1} + n \log^6 n, t + n \log^8 n]$. By a simple relation between two potentials, this implies $\ddot{\Phi}_j^u \leq \Phi_j^{u} \leq n^{7/6}$ (\cref{clm:psi_potential_poly_implies_phi_small}). Now using a drop inequality for the partial potential function $\ddot{\Phi}_j$ (\cref{lem:rec_inequality_phi_psi}), guarantees that \Whp~$\Phi_j^{u} \leq Cn$ for \emph{some single} step $u \in [\beta_{j-1}, \beta_j]$ (\cref{lem:exists_s_st_phi_linear_whp}). 

\paragraph{Stabilization.} To obtain the stronger statement which holds \emph{for all} steps $s \in [\beta_{i-1},\beta_j]$, we will use a concentration inequality. The key point is that for any round $r$ where $\overline{\Psi}_j^{r} \leq Cn^{12}$ the absolute difference $|\overline{\Phi}_j^{r+1} - \overline{\Phi}_j^{r}|$ is at most $n^{1/3}$, because $14 \frac{\alpha_2}{\alpha_1} < 1/3$ (\cref{clm:psi_potential_poly_implies_bounded_diff}). This is crucial so that applying the supermartingale concentration bound~\cref{thm:simplified_chung_lu_theorem_8_5}~from~\cite{CL06} to $\overline{\Phi}_j$ yields, that an $\Phi_j^s = \Oh(n)$ guarantee for all steps $s \in [\beta_j, t + n \log^8 n]$ using a smoothing argument (\cref{clm:phi_j_does_not_drop_quickly}).

As the proof of \cref{lem:new_inductive_step} requires several claims and lemmas, the remainder of this section is divided further in:
\begin{enumerate}
\item Analysis of the (expected) drop of partial and full potentials. (\cref{sec:expected_drop})
\item (Deterministic) inequalities that involve one or two potentials. (\cref{sec:deterministic_claims})
\item Auxiliary (probabilistic) lemmas based on these drop results. (\cref{sec:probabilistic_drop})
\end{enumerate}
After that, we proceed to complete the proof of \cref{lem:new_inductive_step} in \cref{sec:proof_of_key}. Finally, in \cref{sec:proof_of_caching_log_logn} we complete the proof of the main theorem \cref{thm:caching_log_log_n}.

\subsection{Preparation for the Proof of Lemma~\ref{lem:new_inductive_step}}\label{sec:auxiliary}
We define the following key event $\mathcal{K}_{j-1}^s$ for any step $s$ and $1 \leq j \leq j_{\max} - 1$, which will be used frequently in the analysis, as
\[
\mathcal{K}_{j-1}^s := \{ \Phi_{j-1}^s \leq Cn \}.
\]

In the following arguments we will be working frequently with the ``killed'' versions of the potentials, i.e., we condition on $\mathcal{K}_{j-1}^{\rho}$ holding in all steps $\rho \in [\beta_{j-1}, s]$: \[
\tilde{\Phi}_{j}^{s} := \Phi_{j}^{s} \cdot \mathbf{1}_{ \cap_{\rho \in [\beta_{j-1}, s]} \mathcal{K}_{j-1}^{\rho} } \text{ and } \tilde{\Psi}_{j}^{s} := \Psi_{j}^{s} \cdot \mathbf{1}_{ \cap_{\rho \in [\beta_{j-1}, s]} \mathcal{K}_{j-1}^{\rho}}. 
\]

Also, let $T_s(r)$ denote the step corresponding to round $r$ starting from step $s$.

\subsubsection{Analysis of the Drop of the Partial and Full Potentials} \label{sec:expected_drop}
 
Now we will show that when $\overline{\Phi}_{j-1}^r \leq Cn$ holds (implied by $\mathcal{K}_{j-1}^{T(r)}$), we obtain a bound on the number of heavy bins under the \Memory process on an $(a,b)$-biased sampling distribution:

\begin{lem} \label{lem:quantile_bound}
For any $1 \leq j \leq j_{\max} - 1$ and any round $r \geq 0$, if $\overline{\Phi}_{j-1}^r \leq Cn$, then we also have that the number of bins $i$ with $y_i^{r, 0} \geq z_{j-1} + \frac{2v}{\alpha_2}$ is at most
\[
\frac{1}{2b} n \cdot e^{-v^j},
\]
and in any substep $s$ of round $r$, the number of bins $i$  with $y_i^{r, s} \geq z_{j-1} + \frac{2v}{\alpha_2}$ is at most
\[
\frac{1}{b} n \cdot e^{-v^j}.
\]
\end{lem}
\begin{proof}
Recall that $z_j = \frac{5v}{\alpha_2} \cdot j$. The contribution of any bin $i \in [n]$ with load $y_i^{r, 0} \geq z_{j-1} + \frac{2v}{\alpha_2}$ to $\overline{\Phi}_{j-1, i}^r$ is upper bounded by
\begin{align*}
\overline{\Phi}_{j-1, i}^r
 & \geq e^{\alpha_2 \cdot v^{j-1} \cdot (z_{j-1} + \frac{2v}{\alpha_2} - z_{j-1})^+} \\
 & = e^{\alpha_2 \cdot v^{j-1} \cdot \frac{2v}{\alpha_2}} \\
 & = e^{2v^j}.
\end{align*}
Hence, if $\big\{ \overline{\Phi}_{j-1}^{r} \leq Cn \big\}$ holds, then the number of bins $i \in [n]$ with load $y_i^{r, 0} \geq z_{j-1} + \frac{2v}{\alpha_2}$  is at most
\[
 Cn \cdot  e^{-2v^j} \leq \frac{1}{2b} n \cdot e^{-v^j},
\]
using that $e^{-v^j} \geq e^{-v}$ for $j \geq 1$ and $v \geq \log(2Cb)$.

Note that in a single round we allocate to at most $k_j \cdot \frac{v}{\alpha_2}$ bins. So, the number of bins with $y_i^{r, s} \geq z_{j-1} + \frac{2v}{\alpha_2}$ in any substep $s$ in round $r$ can grow to at most
\[
\frac{1}{2b} n \cdot e^{-v^j} + k_j \cdot \frac{2v}{\alpha_2} \leq \frac{1}{b} n \cdot e^{-v^j},
\]
for sufficiently large $n$, using that $e^{-v^j} \geq e^{-v^{j_{\max}}} \geq e^{-\frac{\alpha_2}{2v} \log n} \geq n^{-1/2}$ (since $\alpha_2 \leq 1$ and $v \geq 1$) and $k_j = e^{v^{j+1}} \cdot \log^3 n \leq e^{v^{j_{\max}}} \cdot \log^3 n \leq e^{\frac{\alpha_2}{2v} \log n} \cdot \log^3 n \leq n^{1/7}$ (since $\alpha_2 \leq 4$ and $v \geq 36$).
\end{proof}
	These bounds from \cref{lem:quantile_bound} on the number of heavy bins under allow us to prove a potential drop across layers in the induction.  
\begin{lem} \label{lem:rec_inequality_phi_psi} For any $1 \leq j \leq j_{\max} - 1$, and any round $r \geq 0$, we have that
\[
\Ex{\left. \ddot{\Phi}_j^{r+1} \,\right|\, \mathfrak{F}^r, \mathcal{K}_{j-1}^{r}  } \leq \ddot{\Phi}_j^{r} \cdot \Big(1 - \frac{e^{v^{j+1}}}{n} \Big) + e^{-v^j/2},
\]
and
\[
\Ex{\left. \ddot{\Psi}_j^{r+1} \,\right|\, \mathfrak{F}^r, \mathcal{K}_{j-1}^{r}  } \leq \ddot{\Psi}_j^{r} \cdot \Big(1 - \frac{e^{v^{j+1}}}{n} \Big) + e^{-v^j/2}.
\]
\end{lem}

\begin{proof}
Using \cref{lem:quantile_bound}, we have that the number of bins $\ell \in [n]$ with load $y_{\ell}^{r, s} \geq z_{j-1} + \frac{2v}{\alpha_2}$ for all substeps $s$ of round $r$ is at most $\frac{1}{b} n \cdot e^{-v^j}$.

We will now analyze the expected contribution of a bin $i \in [n]$ to the potential $\ddot{\Phi}_j^{r+1}$.

\noindent\textbf{Case 1 [$y_i^{r, 0} > z_j $]:} When we sample a \textit{heavy} bin $\ell \in [n]$ with $y_{\ell}^{r, s} \geq z_{j-1} + \frac{2v}{\alpha_2}$ (happens with probability at most $\frac{b}{n} \cdot \frac{n}{b} \cdot e^{-v^j} = e^{-v^j}$) then the average increases by $1/n$. Otherwise, we sample a \textit{light} bin, and a new phase begins. So, since bin $i$ is heavy we have that
\begin{align} \label{eq:drop_case_1}
&\Ex{\left. \ddot{\Phi}_{j, i}^{r+1} \, \right| \, \mathfrak{F}^r, \mathcal{K}_{j-1}^r} \notag \\
 & \quad\leq \ddot{\Phi}_{j, i}^r \cdot \Big( \underbrace{\frac{b}{n} \cdot e^{\alpha_2 v^j \cdot (1 - 1/n)}}_{\text{Bin $i$ is chosen}} + \underbrace{\Big(e^{-v^j} - \frac{b}{n}\Big) \cdot e^{-\alpha_2 v^j / n}}_{\text{Another heavy bin is chosen}} + \underbrace{(1 - e^{-v^j}) \cdot \sum_{k = 1}^{k_j} A_k}_{\text{A light bin is chosen}} \Big),
\end{align}
where $A_k$ is the expected contribution to the potential if exactly $k$ phases are executed. Each subsequent phase starts with probability at least $1 - (e^{-v^j})^{v/\alpha_2}$. Hence, we have for $p = e^{-v^{j+1}/\alpha_2} \leq e^{-v^{j+1}}$ (since $\alpha_2 \leq 1$), by domination
\[
\sum_{k = 1}^{k_j} A_k \leq \sum_{k = 1}^{k_j} \underbrace{ \frac{(1 - p)^{k-1} \cdot p}{1 - (1 - p)^{k_j}}}_{\text{Probability of exactly $k$ phases}} \cdot \underbrace{e^{-kv^{j+1}/n}}_{\substack{\text{Change of $
\ddot{\Phi}_{j, i}^r$ due}\\\text{to change in average}}},
\] 
using that after $k$ phases the average load changes by $\frac{kv}{n\alpha_2}$ and that during a phase we don't allocate to heavy bins, i.e., we won't allocate to bin $i$ (as it remains heavy).  Plugging in $p = e^{-v^{j+1}}$, we obtain
\begin{equation}\label{eq:prebddOnA_k}
\sum_{k = 1}^{k_j} A_k 
  \leq e^{-v^{j+1}/n} \cdot e^{-v^{j+1}} \cdot  \frac{1}{1 - (1 - e^{-v^{j+1}})^{k_j}}  \cdot \sum_{k = 0}^{k_j - 1} \Big( (1 - e^{-v^{j+1}}) \cdot e^{-v^{j+1}/n}\Big)^k. \end{equation}
We first bound the sum on the right hand side of \eqref{eq:prebddOnA_k}, as
  \begin{align}\label{eq:sumredux}\sum_{k = 0}^{k_j - 1} \Big( (1 - e^{-v^{j+1}}) \cdot e^{-v^{j+1}/n}\Big)^k& \leq \frac{1}{1 - (1 - e^{-v^{j+1}}) \cdot e^{-v^{j+1}/n}}\notag \\
  	& \stackrel{(a)}{\leq}  \frac{1}{1 - (1 - e^{-v^{j+1}}) \cdot (1 -\frac{1}{2} \cdot \frac{v^{j+1}}{n})} \notag\\
  	&=\frac{1}{e^{-v^{j+1}}} \cdot \frac{1}{1 -\frac{1}{2} \cdot \frac{v^{j+1}}{n} +\frac{1}{2} \cdot \frac{v^{j+1}}{n}\cdot e^{v^{j+1}}}\notag \\
  	 & \stackrel{(b)}{\leq}  \frac{1}{e^{-v^{j+1}}} \cdot \frac{1}{1 + \frac{1}{4} \cdot \frac{v^{j+1}}{n} \cdot e^{v^{j+1}}}\notag \\
  	& \stackrel{(c)}{\leq}  \frac{1}{e^{-v^{j+1}}} \cdot \Big(1 - \frac{1}{8} \cdot \frac{v^{j+1}}{n} \cdot e^{v^{j+1}}\Big),
  	\end{align}using in $(a)$ that $e^z \geq 1 + \frac{1}{2} z$ for $|z| < 1$ and that $v^{j+1} \leq v^{j_{\max}} \leq \log n = o(n)$ (since $j + 1 \leq j_{\max} = \log_v (\frac{\alpha_2}{2v}\log n)$), in $(b)$ that $e^{v^{j+1}} \geq 2$ since $v^{j+1} \geq v \geq 36$, in $(c)$ that $\frac{1}{1 + z} \leq 1 - \frac{z}{2}$ for $0 < z \leq 1$ (\cref{lem:small_z_frac_ineq}) and $\frac{1}{4} \cdot \frac{v^{j+1}}{n} \cdot e^{v^{j+1}} \leq \frac{1}{4} \cdot \log n \cdot n^{-1/2}$. 
  
  Applying the bound from \eqref{eq:sumredux} to  \eqref{eq:prebddOnA_k} gives 
 \begin{align}\label{eq:bddOnA_k}
 \sum_{k = 1}^{k_j} A_k & \leq e^{-v^{j+1}/n} \cdot \frac{1}{1 - (1 - e^{-v^{j+1}})^{k_j}} \cdot \Big(1 - \frac{1}{8} \cdot \frac{v^{j+1}}{n} \cdot e^{v^{j+1}}\Big)\notag \\
 & \stackrel{(a)}{\leq} e^{-\alpha_2 v^{j}/n} \cdot \frac{1}{1 - (1 - e^{-v^{j+1}})^{k_j}} \cdot \Big(1 - \frac{1}{8} \cdot \frac{v^{j+1}}{n} \cdot e^{v^{j+1}}\Big) \notag \\
 & \leq e^{-\alpha_2 v^{j}/n} \cdot \frac{1}{1 - \exp\left(-e^{-v^{j+1}} \cdot k_j\right)} \cdot \Big(1 - \frac{1}{8} \cdot \frac{v^{j+1}}{n} \cdot e^{v^{j+1}}\Big) \notag \\ 
 & \stackrel{(b)}{\leq} e^{-\alpha_2 v^{j}/n} \cdot \left( 1 +  2\exp\left(-e^{-v^{j+1}} \cdot k_j\right)\right) \cdot \Big(1 - \frac{1}{8} \cdot \frac{v^{j+1}}{n} \cdot e^{v^{j+1}}\Big) \notag \\ 
 & = e^{-\alpha_2 v^{j}/n} \cdot \left( 1 +  2\exp\left(-\log^3 n \right)\right) \cdot \Big(1 - \frac{1}{8} \cdot \frac{v^{j+1}}{n} \cdot e^{v^{j+1}}\Big) \notag \\
 & \stackrel{(c)}{\leq} e^{-\alpha_2 v^{j}/n} \cdot \Big(1 - \frac{1}{9} \cdot \frac{v^{j+1}}{n} \cdot e^{v^{j+1}}\Big),
\end{align}
using in $(a)$ that $\alpha_2 \leq 1$ and $v \geq 1$, in $(b)$ that $\frac{1}{1-\eps} \leq 1 + 2\eps$ for $0 < \eps \leq 1/2$ by \cref{lem:small_z_frac_ineq_2} and in $(c)$ that $2\exp\left(-\log^3 n \right) = o(n^{-2})$, while $\frac{v^{j+1}}{n} \cdot e^{v^{j+1}} \geq \frac{1}{n}$.

Returning to \cref{eq:drop_case_1}, we have 
\begin{align*}
\Ex{\left. \ddot{\Phi}_{j, i}^{r+1} \, \right| \, \mathfrak{F}^r}
 & \leq \ddot{\Phi}_{j, i}^r \cdot e^{-\alpha_2 v^j / n} \cdot \Big( \frac{b}{n} \cdot e^{\alpha_2 v^j} + \Big(e^{-v^j} - \frac{b}{n}\Big) + (1 - e^{-v^j}) \cdot e^{\alpha_2 v^j / n} \cdot \sum_{k = 1}^{\infty} A_k \Big) \\
 & = \ddot{\Phi}_{j, i}^r \cdot e^{-\alpha_2 v^j / n} \cdot \Big( 1 + \frac{b}{n} \cdot (e^{\alpha_2 v^j} - 1) + (1 - e^{-v^j}) \cdot \Big(e^{\alpha_2 v^j / n} \cdot \sum_{k = 1}^{\infty} A_k - 1\Big) \Big) \\& \stackrel{\eqref{eq:bddOnA_k}}{\leq}  \ddot{\Phi}_{j, i}^r \cdot e^{-\alpha_2 v^j / n} \cdot \Big( 1 + \frac{b}{n} \cdot e^{\alpha_2 v^j} - (1 - e^{-v^j}) \cdot \Big(\frac{1}{9} \cdot \frac{v^{j+1}}{n} \cdot e^{v^{j+1}}\Big) \Big) \\
 & \stackrel{(a)}{\leq} \ddot{\Phi}_{j, i}^r \cdot e^{-\alpha_2 v^j / n} \cdot \Big( 1 + \frac{b}{n} \cdot e^{\alpha_2 v^j} - \frac{2b}{n} \cdot e^{v^{j+1}} \Big) \\
 & \stackrel{(b)}{\leq} \ddot{\Phi}_{j, i}^r \cdot \Big( 1 -   \frac{e^{v^{j+1}}}{n} \Big), 
\end{align*}
using in $(a)$ that $1 - e^{-v^j} \geq 1 - e^{-v} \geq 1/2$ and $v^j\geq v \geq 36b$, and in $(b)$  that $\alpha_2 \leq 1$ and $b \geq 1$. 

\noindent\textbf{Case 2 [$y_i^{r, 0} \in (z_j - 1 + \frac{1}{n}, z_j]$]:} For this case, we will argue about the aggregate contribution of all such bins $i$. In round $r$, we can allocate at most one ball to $i$, since $y_i^{r, 0} > z_{j-1} + \frac{2v}{\alpha_2}$ (which means that we can only allocate to this bin in the first substep). In this case the potential $\ddot{\Phi}_{j, i}^{r}$ of said bin raises from $0$ to at most $e^{\alpha_2 v^j}$. This can only occur if we pick a heavy bin $\ell \in [n]$ with load $y_{\ell}^r \geq z_{j-1} + \frac{2v}{\alpha_2}$, thus with probability at most $e^{-v^j}$. So, on aggregate we have that
\begin{align*}
\sum_{\ell : y_{\ell}^{r, 0} \in (z_j - 1 + \frac{1}{n}, z_j]} \Ex{\left. \ddot{\Phi}_{j, \ell}^{r+1} \, \right| \, \mathfrak{F}^r} 
& \leq e^{\alpha_2 v^j} \cdot   e^{-v^{j}} \\
&  \stackrel{(a)}{\leq} e^{-v^j/2}  \\
&  \stackrel{(b)}{=} \sum_{\ell : y_{\ell}^{r, 0} \in (z_j - 1 + \frac{1}{n}, z_j]} \ddot{\Phi}_{j, \ell}^{r} \cdot \Big( 1 - \frac{e^{v^{j+1}}}{n} \Big) + e^{-v^j/2},
\end{align*}
where inequality $(a)$ holds since $\alpha_2 \leq 1/2$ and $(b)$ as $\ddot{\Phi}_{j, \ell}^{r} = 0$ for each $\ell$ with $y_{\ell}^{r, 0} \in (z_j - 1 + \frac{1}{n}, z_j]$.

\noindent\textbf{Case 3(a) [$y_i^{r, 0} \leq z_{j-1} + \frac{2v}{\alpha_2}$]:} In round $r$, we could allocate multiple balls to bin $i$, however its load will always remain at most $z_{j-1}+ \frac{4v}{\alpha_2} < z_j $ and so $\ddot{\Phi}_{j, i}^{r+1} = \ddot{\Phi}_{j, i}^r = 0$.

\noindent\textbf{Case 3(b) [$y_i^{r, 0} \in (z_{j-1} + \frac{2v}{\alpha_2}, z_j - 1 + \frac{1}{n}]$]:} In round $t$, we can allocate at most one ball to $i$, since $y_i^{r, 0} > z_{j-1} + \frac{2v}{\alpha_2}$ (which means that we can only allocate to this bin in the first step); so as in Case 3(a), $\ddot{\Phi}_{j, i}^{r+1} = \ddot{\Phi}_{j, i}^r = 0$. 

Hence, by aggregating over the three cases, we get the conclusion for $\ddot{\Phi}_j$. The same analysis also works for $\ddot{\Psi}_j$, as we only used that $\alpha_2 \leq 1/2$ (which also holds for $\alpha_1$) and $\alpha_1 \geq \alpha_2$.
\end{proof}

This implies that when the potentials are large, we get a multiplicative expected drop.

\begin{cor} \label{cor:large_rec_inequality_phi_psi} For any $1 \leq j \leq j_{\max} - 1$, and any round $r \geq 0$, we have that
\[
\Ex{\left. \ddot{\Psi}_j^{r+1} \,\right|\, \mathfrak{F}^r, \ddot{\Psi}_j^{r} \geq 2n } \leq \ddot{\Psi}_j^{r} \cdot \Big(1 - \frac{e^{v^{j+1}}}{2n} \Big),
\]
and
\[
\Ex{\left. \ddot{\Phi}_j^{r+1} \,\right|\, \mathfrak{F}^r, \ddot{\Phi}_j^r \geq 2n } \leq \ddot{\Phi}_j^{r} \cdot \Big(1 - \frac{e^{v^{j+1}}}{2n} \Big).
\]
\end{cor}
\begin{proof}
Using \cref{lem:rec_inequality_phi_psi}, 
\begin{align*}
\Ex{\left. \ddot{\Psi}_j^{r+1} \,\right|\, \mathfrak{F}^r, \ddot{\Psi}_j^{r} \geq 2n } 
 & \leq \ddot{\Psi}_j^{r} \cdot \Big(1 - \frac{e^{v^{j+1}}}{n} \Big) + e^{-v^j/2} \\
 & = \ddot{\Psi}_j^{r} \cdot \Big(1 - \frac{e^{v^{j+1}}}{2n} \Big) - \ddot{\Psi}_j^{r} \cdot \frac{e^{v^{j+1}}}{2n} + e^{-v^j/2} \\
 & \leq \ddot{\Psi}_j^{r} \cdot \Big(1 - \frac{e^{v^{j+1}}}{2n} \Big) - e^{v^{j+1}} + e^{-v^j/2} \\
 & \leq \ddot{\Psi}_j^{r} \cdot \Big(1 - \frac{e^{v^{j+1}}}{2n} \Big).
\end{align*}
Working in exactly the same way for $\ddot{\Phi}_j$, we obtain the second claim.
\end{proof}

For the stabilization phase in the layered induction step (\cref{lem:new_inductive_step}), we make use of the following drop inequalities for the full $\overline{\Phi}_j$ potential:

\begin{lem} \label{lem:rec_inequality_partial_phi_psi} For any $1 \leq j \leq j_{\max} - 1$, and any round $r \geq 0$, we have that
\[
\Ex{\left. \overline{\Phi}_j^{r+1} \,\right|\, \mathfrak{F}^r, \mathcal{K}_{j-1}^{r}} \leq \overline{\Phi}_j^{r} \cdot \Big(1 - \frac{1}{n} \Big) + 2,
\]
and
\[
\Ex{\left. \overline{\Phi}_j^{r+1} \,\right|\, \mathfrak{F}^r, \mathcal{K}_{j-1}^{r}, \overline{\Phi}_j^{r} \geq 4n} \leq \overline{\Phi}_j^{r} \cdot \Big(1 - \frac{1}{2n} \Big).
\]
\end{lem}
\begin{proof}
For the first statement, the analysis proceeds similarly to \cref{lem:rec_inequality_phi_psi}, with Case 1 being unchanged:

\noindent\textbf{Case 2 [$y_i^{r, 0} \in (z_j - 1 + \frac{1}{n}, z_j]$]:} For this case, we will argue again about the aggregate contribution of all such bins $i$. In round $r$, we can allocate at most one ball to $i$, since $y_i^{r, 0} > z_{j-1} + \frac{2v}{\alpha_2}$ (which means that we can only allocate to this bin in the first step). In this case the potential $\overline{\Phi}_{j, i}^{r}$ of said bin raises from $0$ to at most $e^{\alpha_2 v^j}$. This can only occur if we pick a bin $\ell \in [n]$ with load $y_{\ell}^r \geq z_{j-1} + \frac{2v}{\alpha_2}$, thus with probability at most $e^{- v^j}$. Thus we have  \begin{align*}
\sum_{i : y_i^{r, 0} \in (z_j - 1 + \frac{1}{n}, z_j]} \Ex{\left. \overline{\Phi}_{j, i}^{r+1} \, \right| \, \mathfrak{F}^r} 
& \leq \sum_{i : y_i^{r, 0} \in (z_j - 1 + \frac{1}{n}, z_j]} \overline{\Phi}_{j, i}^{r} + e^{\alpha_2 v^j} \cdot   e^{-v^{j}} \\
&  \stackrel{(a)}{\leq}  \sum_{i : y_i^{r, 0} \in (z_j - 1 + \frac{1}{n}, z_j]} \overline{\Phi}_{j, i}^{r} + 1  \\
&  \stackrel{(b)}{=} \sum_{i : y_i^{r, 0} \in (z_j - 1 + \frac{1}{n}, z_j]}\left( \overline{\Phi}_{j, i}^{r} \cdot \Big( 1 - \frac{1}{n} \Big) + \frac{1}{n} \right)+ 1,
\end{align*}
where inequality $(a)$ holds since $\alpha_2 \leq 1/2$ and $(b)$ as $\overline{\Phi}_{j, i}^{r} = 1$ for each $i$ with $y_i^{r, 0} \in (z_j - 1 + \frac{1}{n}, z_j]$.

\noindent\textbf{Case 3(a) [$y_i^{r, 0} \leq z_{j-1} + \frac{2v}{\alpha_2}$]:} In round $r$, we could allocate multiple balls to bin $i$, however its load will always remain at most $z_{j-1}+ \frac{4v}{\alpha_2} < z_j $ and so $\overline{\Phi}_{j, i}^{r+1} = \overline{\Phi}_{j, i}^r = 1$. Thus,
\[
\Ex{\left. \overline{\Phi}_{j, i}^{r+1} \, \right| \, \mathfrak{F}^r} 
= \overline{\Phi}_{j, i}^r \cdot \Big(1 - \frac{1}{n}\Big) + \frac{1}{n}.
\]

\noindent\textbf{Case 3(b) [$y_i^{r, 0} \in (z_{j-1} + \frac{2v}{\alpha_2}, z_j - 1 + \frac{1}{n}]$]:} In round $r$, we can allocate at most one ball to $i$, since $y_i^{r, 0} > z_{j-1} + \frac{2v}{\alpha_2}$ (which means that we can only allocate to this bin in the first step); so as in Case 3(a), $\overline{\Phi}_{j, i}^{r+1} = \overline{\Phi}_{j, i}^r = 1$. 
Thus,
\[
\Ex{\left. \overline{\Phi}_{j, i}^{r+1} \, \right| \, \mathfrak{F}^r} 
= \overline{\Phi}_{j, i}^r \cdot \Big(1 - \frac{1}{n}\Big) + \frac{1}{n}.
\]
Hence, by aggregating over the three cases, we get the conclusion for $\overline{\Phi}_j$. 

For the second statement,\[
\Ex{\left. \overline{\Phi}_j^{r+1} \,\right|\, \mathfrak{F}^r, \mathcal{K}_{j-1}^{r}, \overline{\Phi}_j^{r} \geq 4n} \leq \overline{\Phi}_j^{r} \cdot \Big(1 - \frac{1}{2n} \Big) - \overline{\Phi}_j^{r} \cdot \frac{1}{2n} + 2 \leq \overline{\Phi}_j^{r} \cdot \Big(1 - \frac{1}{2n} \Big). \qedhere
\]
\end{proof}

\subsubsection{Deterministic Relations between the Potential Functions}\label{sec:deterministic_claims}

We collect several basic facts about the potential functions $\Phi_j^{s}$ and $\Psi_j^{s}$ related to the \Memory process on an $(a,b)$-biased sampling distribution with parameters fixed by \eqref{eq:fixingconsts}.

\begin{clm} \label{clm:phi_small_implies_psi_plus_one_small} For any $1 \leq j \leq j_{\max} - 1$ and for any step $s \geq 0$,
\[
\Phi_{j-1}^{s} \leq 2Cn ~~ \Rightarrow ~~ \Psi_j^{s} \leq \exp\Big(\frac{1}{2} \cdot \log^3 n\Big).
\]
\end{clm}
\begin{proof}
Assuming $\Phi_{j-1}^{s} \leq 2Cn$ implies that for any bin $i \in [n]$, 
\begin{align*}
\Phi_{j-1, i}^s = \exp \Big(\alpha_2\cdot v^{j-1} \cdot \Big(y_i^s - \frac{5v}{\alpha_2} \cdot (j-1) \Big)^+ \Big) \leq 2 Cn.
\end{align*}
By re-arranging we get 
\begin{align*}
y_i^{s} 
 & \leq \frac{\log(2Cn)}{\alpha_2} \cdot v^{-(j-1)} + \frac{5v}{\alpha_2} \cdot (j-1) \leq \log^{1.5} n,
\end{align*}
for sufficiently large $n$, since $j \leq j_{\max} = \Oh(\log \log n)$. 

Hence, we can upper bound $\Psi_j^s$ by
\begin{align*}
\Psi_{j}^{s} 
 & \leq n \cdot \exp\Big(\alpha_1 \cdot v^{j} \cdot \log^{1.5} n \Big) \leq \exp\Big( \frac{1}{2} \cdot \log^3 n\Big),
\end{align*}
for sufficiently large $n$, using that $v^j \leq v^{j_{\max}} \leq \log n$. 
\end{proof}

The next claim bounds $\Phi_j$ given that $\Psi_j = \poly(n)$.

\begin{clm} \label{clm:psi_potential_poly_implies_phi_small}
For any $1 \leq j \leq j_{\max} - 1$ and for any step $s \geq 0$, if $\Psi_j^{s} \leq Cn^{12}$, then for any bin $i \in [n]$, we have that $\Phi_{j, i}^{s} \leq n^{1/6}$ and so by aggregating over all bins $\Phi_j^{s} \leq n^{7/6}$.
\end{clm}
\begin{proof}
We will begin by showing that when $\Psi_j^{s} \leq Cn^{12}$ we have that for any bin $i \in [n]$,
\[
y_i^{s} \leq z_j + \frac{14}{\alpha_1} \cdot (\log n) \cdot v^{-j}.
\]
Assuming the contrary, i.e., that for some bin $i$, $y_i^{s} > z_j + \frac{14}{\alpha_1} \cdot (\log n) \cdot v^{-j}$, then we get $\Psi_j^{s} > \exp(\alpha_1 \cdot \frac{14}{\alpha_1} \cdot \log n ) = n^{14}$, which is a contradiction.

Next, we turn to upper bounding the contribution of any bin $i \in [n]$,
\begin{align}
\Phi_{j,i}^{s}  
 & = \exp \bigl( \alpha_2 \cdot v^{j} \cdot \bigl( y_i^s - z_j  \bigr)^{+} \bigr) \notag \\
 & \leq \exp\Big( \alpha_2 \cdot v^j \cdot \frac{14}{\alpha_1} \cdot (\log n) \cdot v^{-j} \Big) \notag \\
 &= \exp\Big( \frac{14 \cdot \alpha_2}{\alpha_1} \cdot \log n  \Big) = n^{1/6}, \label{eq:n6_upper_bound}
\end{align}
since $\alpha_1 = 6 \cdot 14 \cdot \alpha_2$. Hence, by aggregating over all bins, 
\[
\Phi_j^s \leq n \cdot n^{1/6} = n^{7/6}. \qedhere
\]
\end{proof}

The next claim is crucial for applying the concentration inequality in \cref{lem:new_inductive_step}, since the third statement bounds the maximum additive change of $\overline{\Phi}^{r}$ (assuming $\overline{\Psi}^{r} = \poly(n)$):
\begin{clm} \label{clm:psi_potential_poly_implies_bounded_diff}
For any $1 \leq j \leq j_{\max} - 1$ and for any round $r \geq 0$, if  $\overline{\Psi}_j^{r} \leq Cn^{12}$, then $| \overline{\Phi}_j^{r+1} - \overline{\Phi}_j^{r} | \leq n^{1/3}$.
\end{clm}
\begin{proof} 
We will start by obtaining lower and upper bounds for $\overline{\Phi}_j^{r+1}$. By \cref{clm:psi_potential_poly_implies_phi_small}, in any round $r$ with $\Psi_j^r \leq Cn^{12}$, we have that $\Phi_{j, i}^r \leq n^{1/6}$ for each bin. Now, for the upper bound, note that in any round $r+1$ we can allocate to at most one heavy bin. Let $i \in [n]$ be that heavy bin, then
\[
\overline{\Phi}_j^{r+1} \leq \overline{\Phi}_j^{r} + \overline{\Phi}_{j, i}^r \cdot \exp\bigl( \alpha_2 \cdot v^j \bigr) \leq \overline{\Phi}_j^{r} + n^{1/6} \cdot n^{1/6} = \overline{\Phi}_j^{r} + n^{1/3},
\]
using that $v^j \leq v^{j_{\max}} \leq \log n$, $\alpha_2 \leq 1/6$ by \eqref{eq:fixingconsts} and applying \cref{eq:n6_upper_bound}. For the lower bound, we pessimistically assume that all bin loads decrease by $1/n$ in each step of the round $r+1$. So, since there are at most $k_j \cdot \frac{v}{\alpha_2}$ steps
\[
\overline{\Phi}_j^{r+1} \geq \overline{\Phi}_j^{r} \cdot \exp\Big( - \frac{v^{j+1}}{n} \cdot k_j \Big) \geq \overline{\Phi}_j^{r} \cdot \Big(1 - \frac{v^{j+1}}{n} \cdot k_j \Big) \geq \Phi_j^{r} - \frac{n \cdot n^{1/6}}{n} \cdot (\log n) \cdot n^{1/7} \geq \overline{\Phi}_j^{r} - n^{1/3},
\]
using that $e^x \geq 1 + x$ (for any $x$), $v^{j+1} \leq v^{j_{\max}} \leq \log n$, $\overline{\Phi}_j^r \leq n \cdot n^{1/6}$ by \cref{clm:psi_potential_poly_implies_phi_small} for $s = T(r)$ and $k_j \leq n^{1/7}$. Combining the two bounds we get the statement.
\end{proof}

The next claim is a simple ``smoothness'' argument showing that the potential cannot decrease quickly within $n/\log^2 n$ steps. The derivation is elementary and relies on the fact that the average load does not change by more than $1/\log^2 n$ within these steps.
\begin{clm} \label{clm:phi_j_does_not_drop_quickly}
For any $1 \leq j \leq j_{\max}-1$, any step $s \geq 0$ and any step $u \in [s, s + n/\log^2 n]$, we have $\Phi_j^{u} \geq 0.99 \cdot \Phi_j^{s}$.
\end{clm}
\begin{proof}
The normalized load after $u - s$ steps can decrease by at most $\frac{u - s}{n} \leq \frac{1}{\log^2 n}$. Hence, for any bin $i \in [n]$,
\begin{align*}
\Phi_{j, i}^{u} 
 & = e^{\alpha_2 \cdot v^j \cdot (y_i^{u} - z_j )^{+}}
 \geq e^{\alpha_2 \cdot v^j \cdot (y_i^{s} - \frac{u-s}{n} - z_j )^{+}}
 \geq e^{\alpha_2 \cdot v^j \cdot (y_i^{s} - z_j )^+ - \alpha_2 \cdot v^j \cdot \frac{1}{\log^2 n} }
 = \Phi_{j, i}^{s} \cdot e^{-\frac{\alpha_2 \cdot v^j}{\log^2 n}} \\
 & \geq \Phi_{j, i}^{s} \cdot e^{-o(1)} \geq 0.99 \cdot \Phi_{j, i}^{s},
\end{align*}
for sufficiently large $n$, using that $v^j \leq v^{j_{\max}} \leq \log n$ and $\alpha_2 \leq 1$. By aggregating over all bins, we get the claim.
\end{proof} 

The next claim relates the full and partial potentials defined at the start of Section \ref{sec:layered_induction}.

\begin{clm} \label{lem:potential_within_rounds}
For any round $r \geq 0$ with one or more phases we have that for any substep $s \geq 0$ of round $r$, $(i)$ $\Phi_j^{r, s} \leq \overline{\Phi}_j^r$, $(ii)$ $\Psi_j^{r, s} \leq \overline{\Psi}_j^r$, $(iii)$ $\dot{\Phi}_j^{T(r) + s} \leq \ddot{\Phi}_j^r$ and $(iv)$ $\dot{\Psi}_j^{T(r) + s} \leq \ddot{\Psi}_j^r$.
\end{clm}
\begin{proof}
For all substeps of a round, we allocate to a bin that does not contribute to the potential. Hence, the potential can only decrease (because of the change of the average load).
\end{proof}

\subsubsection{Auxiliary Probabilistic Lemmas on the Potential Functions (Recovery) }\label{sec:probabilistic_drop}
The first lemma proves that $\tilde{\Psi}_{j}^{s}=\Psi_j^{s} \cdot \mathbf{1}_{\cap_{r\in [\beta_{j-1}, s]} \mathcal{K}_{j-1}^{r}}$ is small in expectation for \textit{all} rounds $s \geq \beta_{j-1}+n\log^6 n$.

\begin{lem} \label{lem:newexists_s_st_ex_psi_linear}
For any $1 \leq j \leq j_{\max} - 1$, we have that for any step $s \geq \beta_{j-1} + n \log^6 n$, $\ex{\tilde{\Psi}_j^{s}} \leq Cn$.
\end{lem}
\begin{proof}
By \cref{lem:rec_inequality_phi_psi} for any round $r \geq 0$ starting to count from step $\beta_{j-1}$, \begin{align} \label{eq:ddot_psi_rec_ineq}
\Ex{\left. \ddot{\Psi}_{j}^{r+1} \,\right|\, \mathfrak{F}^r, \mathcal{K}_{j-1}^r  } \leq \ddot{\Psi}_j^{r} \cdot \Big(1 - \frac{e^{v^{j+1}}}{n} \Big) + e^{-v^j/2}.
\end{align}
Now we define the potential function for any round $r$ after step $\beta_{j-1}$,\[
\widehat{\Psi}_j^r := \ddot{\Psi}_j^r \cdot \mathbf{1}_{\cap_{\rho \in [\beta_{j-1}, T_{\beta_{j-1}}(r)]} \mathcal{K}_{j-1}^{\rho}}.
\]
Note that although this potential has a hat is it not related to the \DWeakMemory memory process. Next observe that whenever $\mathbf{1}_{\cap_{\rho \in [\beta_{j-1}, T_{\beta_{j-1}}(r)]} \mathcal{K}_{j-1}^{\rho}} = 0$, it follows deterministically that $\widehat{\Psi}_j^{r+1} = 0$, and hence by \cref{eq:ddot_psi_rec_ineq}, \begin{align} \label{eq:tilde_psi_drop}
\Ex{\left. \widehat{\Psi}_{j}^{r+1} \,\right|\, \mathfrak{F}^r} \leq \widehat{\Psi}_j^{r} \cdot \Big(1 - \frac{e^{v^{j+1}}}{2n} \Big)  + e^{-v^j/2}.
\end{align}
To upper bound $\Ex{\widehat{\Psi}_j^r}$, it suffices to upper bound $\Ex{\left. \widehat{\Psi}_j^r \,\,\right|\, \mathcal{K}_{j-1}^{\beta_{j-1}}}$, since
\begin{align} \label{eq:ex_tilde_psi_conditioned_on_K}
\Ex{\widehat{\Psi}_j^r} = \Ex{\left. \widehat{\Psi}_j^r \,\,\right|\, \mathcal{K}_{j-1}^{\beta_{j-1}}} \cdot \Pro{\mathcal{K}_{j-1}^{\beta_{j-1}}} + 0 \cdot \Pro{\neg \mathcal{K}_{j-1}^{\beta_{j-1}}} \leq \Ex{\left. \widehat{\Psi}_j^r \,\,\right|\, \mathcal{K}_{j-1}^{\beta_{j-1}}}.
\end{align}
When $\mathcal{K}_{j-1}^{\beta_{j-1}}$ holds, by definition $\Phi_{j-1}^{\beta_{j-1}} \leq 2Cn$ holds, and by \cref{clm:phi_small_implies_psi_plus_one_small}, $\widehat{\Psi}_{j}^{\beta_{j-1}} \leq \Psi_{j}^{\beta_{j-1}} \leq e^{\frac{1}{2} \cdot \log^3 n}$. Applying \cref{lem:geometric_arithmetic} to \cref{eq:tilde_psi_drop} (with $a = 1 - \frac{1}{2n} \cdot e^{v^{j+1}}$ and $b = e^{-v^j}$), for any round $r \geq n \cdot e^{-v^{j+1}} \cdot \log^3 n$ where round $r=0$ starts at step $\beta_{j-1}$, we have 
\begin{align*}
\Ex{\left. \widehat{\Psi}_j^r \,\, \right\vert \,  \mathcal{K}_{j-1}^{\beta_{j-1}}}  & \leq 
\Ex{\left. \widehat{\Psi}_j^r \,\, \right\vert \, \mathfrak{F}^{\beta_{j-1}},  \tilde{\Psi}_j^{\beta_{j-1}} \leq e^{\frac{1}{2} \cdot \log^3 n}} \\
 & \leq \Big(1 - \frac{e^{v^{j+1}}}{2n} \Big)^{r} \cdot e^{\frac{1}{2} \cdot \log^3 n} + \frac{2n \cdot e^{-v^j}}{e^{v^{j+1}}} \\
 & \leq \Big(1 - \frac{e^{v^{j+1}}}{2n} \Big)^{n \cdot e^{-v^{j+1}} \cdot \log^3 n} \cdot e^{\frac{1}{2} \cdot \log^3 n} + 2n. \end{align*}Now, using the inequalities $e^x \geq 1 + x$ (for any $x$) and $C \geq 6$ (from \eqref{eq:fixingconsts}), we obtain
 \[
 \Ex{\left. \widehat{\Psi}_j^r \,\, \right\vert \,  \mathcal{K}_{j-1}^{\beta_{j-1}}}  \leq e^{-\frac{1}{2} \cdot \log^3 n} \cdot e^{\frac{1}{2} \cdot \log^3 n} + 2n \leq 1 + 2n \leq \frac{C}{2} n.\] Combining this with \cref{eq:ex_tilde_psi_conditioned_on_K}, for any round $r \geq n \cdot e^{-v^{j+1}} \cdot \log^3 n$,
\[
\Ex{\widehat{\Psi}_j^r} \leq \frac{C}{2} n,
\]
which implies 
\[
\Ex{\ddot{\Psi}_j^r \cdot \mathbf{1}_{\cap_{\rho \in [\beta_{j-1}, T_{\beta_{j-1}}(r)]} \mathcal{K}_{j-1}^{\rho}}} \leq \frac{C}{2} n.
\]
Therefore,
\begin{align*}
\Ex{\tilde{\Psi}_j^{T_{\beta_{j-1}}(r)}} 
 & = \Ex{\overline{\Psi}_j^{r} \cdot \mathbf{1}_{\cap_{\rho \in [\beta_{j-1}, T_{\beta_{j-1}}(r)]} \mathcal{K}_{j-1}^{\rho}}} \\
 & \stackrel{(a)}{\leq} \Ex{\dot{\Psi}_j^{r} \cdot \mathbf{1}_{\cap_{\rho \in [\beta_{j-1}, T_{\beta_{j-1}}(r)]} \mathcal{K}_{j-1}^{\rho}}} + \Ex{n \cdot \mathbf{1}_{\cap_{\rho \in [\beta_{j-1}, T_{\beta_{j-1}}(r)]} \mathcal{K}_{j-1}^{\rho}}} \\
 & \leq \frac{C}{2} n + n ,
\end{align*}
using in $(a)$ that $\overline{\Psi}_j^r \leq \dot{\Psi}_j^r + n$. Since in each round we allocate at most $k_j = e^{v^{j+1}} \cdot \log^3 n$ balls, in  $n \cdot e^{-v^{j+1}} \cdot \log^3 n$ rounds we allocate at most $n \log^6 n$ balls. Hence by \cref{clm:phi_j_does_not_drop_quickly} and \cref{lem:potential_within_rounds}, for any step $s \geq T_{\beta_{j-1}}(r) \geq \beta_{j-1} + n \log^6 n$,
\[
\Ex{\tilde{\Phi}_j^s} \leq \frac{C/2 + 1}{0.99} \cdot n \leq Cn,
\]
using that $C \geq 6$ by \eqref{eq:fixingconsts}.
\end{proof}

We now switch to the other potential function $\tilde{\Phi}_j^{s} = \Phi_j^{s} \cdot \mathbf{1}_{ \cap_{r\in [\beta_{j-1}, s]} \mathcal{K}_{j-1}^{r}}$, and prove that \Whp~it is linear in \emph{at least one} round in $[\beta_{j-1} + n \log^3 n, \beta_j]$.

\begin{lem} \label{lem:exists_s_st_phi_linear_whp}
For any $1 \leq j \leq j_{\max} - 1$ it holds that,
\[
\Pro{\bigcup_{r \in [\beta_{j-1}+ n \log^6 n, \beta_j]} \left\{ \tilde{\Phi}_j^r \leq Cn \right\}} \geq 1 - n^{-5}.
\]
\end{lem}
\begin{proof}
Let $t_0 := \beta_{j-1} + n \log^6 n$. Using \cref{lem:newexists_s_st_ex_psi_linear} and Markov's inequality, we obtain
\[
\Pro{\tilde{\Psi}_j^{t_0} \leq Cn^{12}} \geq 1 - n^{11}.
\]
Using \cref{clm:psi_potential_poly_implies_phi_small} and the fact that $\{ \tilde{\Phi}_j^{t_0} = 0 \}$ holds iff $\{ \tilde{\Psi}_j^{t_0} = 0 \}$, we get that $\{ \tilde{\Psi}_j^{t_0} \leq C n^{12} \}$ implies $\{ \tilde{\Phi}_j^{t_0} \leq n^{7/6} \}$, so
\begin{align} \label{eq:tilde_phi_t0_bound_43}
\Pro{\tilde{\Phi}_j^{t_0} \leq n^{7/6}} \geq 1 - n^{11}.
\end{align}
Assume now that $\{ \tilde{\Phi}_j^{t_0} \leq n^{7/6} \}$ holds. For any round $r \geq  0$ after step $t_0$, we define
\[
\widehat{\Phi}_{j}^r := 
\ddot{\Phi}_j^{r} \cdot \mathbf{1}_{\cap_{\rho \in [\beta_{j-1}, T_{t_0}(r)]} \mathcal{K}_{j-1}^{\rho}} \cdot \mathbf{1}_{\cap_{\rho \in [0, r]} \ddot{\Phi}_j^{\rho} > \frac{C}{2} n}.
\]
By \cref{cor:large_rec_inequality_phi_psi} since $C \geq 6$, for any round $r$ after step $t_0$,
\[ 
\Ex{\left. \ddot{\Phi}_{j}^{r+1} \,\right|\, \mathfrak{F}^{r}, \ddot{\Phi}_{j}^{r} > Cn } 
 \leq \ddot{\Phi}_j^{r} \cdot \Big(1 - \frac{e^{v^{j+1}}}{2n} \Big).
\]
Since whenever $\{ \ddot{\Phi}_{j}^{r} = 0 \}$, it follows deterministically that $\{ \ddot{\Phi}_{j}^{r+1} = 0 \}$, we also have that for any round $r$ after step $t_0$,
\begin{align} \label{eq:phi_i_plus_1_large_phi}
\Ex{\left. \widehat{\Phi}_{j}^{r+1} \,\right|\, \mathfrak{F}^{r}} \leq \widehat{\Phi}_{j}^{r} \cdot \Big(1 - \frac{e^{v^{j+1}}}{2n} \Big).
\end{align}
By inductively applying~\cref{eq:phi_i_plus_1_large_phi} for $\Delta := n \cdot e^{-v^{j+1}} \cdot \log^2 n$ rounds starting at step $t_0$, we have
\begin{align*}
\Ex{\left. \widehat{\Phi}_{j}^{\Delta} \,\,\right|\,\, \mathfrak{F}^{t_0}, \tilde{\Phi}_j^{t_0} \leq n^{7/6} } 
 & \leq \tilde{\Phi}_{j}^{t_0} \cdot \Big(1 - \frac{e^{v^{j+1}}}{2n} \Big)^{T} \\
 & \leq \tilde{\Phi}_{j}^{t_0} \cdot \Big(1 - \frac{e^{v^{j+1}}}{2n} \Big)^{n \cdot e^{-v^{j+1}} \cdot \log^2 n} \\
 & \leq n^{7/6} \cdot e^{- \frac{1}{2} \log^2 n} \leq n^{-6},
\end{align*}
for sufficiently large $n$, using that $e^x \geq 1 + x$ (for any $x$). Hence, by Markov's inequality,
\[
\Pro{\left. \widehat{\Phi}_{j}^{\Delta} \leq 1 \,\right|\, \mathfrak{F}^{t_0}, \tilde{\Phi}_j^{t_0} \leq n^{7/6}} \geq 1 - n^{-6}.
\]
When $\{ \widehat{\Phi}_{j}^{T} \leq 1\}$ holds, either $\ddot{\Phi}_j^{T} \leq 1 \leq Cn$ or for one of the indicators it holds that either
\begin{center}
$\mathbf{1}_{ \cap_{\rho\in [\beta_{j-1}, T_{t_0}(\Delta)]} \mathcal{K}_{j-1}^{\rho} } = 0\quad $ or $\quad \mathbf{1}_{\cap_{\rho \in [0, \Delta]} \ddot{\Phi}_j^{\rho} > \frac{C}{2}n} = 0$.
\end{center}
In either case, these imply that there exists a round $r \in [0, \Delta]$ such that $\tilde{\Phi}_j^{T_{t_0}(r)} \leq Cn$, since when $\ddot{\Phi}_j^{r} \leq \frac{C}{2} n$ then $\tilde{\Phi}_j^{T_{t_0}(r)} \leq \frac{C}{2} n + n \leq Cn$. 

Since in any round we allocate at most $k_j = e^{v^{j+1}} \cdot \log^3 n$ balls, in these $\Delta$ rounds we can allocate at most $n \log^5 n$ balls. So this implies that there exists a step $s \in [\beta_{j-1}+ n \log^6 n, \beta_j]$, such that $\tilde{\Phi}_j^s \leq Cn$. Hence, since $\beta_{j-1} + n \log^6 n + n \log^5 n \leq \beta_j$, we have that
\[
\Pro{\left.\bigcup_{s \in [\beta_{j-1}+ n \log^6 n, \beta_j]} \left\{ \tilde{\Phi}_j^{s} \leq Cn \right\} ~\right|~ \mathfrak{F}^{t_0}, \tilde{\Phi}_j^{t_0} \leq n^{7/6}} \geq 1 - n^{-6}.
\]
Finally, combining with \cref{eq:tilde_phi_t0_bound_43}, we get
\[
\Pro{\bigcup_{s \in [\beta_{j-1}+ n \log^6 n, \beta_j]} \left\{ \tilde{\Phi}_j^{s} \leq Cn \right\}} \geq \left(1-n^{-6} \right) \cdot \left(1 - n^{-11} \right) \geq  1 - n^{-5}. \qedhere
\]
\end{proof}

\subsection{Completing the Proof of Key Lemma \texorpdfstring{(\cref{lem:new_inductive_step})}{}}\label{sec:proof_of_key}

The proof of \cref{lem:new_inductive_step} shares some of the ideas from the proof of \cref{thm:gamma_concentration}. However, there we could more generously take a union bound over the entire time-interval (consisting of $n \cdot \polylog(n)$ steps) to ensure that the potential is indeed small everywhere with high probability. Here we cannot afford to lose a polynomial factor in the error probability, as the induction step has to be applied $j_{\max} = \Theta(\log \log n)$ times. To overcome this, we will partition the time-interval into consecutive intervals of length $n/\log^2 n$. Then, we will prove that at the end of each such interval the potential is small \Whp, and finally use a simple smoothness argument to argue that the potential is small \Whp~in \emph{all} steps.
 
\begin{lem}[\textbf{Induction Step}]\label{lem:new_inductive_step}
\InductionStep 
\end{lem}

\begin{proof}
Consider an arbitrary step $t_0 \in[ \beta_{j-1} + n \log^6 n, \beta_j]$ and recall from Page \pageref{def:k_j} that $k_j$ is the number of phases of the folded process in the $j$-th layer of the induction. Our goal is to prove that $\Phi_j^u \leq 2Cn$ for all $u \in [t_0, t + n \log^8 n]$. We proceed by grouping the steps into at most $q \leq \log^{10} n$ \textit{epochs}, which are ordered counting from step $t_0$. For each $i\in [q]$ the $i$-th epoch lasts for a random number $N_i$ of steps with $N_i \in [n/\log^2 n - k_j \cdot \frac{v}{\alpha_2}, n/\log^2 n]$, except for possibly the last epoch which is not subject to the same lower bound on its length but contains less than $ n/\log^2 n$ steps. In particular, let $r_i$ denote the starting round of the $i$-th epoch, then for any round $r \in [r_i, r_i + \Delta]$ for $\Delta = n/\log^2 n$, we define
 
\[
X_i^{r} := \begin{cases}
X_i^{r-1} & \text{if we have allocated more than $\frac{n}{\log^2 n} - k_j$ balls in rounds $[r_i, r)$,} \\
\overline{\Phi}_j^{r} & \text{else if }\exists \rho \in [r_i, r)$ such that $\overline{\Phi}_j^{\rho} \geq 5n, \\
5n+n^{1/3} & \text{otherwise}.
\end{cases}
\]
This means that the random variable $X_i^{r} $ is stopped before allocating $n/\log^2 n$ balls and so when $\overline{\Phi}^{r_i + \Delta} = \Oh(n)$, by the smoothness argument \cref{clm:phi_j_does_not_drop_quickly}, we have that $\Phi_j^s = \Oh(n)$ for all steps $s$ in the epoch, i.e., $s \in [T_{t_0}(r_i),T_{t_0}(r_{i+1}))$. Note that the third branch condition in the definition of $X_i^r$ can only be satisfied for some rounds in the beginning of the epoch, this is since once the second condition has been activated only the second or first conditions can be satisfied.  

Following the notation of \cref{thm:simplified_chung_lu_theorem_8_5}, we define the bad event $B_i^r$ for any $r \in [r_i, r_i + \Delta]$, as the complement of 
\[
\left( \bigcap_{u \in [T_{t_0}(r_i),T_{t_0}(r))} \left\{ \tilde{\Psi}_j^{u} \leq Cn^{12} \right\} \right)
\bigcap \left( \bigcap_{u \in [T_{t_0}(r_i),T_{t_0}(r))} \mathcal{K}_{j-1}^{u}  \right) \subseteq \bigcap_{u \in [T_{t_0}(r_i),T_{t_0}(r))} \Big\{ \Psi_j^{u} \leq Cn^{12} \Big\}.
\]
We will now bound the probability of this bad event occurring. By \cref{lem:newexists_s_st_ex_psi_linear}, $\ex{\tilde{\Psi}_j^{u}} \leq Cn$ for any step $u \in [\beta_{j-1} +  n \log^6 n, t + n \log^8 n]$. Using Markov's inequality and the union bound over steps $u \in [\beta_{j-1} +  n \log^6 n, t + n \log^8 n]$ it follows that
\begin{align}
 \Pro{\bigcap_{u \in [\beta_{j-1} +  n \log^6 n,t + n \log^8 n]} \left\lbrace \tilde{\Psi}_j^{u} \leq Cn^{12} \right\rbrace } &\geq 1 - n^{-11} \cdot \left(n \log^8 n + (2j - 1) n \log^6 n \right)\notag \\ &\geq 1 - n^{-9}. \label{eq:badevent}
\end{align}
By the hypothesis of this lemma for $j-1$, it holds that
\begin{align} \label{eq:precondition_of_induction_step}
 \Pro{\bigcap_{u \in [\beta_{j-1},t+n \log^8 n]} \mathcal{K}_{j-1}^{u}} \geq 1 - \frac{(\log n)^{11(j-1)}}{n^4}.
\end{align}
Hence, by the union bound over \cref{eq:precondition_of_induction_step} and \cref{eq:badevent} (since $[T_{t_0}(r_i),T_{t_0}(r)) \subseteq [\beta_{j-1} +  n \log^6 n,t + n \log^8 n]$), 
\[
\Pro{ \neg B_i^{r}} \geq 1 - n^{-9} - \frac{(\log n)^{11(j-1)}}{n^4} \geq 1 - \frac{2 (\log n)^{11(j-1)}}{n^4}.
\]

We will use the following claim, to establish the preconditions of \cref{thm:simplified_chung_lu_theorem_8_5} for $X_i^{r}$.

\begin{clm}\label{clm:x_i_preconditions}
Consider the $i$-th epoch for $i \in [q]$. Then for the random variables $X_i^{r}$, for any round $r \in [r_i, r_i + \Delta]$ and any filtration $\mathfrak{F}^{r-1}$, it follows that
\begin{align*}
 (i)  & \quad \Ex{ X_{i}^{r} \, \mid \, \mathfrak{F}^{r-1}, \neg B_i^{r-1}} \leq X_{i}^{r-1}, \\
 (ii) & \quad \left( \left. \left| X_i^{r} - X_i^{r-1} \, \right| \,~ \right| \, \mathfrak{F}^{r-1}, \neg B_i^{r-1}\right) \leq 2 n^{1/3}.
\end{align*}
\end{clm}
\begin{pocd}{clm:x_i_preconditions}
Let $\tau := \inf \{ r \geq r_i \colon T_{t_0}(r) - T_{t_0}(r_i) \geq \frac{n}{\log^2 n} - k_j \}$. Then for any round $r - 1 \geq \tau$, the conditions are trivially satisfied as $X_i^r = X_i^{r-1}$.

For $r-1 < \tau$, recall that by \cref{lem:rec_inequality_partial_phi_psi}, for any round $r \in (r_i, r_i + \Delta]$ with $\overline{\Phi}_j^{r-1} \geq 4n$, \begin{align}
\Ex{\overline{\Phi}_j^{r} \,\left| \, \mathfrak{F}^{r-1}, \mathcal{K}_{j-1}^{r-1}, \overline{\Phi}_j^{r-1} \geq 4n \right.} 
 &\leq \overline{\Phi}_j^{r-1} \cdot \Big(1 - \frac{1}{2n}\Big) - \overline{\Phi}_j^{r-1} \cdot \frac{1}{2n} + 2\notag\\
 &\leq \overline{\Phi}_j^{r-1} \cdot \Big(1 - \frac{1}{2n}\Big), \label{eq:drop}
\end{align}
and consider the following cases:

\noindent \textbf{Case 1 [$\overline{\Phi}_j^{r_i} \geq 5n + n^{1/3}$]:} By \cref{clm:phi_j_does_not_drop_quickly} for all $r \in (r_i, r_i + \Delta]$, we have $\Phi_j^{r-1} \geq 0.99 \cdot (5n + n^{1/3}) \geq 4n$ (as we allocate at most $\Delta \leq n/\log^2 n$ balls in these rounds) and so the first statement follows from \cref{eq:drop} and the second statement since conditional on $\neg B_i^{r-1}$, the precondition of~\cref{clm:psi_potential_poly_implies_bounded_diff} holds.

\noindent \textbf{Case 2:} Otherwise, let $\sigma := \inf \{ r \geq r_i \colon \overline{\Phi}_j^{r} \geq 5n\}$. We consider the following three subcases (see \cref{fig:x_i_preconditions}):
\begin{itemize} \itemsep0pt
    \item \textbf{Case 2(a) [$r-1 < \sigma$]:} Here $X_i^{r} = X_i^{r-1}$ (using that $\overline{\Phi}_j^{r_i} < 5n + n^{1/3}$), so the two statements hold trivially.
    \item \textbf{Case 2(b) [$r-1 = \sigma$]:} We have that $\overline{\Phi}_j^{r-1} \geq 5n$ and $\overline{\Phi}_j^{r-2} < 5n$. 
    By~\cref{clm:psi_potential_poly_implies_bounded_diff} since $\neg B_i^{r-1}$ holds, we obtain that  $\Phi_{j}^{r-1} \leq \overline{\Phi}_{j}^{r-2} + n^{1/3} < 5n + n^{1/3}$. Further, by definition, 
    $X_i^{r-1} = 5n + n^{1/3}$ and $X_i^r = \overline{\Phi}_j^r$, so by \cref{eq:drop},\[\Ex{\left. X_i^{r} \, \right| \, \mathfrak{F}^{r-1}, \neg B_i^{r-1}} = \Ex{\left. \overline{\Phi}_j^{r}  \, \right| \, \mathfrak{F}^{r-1}, \neg B_i^{r-1}}\leq \overline{\Phi}_j^{r-1} < X_i^{r-1},
    \] 
    which establishes the first statement. For the second statement, we have \[
    \left|X_i^{r} - X_i^{r-1}\right| = \left|\overline{\Phi}_j^{r} - 5n - n^{1/3} \right| \leq \left|\overline{\Phi}_j^{r-1} - 5n - n^{1/3} \right| + \left|\overline{\Phi}_j^{r} - \overline{\Phi}_j^{r-1} \right| \leq 2n^{1/3},
    \]
    where in the second inequality we used \cref{clm:psi_potential_poly_implies_bounded_diff}.
    \item \textbf{Case 2(c) [$r-1 > \sigma$]:} Here, $X_i^{r-1} = \overline{\Phi}_j^{r-1}$ and $X_i^r = \overline{\Phi}_j^r$. Since $\overline{\Phi}_j^{\sigma} \geq 5n$, by~\cref{clm:phi_j_does_not_drop_quickly} (as $r-\sigma \leq n/\log^2 n$), we also have that
    \[
     \overline{\Phi}_{j}^{r-1} \geq 0.99 \cdot \overline{\Phi}_{j}^{\sigma} \geq 0.99 \cdot 5n \geq 4n,
    \]
    and thus by \cref{eq:drop}, the first statement follows. The second statement follows, since conditional on $\neg B_i^{r}$, the precondition of~\cref{clm:psi_potential_poly_implies_bounded_diff} holds.\qedhere 
\end{itemize}
\end{pocd}

\begin{figure}
    \centering
    \includegraphics[scale=0.65]{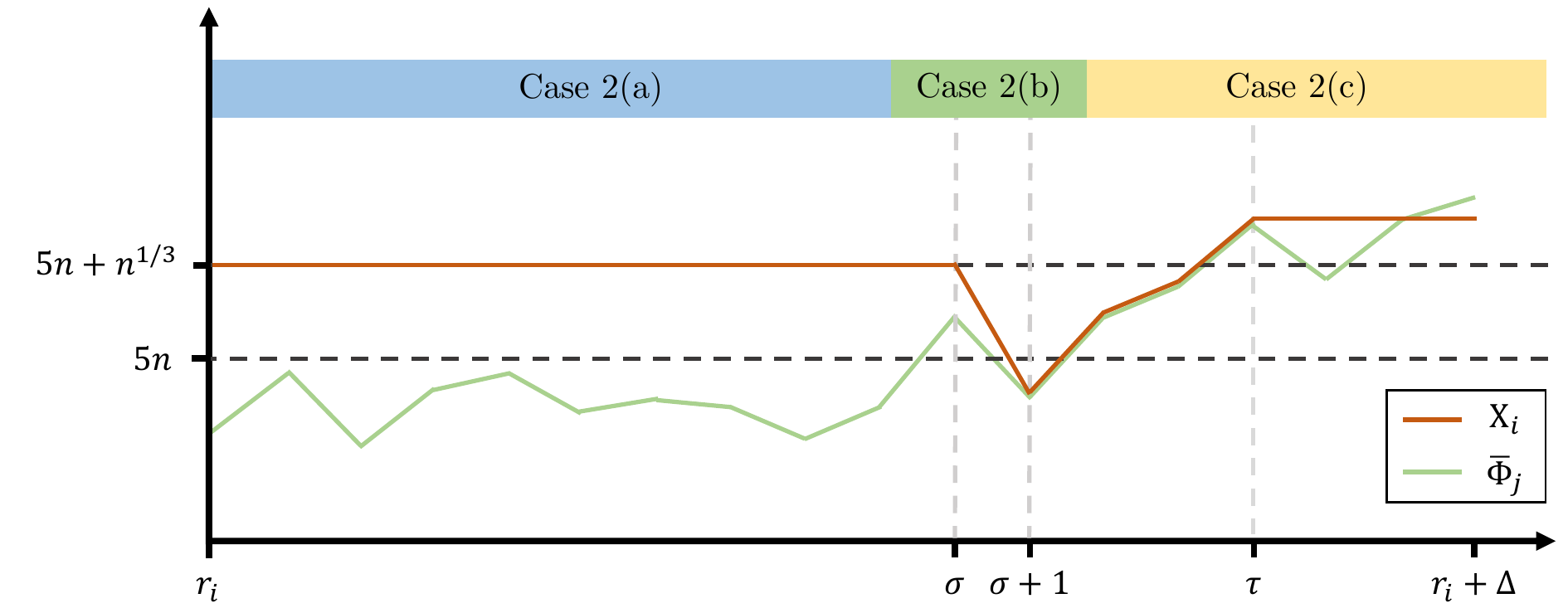}
    \caption{Visualization of the three subcases of Case 2 in the proof of \cref{clm:x_i_preconditions}.}
    \label{fig:x_i_preconditions}
\end{figure}
 
By~\cref{clm:x_i_preconditions}, we have now verified that $X_{i}^{r}$ satisfies the preconditions of \cref{thm:simplified_chung_lu_theorem_8_5} for any filtration $\mathfrak{F}^{r-1}$ where the associated bad event $B_i^{r-1}$ does not hold. Hence, for $N = \Delta$, $\lambda = \frac{n}{\log^{10} n}$ and $D = 2n^{1/3}$, we get that 
\begin{align*}
\Pro{X_{i}^{r_i} \geq X_{i}^{r_{i+1}} + \lambda} \leq \exp\left(-\frac{n^2/\log^{20}n}{10 \cdot \Delta \cdot (4 n^{2/3})} \right) + \frac{2 (\log n)^{11(j - 1)}}{n^4} \leq \frac{3 (\log n)^{11(j - 1)}}{n^4}.
\end{align*}
Taking the union bound over the at most $\log^{10} n$ intervals $i \in [q]$, it follows that
\begin{align*}
\Pro{ \left. \bigcup_{i \in [q]} \left\{ X_{i}^{r_{i+1}} \geq X_{1}^{r_1} + i \cdot \frac{n}{\log^{10} n} \right\} \,\right|\, \mathfrak{F}^{t_0}} \leq \log^{10} n \cdot \frac{3 (\log n)^{11(j - 1)}}{n^4} \leq \frac{1}{2} \cdot \frac{(\log n)^{11j}}{n^4}.
\end{align*}
Conditioning on $\Phi_j^{t_0} \leq Cn$ (which implies $X_1^{r_1} \leq Cn$ since $C \geq  6$) and $q \leq \log^{10} n$, we get
\begin{align*}
\Pro{ \left. \bigcap_{i \in [q]} \left\{ X_{i}^{r_{i+1}} \leq Cn + n \right\} ~\right|~ \mathfrak{F}^{t_0}, \Phi_j^{t_0} \leq Cn } \geq 1 -  \frac{1}{2} \cdot \frac{(\log n)^{11j}}{n^4}. 
\end{align*}
Since $\overline{\Phi}_{j}^{r_{i+1}} \leq \max \big\{ X_{i}^{r_{i+1}}, 5n + n^{1/3} \big\} \leq Cn + n$, we have
\begin{align*}
\Pro{ \left. \bigcap_{i \in [q]} \left\{ \overline{\Phi}_j^{r_{i+1}} \leq Cn + n \right\} ~\right|~ \mathfrak{F}^{t_0}, \Phi_j^{t_0} \leq Cn } \geq 1 -  \frac{1}{2} \cdot \frac{(\log n)^{11j}}{n^4}. 
\end{align*}
Given that the $i$-th epoch contains at most $n/\log^2 n$ allocations, we have that\begin{align*}
\Pro{ \left. \bigcap_{s \in [t_0, t + n \log^8 n]} \bigcup_{u \in [s, s + n/\log^2 n]} \left\{ \Phi_j^u \leq (C+1)n \right\} ~\right|~ \mathfrak{F}^{t_0}, \Phi_j^{t_0} \leq Cn } \geq 1 -  \frac{1}{2} \cdot \frac{(\log n)^{11j}}{n^4}. 
\end{align*}
Applying the smoothness argument of \cref{clm:phi_j_does_not_drop_quickly} for each step $s$, we get
\begin{align*}
\Pro{\left. \bigcap_{s \in [t_0, t + n \log^8 n]} \left\{ \Phi_j^{s} \leq \frac{C+ 1}{0.99} \cdot n \right\} ~\right|~ \mathfrak{F}^{t_0}, \Phi_j^{t_0} \leq Cn } \geq 1 -  \frac{1}{2} \cdot \frac{(\log n)^{11j}}{n^4}.
\end{align*}
Since $2C \geq \frac{C+1}{0.99}$ (as $C\geq 6$) and $t_0 \leq \beta_j$,
\begin{align} \label{eq:phi_t0_good_whp_conclusion}
\Pro{\left. \bigcap_{s \in [\beta_j, t + n \log^8 n]} \left\{ \Phi_j^{s} \leq 2Cn \right\} ~\right|~ \mathfrak{F}^{t_0}, \Phi_j^{t_0} \leq Cn } \geq 1 -  \frac{1}{2} \cdot \frac{(\log n)^{11j}}{n^4}.
\end{align} 
Now, the next step is to obtain a bound without the conditioning. To this end, we define the stopping time $\tau := \inf\{ t_0 \geq \beta_{j-1} + n \log^5 n \colon \Phi_j^{t_0} \leq Cn \}$. Then by a union bound, 
\begin{align*}
\Pro{\tau \leq \beta_j} 
 & \geq \Pro{ \left( \bigcup_{s \in [\beta_{j-1}+ n \log^5 n, \beta_j]} \left\{ \tilde{\Phi}_j^{s} \leq Cn \right\} \right) \cap \bigcap_{s \in [\beta_{j-1}+ n \log^5 n, \beta_j]} \left\{ \Phi_{j-1}^{s} \leq 2Cn \right\}}  \\
 &\geq \Pro{ \bigcup_{s \in [\beta_{j-1}+ n \log^5 n, \beta_j]} \left\{ \tilde{\Phi}_j^{s} \leq Cn \right\}
 } - \Pro{ \neg \bigcap_{s \in [\beta_{j-1}+ n \log^5 n, \beta_j]} \left\{ \Phi_{j-1}^{s} \leq 2Cn \right\} } 
 \\ 
 & \geq 1 - n^{-5} -  \frac{(\log n)^{11(j-1)}}{n^4} \geq 1 - \frac{2(\log n)^{11(j-1)}}{n^4},
\end{align*}
where in the third inequality we used \cref{lem:exists_s_st_phi_linear_whp} and \cref{eq:precondition_of_induction_step} (since $\beta_j \leq t + n \log^5 n$).

Finally, we get the conclusion by combining this with \cref{eq:phi_t0_good_whp_conclusion},
\begin{align*}
\lefteqn{\Pro{\bigcap_{s \in [\beta_j, t + n \log^8 n]} \left\{ \Phi_j^{s} \leq 2Cn \right\}}} \\ 
 & \geq \sum_{t_0 = \beta_{j-1} + n \log^5 n}^{\beta_j} \Pro{\left. \bigcap_{s \in [\beta_j, t + n \log^8 n]} \left\{ \Phi_j^{s} \leq 2Cn \right\} ~\right|~ \mathfrak{F}^{t_0}, \Phi_j^{t_0} \leq Cn } \cdot \Pro{\tau = t_0} \\
 & \geq \left( 1 - \frac{1}{2} \cdot \frac{(\log n)^{11j}}{n^4} \right) \cdot \left( 1 - \frac{2(\log n)^{11(j-1)}}{n^4} \right) \\
 & \geq 1 - \frac{(\log n)^{11j}}{n^4}. \qedhere
\end{align*}
\end{proof}

\subsection{Proof of Main Theorem~\texorpdfstring{(\cref{thm:caching_log_log_n}) using~\cref{lem:new_inductive_step}}{}} \label{sec:proof_of_caching_log_logn}

{\renewcommand{\thethm}{\ref{thm:caching_log_log_n}}
	\begin{thm}[Restated]
\CachingLogLogN
	\end{thm} }
	\addtocounter{thm}{-1}

\begin{proof}
Consider first the case where $m \geq n \log^8 n$ and let $t = m - n \log^8 n$. We will proceed by induction on the potential functions $\Phi_j$ for $j = 1, \ldots , j_{\max} - 1$. The base case follows by \cref{thm:gamma_concentration}, since $\Phi_0^s \leq \Gamma^s$ for all steps $s \in [t, t + n \log^8 n]$. Hence,
\[
\Pro{\bigcap_{s \in [t, m]} \left\{ \Phi_0^{s} \leq 2Cn \right\} } \geq 1 - n^{-4}.
\]
For the induction step, we use \cref{lem:new_inductive_step}. After $j_{\max} = \Theta(\log \log n)$ applications, we get
\[
\Pro{ \bigcap_{s \in [t+ \beta_{j_{\max}-1},m]} \left\{ \Phi_{j_{\max}-1}^{s} \leq 2Cn \right\} } \geq 1-\frac{(\log n)^{11 \cdot (j_{\max} - 1)}}{n^4} \geq 1 - n^{-3}.
\]
When this event occurs, the gap at round $m$ cannot be more than $z_{j_{\max} - 1} + \frac{4v^2}{\alpha_2^2}$ for some constant $\kappa := \kappa(\alpha_2, v) > 0$, since otherwise
\begin{align*}
2Cn \geq \Phi_{j_{\max}-1}^{m} & \geq \exp \Big(\alpha_2 \cdot v^{j_{\max} - 1} \cdot \frac{4v^2}{\alpha_2^2} \Big) \\
 & = 
 \exp\Big(4 \cdot v^{j_{\max}} \cdot \frac{v}{\alpha_2} \Big) = \exp\Big(4 \cdot \frac{\alpha_2}{2v} \log n \cdot \frac{v}{\alpha_2} \Big) = \exp(2 \cdot \log n) = n^2,
\end{align*}
which leads to a contradiction. Hence, $\Gap(m) \leq z_{j_{\max} - 1} + \frac{2v^2}{\alpha_2^2} = \frac{5v}{\alpha_2} \cdot (\log_v \left(\frac{\alpha_2}{2v} \log n\right) - 1) + \frac{2v^2}{\alpha_2^2} \leq \kappa \cdot \log \log n$, for some constant $\kappa := \kappa(\alpha_2, v), \kappa(a, b) > 0$, so
\[
\Pro{\Gap(m) \leq \kappa \cdot \log \log n} \geq 1 - n^{-3}.
\]

The other case is $m < n \log^8 n$, when some of the $\beta_j$'s of the analysis above will be negative. To fix this, consider a modified process. The modified process starts at round $n \log^8 n - m$ with an empty load vector. For any time $t \in [n \log^8 n,-m]$, it allocates a ball of fractional weight $\frac{1}{n}$ to each of the $n$ bins. For $t \geq 1$, it works exactly as the original \Memory process. Since the load vector is perfectly balanced at each step $t < 0$, it follows that $\dot{\Psi}_j^{t} = \dot{\Phi}_j^{t} = 0$ (and $\Phi_j^t = \Psi_j^t = n$) deterministically. Since our proof relies only on upper bounds on the potential functions, these are trivially satisfied and hence the above analysis applies for the modified process. Further, as the relative loads of the modified process and the original process behave identically for $t \geq 1$, the statement follows.
\end{proof}

\section{Proofs for \Memory and \texorpdfstring{$d$}{d}-\WeakMemory in Relaxed Settings} \label{sec:relaxed_settings}

In this section we extend our bounded gap result for \Memory  to arbitrary sampling distributions with full support and prove an $\mathcal{O}(\log n)$ gap bound for \DWeakMemory with weighted balls. 

\subsection{General Biased Sampling}
Recall that \cref{thm:caching_log_log_n} proves that, for any constants $a,b\geq 1$, the \Memory process on any $(a,b)$-biased distribution has an $\mathcal{O}(\log\log n)$ gap \Whp.  The following result shows that the gap remains bounded independently of $m$ \Whp~for \textit{any} sampling distributions with full support.  

We start with the following simple lemma that allows us to apply \cref{thm:ls22_thm_3_1} as we did in the proof of \cref{lem:expectation_bound}.

\begin{lem} \label{lem;arbitrary_to_const_quantile}
Consider any probability vector $p$ satisfying condition $\mathcal{C}_1$ for $(\delta, \eps)$ for some $\delta \leq 1/2$ and $\eps \in (0, 1)$. Then, $p$ also satisfies $\mathcal{C}_1$ for $(1/2, \eps')$ with $\eps' = \eps/n$.
\end{lem}
\begin{proof}
We will prove that $p$ satisfies the $\mathcal{C}_1$ condition for $(1/2, \eps')$ with $\eps' = \eps/n$. We consider the following cases for the index $k \in [n]$:
\begin{itemize}
  \item \textbf{Case A [$1 \leq k \leq \delta n$]:} 
  \[
  \sum_{i = 1}^k p_i 
    \leq (1 - \eps) \cdot \frac{k}{n} 
    \leq \left( 1 - \frac{\eps}{n} \right) \cdot \frac{k}{n}.
  \]
  \item \textbf{Case B [$\delta n < k \leq n/2$]:}
  \begin{align*}
  \sum_{i = 1}^k p_i 
    & \leq 1 - \left( 1 + \eps \cdot \frac{\delta}{1 - \delta} \right) \cdot \left( 1 - \frac{k}{n} \right) \\
    & = \frac{k}{n} + \frac{\eps\delta}{1 - \delta} \cdot \left( \frac{k}{n} - 1 \right) \\ 
    & \stackrel{(a)}{\leq} \frac{k}{n} - \frac{\eps\delta}{1 - \delta} \cdot \frac{k}{n} \\
    & = \frac{k}{n} \cdot \left( 1 - \frac{\eps\delta}{1 - \delta} \right) \\
    & \stackrel{(b)}{\leq} \frac{k}{n} \cdot \left( 1 - \frac{\eps}{n} \right),
  \end{align*}
  using in $(a)$ that $k \leq n/2$ and in $(b)$ that $\delta/(1-\delta) \geq 1/n$ since $\delta \geq 1/n$.
  \item \textbf{Case C [$n/2 < k$]:}
  \begin{align*}
  \sum_{i = k}^n p_i 
   & \geq \left( 1 + \eps \cdot \frac{\delta}{1 - \delta} \right) \cdot \frac{n - k + 1}{n} \geq \left( 1 + \frac{\eps}{n} \right) \cdot \frac{n - k + 1}{n}. \qedhere
  \end{align*}
\end{itemize}
\end{proof}

{\renewcommand{\thethm}{\ref{thm:arbdist}}
	\begin{thm}[Restated]
\arbdist
	\end{thm} }
	\addtocounter{thm}{-1} 
\begin{proof}
To begin, we observe the every distribution $s$ satisfying $s_{\min}>0$ is an $(a,b)$-biased distribution for $a = 1/(s_{\min}n)$ and $b\leq n(1-(n-1)s_{\min})$. Also observe that since every $(a,b)$-biased distribution is an  $(a',b')$-biased distribution for any $a'\geq a$ and $b'\geq b$. Thus is suffices to consider  $(a,b)$-biased distributions where $a = 1/(s_{\min}n)$ and $b\leq n(1-(n-1)s_{\min})$ and we can assume both $a$ and $b$ are sufficiently large.

We proceed the same way as in the proof of \cref{lem:expectation_bound}. In \cref{clm:c1satisfied}, we showed that $(\widehat{p}_i)_{i \in [n]}$ satisfies condition \COne for $\delta=M/n=\frac{a-1}{ab-1}$, $\eps = 1/2$ when we pick $d = \left\lceil \frac{2b(ab-1)^2}{(b-1)^2(1-\eps)} \right\rceil$. Since $(a, b)$ are not-necessarily constants, we use \cref{lem;arbitrary_to_const_quantile} to obtain that $(\widehat{p}_i)_{i \in [n]}$ satisfies \COne for (constant) $\delta' = 1/2$ and $\eps' = 1/n$. This will allow us to apply \cref{thm:ls22_thm_3_1} which requires the quantile $\delta$ to be constant.  Observe that by \eqref{eq:ParametersforExpDrop} we can also take $K=2bd^3$.

As in the proof of \cref{lem:expectation_bound}, applying \cref{thm:ls22_thm_3_1}, we get that there exists a constant $c > 0$ such that the hyperbolic cosine potential $\Gamma$ with smoothing parameter $\alpha \leq  \min\{\frac{1}{d}, \frac{\eps'\delta}{d\cdot 8K}   \} $, for any step $s$ such that $d$ divides $s$, we have that\[
\Ex{\Gamma^s} \leq cn.
\]
By Markov's inequality, we have that\[
\Pro{\Gamma^s \leq cn^3} \geq 1 - n^{-2}.
\]
Hence, since the gap can increase by at most $d$ in an interval of $d$ steps, we have that for every step $t$,
\[
\Pro{\Gap(t) \leq d + \frac{\log(cn^3)}{\alpha}} \geq 1 - n^{-2}.
\] Recall that we can assume that $a,b>100$ by the opening paragraph. Finally as $\eps'=1/n$ and  $\delta =  \frac{a-1}{ab-1}\geq 1/n $, we can choose  $\alpha =   \Omega(\frac{1}{n^2 bd^4 })$. Thus, as $d = \Theta\left(  \frac{b(ab-1)^2}{(b-1)^2 }\right)$, we have  
 \[
d + \frac{\log(cn^3)}{\alpha}  = \mathcal{O}\left(bd^4 \cdot n^2\log n\right)  =  \mathcal{O}\left(\frac{ b^5(ab-1)^{10}}{(b-1)^{10}}\cdot n^3\right)   =  \Oh\left(a^{10} b^5 \cdot n^3\right) =  \Oh\left(\frac{n^8}{s_{\min}^{10}}\right),
\]which concludes the claim.
\end{proof}

\subsection{Weighted 2-\WeakMemory}

We consider the weighted setting for $2$-\WeakMemory. Following~\cite{PTW15}, the weight of each ball will be drawn independently from a fixed distribution $W$ over $[0,\infty)$, satisfying:
\begin{align} 
  &\ex{W} = 1.\label{eq:weight1}\\
  &\ex{e^{\lambda W} } < \infty \text{ for some constant }\lambda > 0\label{eq:weight2}.
\end{align}
It is clear that when $\ex{W} = \Theta(1)$, by scaling $W$, we can always achieve $\ex{W}=1$.
Specific examples of distributions satisfying above conditions (after scaling) are the geometric, exponential, binomial and Poisson distributions.

We will make use of the following lemma:
\begin{lem}[Lemma 2.4 in \cite{LS22Batched}]\label{lem:lss22_lem_2_4} For any random variable $W$ satisfying \eqref{eq:weight1} and \eqref{eq:weight2}, there exists $S := S(\lambda) \geq \max(1, 1/\lambda)$, such that for any  $\alpha \in (0, \min(\lambda/2, 1))$ and any $\kappa \in [-1,1]$,
\[
\Ex{e^{\alpha \cdot \kappa \cdot W}} \leq 1 + \alpha \cdot \kappa + S \alpha^2 \cdot \kappa^2.
\]
\end{lem}
We can now state and prove the result of this subsection.

{\renewcommand{\thethm}{\ref{thm:weighted}}
	\begin{thm}[Restated]
\weighted
	\end{thm} }
	\addtocounter{thm}{-1}

\begin{proof}
Consider any even step $2t \geq 0$ for the $2$-\WeakMemory process and let $W_1$ and $W_2$ be the weights of the $(2t + 1)$-th and $(2t + 2)$-th balls. Then, with probability:
\begin{itemize}
  \item $q_{11} = \frac{1}{n} \cdot \frac{i}{n}$: we allocate both balls to the $i$-th most loaded bin.
  \item $q_{01} = \frac{i-1}{n} \cdot \frac{1}{n}$: we allocate the second ball to the $i$-th most loaded bin.
  \item $q_{10} = \frac{1}{n} \cdot \left(1 - \frac{i}{n}\right)$: we allocate the first ball to the $i$-th most loaded bin.
  \item $q_{00} = 1 - q_{11} - q_{01} - q_{10}$: we allocate none of the two balls to the $i$-th most loaded bin.
\end{itemize}
We now consider the expected change of the $\Phi$ potential over two steps  
\begin{align*}
\Ex{\left. \Phi^{2t+2} \,\right|\, \mathfrak{F}^{2t} }
 & = \Phi^{2t} \cdot \Big( q_{00} \cdot \Ex{e^{-\alpha (W_1 + W_2)/n}}  + q_{10} \cdot \Ex{e^{\alpha W_1 (1 - 1/n)}} \cdot \Ex{e^{-\alpha W_2/n}} \\
 & \qquad \qquad + q_{01} \cdot \Ex{e^{-\alpha W_1/n}} \cdot \Ex{e^{\alpha W_2 (1 - 1/n)}} + q_{11} \cdot \Ex{e^{\alpha (W_1 + W_2) (1 - 1/n)}} \Big) \\
 & \stackrel{(a)}{\leq} \Phi^{2t} \cdot \Big( 1 + \alpha \cdot \Big( - q_{00} \cdot \frac{2}{n} + q_{10} \cdot \Big( 1 - \frac{2}{n} \Big) \\
 & \qquad \qquad + q_{01} \cdot \Big( 1 - \frac{2}{n} \Big) + q_{11} \cdot \Big( 2 - \frac{2}{n} \Big) \Big) + 6S \cdot \frac{\alpha^2}{n} \Big) \\
 & = \Phi^{2t} \cdot \Big( 1 + \alpha \cdot \Big( \frac{i-1 + n-i + 2i}{n^2} - \frac{2}{n} \Big) + 6S \cdot \frac{\alpha^2}{n} \Big) \\
 & = \Phi^{2t} \cdot \Big( 1 + \alpha \cdot \Big( \frac{2i - 1}{n^2} - \frac{1}{n}\Big) + 6S \cdot \frac{\alpha^2}{n} \Big),
\end{align*}
using \cref{lem:lss22_lem_2_4} in $(a)$, for some constant  $S := S(\lambda) \geq \max(1, 1/\lambda)$.

Similarly, we have \[
\Ex{\left. \Psi^{2t+2} \right| \mathfrak{F}^{2t}} 
 \leq \Psi^{2t} \cdot \Big( 1 + \alpha \cdot \Big( \frac{2i - 1}{n^2} - \frac{1}{n}\Big) + 6S \cdot \frac{\alpha^2}{n} \Big).
\]
By noticing that $p_i = \frac{2i - 1}{n^2}$ is the \TwoChoice probability allocation vector, by \cite[Proposition 2.3]{LS22Batched} it satisfies condition $\mathcal{C}_1$ with $\delta=\frac{1}{4}$ and $\epsilon=\frac{1}{2}$. So, applying \cref{thm:ls22_thm_3_1} at even indices, there exists a constant $c > 0$ such that for $\alpha = \frac{1}{384 S}$ \[
\Ex{\Gamma^{2t+2}} \leq \Gamma^{2t} \cdot \Big(1 - \frac{\alpha}{384n} \Big) + c \cdot \alpha,
\]
implying that for any step $2t \geq 0$,\[
\Ex{\Gamma^{2t}} \leq 384c \cdot n.
\]
By using Markov's inequality,
\[
\Pro{\Gap(2t) \leq \frac{1}{\alpha} \cdot \left( 3 \log n + \log(384c) \right)} \geq 1 - n^{-2}.
\]
Let $t := \lceil m/2 \rceil$. If $m = 2t$, then we are done, otherwise the $(2t)$-th ball is \Whp~$\Oh(\log n)$, because the MGF is constant. Hence, in that last step the gap cannot change by more than $\Oh(\log n/n)$ and hence, we deduce the claim for $\kappa := 4/\alpha$. 
\end{proof}

\section{Proofs of the Lower Bounds}\label{sec:lower}
Our first result gives a lower bound on the gap of the \Memory process.

{\renewcommand{\thethm}{\ref{thm:caching_lower}}
	\begin{thm}[Restated]
\cachelower 
	\end{thm} }
	\addtocounter{thm}{-1}

On a high level, the proof of this theorem follows the layered induction argument used by \cite{ABKU99} to lower bound the gap of \TwoChoice in the lightly loaded case. However, for the \Memory process in the heavily loaded case, we require some additional arguments to bootstrap the induction and also deal with correlations introduced by the cache. 

\subsection{Proof of Theorem \ref{thm:caching_lower}}

We first provide some general definitions and notations, used in the proof. Recall the definition of the hyperbolic cosine potential from \cref{sec:prelim}:
\[
 \Gamma^{t} = \sum_{i=1}^n e^{\alpha y_i^t} + e^{-\alpha y_i^t}.
\]
The starting point of the proof is round $t_0:=m-n$.
For convenience we may assume that $m$ is a multiple of $n$. This holds since if $m$ is not a multiple of $n$, we may simply apply the theorem to the largest $\tilde{m} \leq m$ which is a multiple of $n$, and deduce the same gap bound as the (normalized) load of any bin can only decrease by at most $1$ within $n$ steps.

We will now divide the interval $[t_0,m]$ into consecutive phases, labeled $1$ to $j_{\max}:= \epsilon \cdot \log \log n$ for a sufficiently small constant $\epsilon := 1/(2 \log C)$, where $C:=100$. Each phase lasts for $\ell:=n/ j_{\max}$ steps. Further, associated to each phase $j \geq 1$, we define the set of bins 
\[
  B_{j} := \left\{ i \in [n] \colon x_i^{t_0
 +j \cdot \ell} \geq \frac{t_0}{n} + L + j \right\},
\]
for some constant $L$ to be defined below and we define the event
\[
 \mathcal{E}_j := \left\{ \left| B_j \right | \geq \frac{n}{(\log n)^{C^{j}}} \right\}.
\] 
  We also define the events $\mathcal{F} := \{ \Gamma^{t_0} \leq 6cn \}$ and $\mathcal{L} = \{ y_{\min}^{t_0} \geq -\tilde{c} \cdot \log n\}$ for some constants $c, \tilde{c}>0$. The key step in the proof of Theorem \ref{thm:caching_lower} is to establish the following two claims about the events $\mathcal{E}_j$, where $0 \leq j \leq j_{\max}$.
\begin{clm}\label{lem:start_event}
We have that $\mathcal{L} \cap \mathcal{E}_0 \subseteq \mathcal{F}$ and $\Pro{\mathcal{F}} \geq 1 - n^{-2}$.
\end{clm}  

\begin{clm}
\label{lem:chain_events}
For any $j \in [1,j_{\max}]$ we have $\Pro{ \mathcal{E}_j \, \left| \, \mathfrak{F}^{t_0+\ell \cdot (j-1)},\; \cap_{k=0}^{j-1} \mathcal{E}_k , \; \mathcal{F} \right.} \geq 1 - n^{-\omega(1)}$.
\end{clm}

Before establishing these two claims, we complete the proof of \cref{thm:caching_lower}.

\begin{proof}[Proof of \cref{thm:caching_lower} (assuming \cref{lem:start_event} and \cref{lem:chain_events} hold)]
First observe that for any $1 \leq j \leq j_{\max}$,  if the event $\mathcal{E}_j$ holds then  
\begin{align}
 |B_j| \geq \frac{n}{(\log n)^{C^{j}}} \geq \frac{n}{(\log n)^{C^{j_{\max}}}} \geq \frac{n}{(\log n)^{\sqrt{\log n}}} \geq \sqrt{n}. \label{eq:is_poly}
\end{align}
Thus if $\mathcal{E}_{j_{\max}}$ holds, there is a bin $i \in [n]$ with $x_i^{m} \geq \frac{t_0}{n} + L + j_{\max}$, hence
\[
 y_i^m = x_i^m - \frac{m}{n} 
 \geq \frac{t_0}{n} + L + j_{\max} - \frac{m}{n} = L + j_{\max} - 1   = \Omega(\log \log n),
\]
as $|L|=\mathcal{O}(1)$. It remains to lower bound $\Pro{\mathcal{E}_{j_{\max}}}$. Using \cref{lem:chain_events}, 
\begin{align*}
 \Pro{ \mathcal{E}_{j_{\max}} \, \mid \, \mathfrak{F}^{t_0}, \mathcal{F} }
 & \geq \Pro{ \left. \bigcap_{k = 0}^{j_{\max}} \mathcal{E}_{k} \, \right| \, \mathfrak{F}^{t_0}, \mathcal{F} } \\
 & = \Pro{ \mathcal{E}_{j_{\max}} \, \left| \, \mathfrak{F}^{t_0}, \mathcal{F}, \bigcap_{k = 0}^{j_{\max} - 1} \mathcal{E}_{k} \right. } \cdot \Pro{ \left. \bigcap_{k = 0}^{j_{\max} - 1} \mathcal{E}_{k} \, \right| \, \mathfrak{F}^{t_0}, \mathcal{F}} \\
 & \geq (1 - n^{-\omega(1)})^{j_{\max}} \\
 & =1 - n^{-\omega(1)}.
 \end{align*}

Further, as $\mathcal{F}$ is measurable with respect to  $\mathfrak{F}^{t_0}$ and $\Pro { \mathcal{F}  } \geq 1-n^{-2}$ by \cref{lem:start_event},  
\begin{align*}
  \Pro{ \mathcal{E}_{j_{\max}} } &\geq \Pro{ \mathcal{E}_{j_{\max}} \cap \mathcal{F} } \\
  &\geq \Pro{\mathcal{F}} \cdot \Ex{\Pro{ \mathcal{E}_{j_{\max}} \, \mid \, \mathfrak{F}^{t_0}, \; \mathcal{F} }} \\
  &\geq (1-n^{-2}) \cdot \left( 1 - n^{-\omega(1)} \right) \\
  &\geq 1-n^{-1},
\end{align*}
which completes the proof.
 \end{proof}

It remains to prove \cref{lem:start_event} and \cref{lem:chain_events}. First we establish \cref{lem:start_event}, which concerns the event $\mathcal{E}_0$.
\begin{proof}[Proof of \cref{lem:start_event}]
By \cref{simp}, there exist some constants $c,\alpha > 0$ such that for $\Gamma := \Gamma(\alpha)$ and $\mathcal{F}:=\left\{ \Gamma^{t_0} \leq 6cn \right\}$ we have
\begin{equation*}
 \Pro { \mathcal{F} } \geq 1-n^{-2}.  
\end{equation*}
Furthermore, for convenience, in the following we can assume that $c$ in the definition of $\mathcal{F}$ satisfies $c \geq 1$.
By the definition of $\Gamma$, the event $\mathcal{F}$ implies that $y_{\min}^{t_0} \geq -\tilde{c} \cdot \log n$ for some constant $\tilde{c}>0$, that is $\mathcal{L} \supseteq \mathcal{F}$. By a first moment argument,  $\mathcal{F}$ also implies that there exist constants $c_1 >0 $ and $c_2 > 0$ such that at time $t_0$, there are at least $c_1 \cdot n$ bins with a load in $[-c_2,+c_2]$. Hence by the pigeonhole principle, there is a load threshold $L \in [-c_2,c_2]$ such that the number of bins $i$ with $y_i^{t_0} = L$ (equivalently, $x_i^{t_0} = \frac{t_0}{n} + L$) is at least $c_1/(2c_2+1) \cdot n$. Note that $L$ may be positive or negative (or zero), all we need is that it is in the interval $[-c_2,c_2]$. Let us define $B_0 :=\left\{ i \in [n]  \colon x_i^{t_0} = \frac{t_0}{n} + L \right\}$; so $|B_0| \geq c_3 \cdot n \geq \frac{n}{(\log n)^{C}}$ for the constant $c_3 :=c_1/(2c_2+1)$.  
\end{proof}

Next we prove the more involved induction step from $j-1$ to $j$: 
\begin{proof}[Proof of \cref{lem:chain_events}]
In order to establish this key inequality, we will start the analysis from step $t_0+\ell \cdot (j-1)$ onwards (the first step of phase $j$), and assume for this step an arbitrary load (and cache) configuration such that $\mathcal{L} \cap \bigcap_{k=0}^{j-1} \mathcal{E}_{k}$ holds. 

Consider any step $s$ in phase $j$. Since $\mathcal{E}_{j-1}$ holds, we have
\[
 \left| \left\{ i \in [n] \colon x_i^{s} \geq \frac{t_0}{n} + L + (j-1)  \right\} \right| \geq \frac{n}{(\log n)^{C^{j-1}}}.
\]
Further, we may assume that 
\begin{align}
 \left| \left\{ i \in [n] \colon x_i^{s} \geq \frac{t_0}{n} + L + j \right\} \right| < \frac{n}{(\log n)^{C^{j}}}, \label{eq:assumption}
\end{align}
since otherwise $\mathcal{E}_{j}$ holds, and we are done. Combining the last two inequalities yields,
\begin{align}
 \left| \left\{ i \in [n] \colon x_i^{s} = \frac{t_0}{n} + L + (j-1)  \right\} \right| 
 &\geq \frac{n}{(\log n)^{C^{j-1}}} - \frac{n}{(\log n)^{C^{j}}}  \geq \frac{n}{2 (\log n)^{C^{j-1}}}. \label{eq:equal_lower}
\end{align}

Regarding the bin $b^s$ in the cache at step $s$, we can deduce that $x_{b^s}^{s} \geq x_{\min}^s \geq x_{\min}^{t_0}  \geq \frac{t_0}{n} -\tilde{c} \cdot \log n$, since we are assuming that the event $\mathcal{L}$ holds. 

Recall that whenever in a round $r$ the load of the sampled bin $i^r$ is strictly greater than the load of the cached bin $b^r$, we allocate the ball to the cached bin (and do not update the cache). Further, if the sampled load is at least the load of the cache, then we necessarily increment one bin with that load.
With this in mind, define the stopping time
\[
\rho:= \min \left\{ r \geq s\; \colon \;x_{b^s}^r \geq \varphi(j) \right\},
\]
where $\varphi(j) := \frac{t_0}{n} + L + (j-1)$. Note that
\begin{align}
 &\Pro{ \rho \leq \varphi(j) - x_{b^s}^s  \, \left| \, \mathfrak{F}^{s}, \;\cap_{k=0}^{j-1} \mathcal{E}_k , \; \mathcal{F}\right.} \notag \\&\qquad \geq \Pro{  \bigcap_{r=s}^{\varphi(j) - x_{b^s}^s} \left\{ x_{i^{r}}^r \geq x_{b^{s}}^s + (r-s)   \right\} ~\bigg|~ \mathfrak{F}^{s}, \;\cap_{k=0}^{j-1} \mathcal{E}_k , \; \mathcal{F}} \notag \\
 &\qquad \geq \prod_{r=s}^{\varphi(j) - x_{b^s}^s}\Pro{    x_{i^{r}}^r \geq x_{b^{s}}^s + (r-s)   ~\bigg|~ \mathfrak{F}^{r-1},\; \cap_{k=0}^{j-1} \mathcal{E}_k, \; \mathcal{F} } \notag \\
 & \qquad \geq \prod_{\lambda=x_{b_s}^s-\frac{t_0}{n}}^{L+j-1} p(\lambda).
 \label{eq:cache_bound}
\end{align}
where
\[
 p(\lambda):= \Pro{    x_{i^{\frac{t_0}{n} + \lambda-x_{b_s}^s}}^{\frac{t_0}{n} + \lambda-x_{b_s}^s} \geq \frac{t_0}{n} + \lambda  ~\bigg|~ \mathfrak{F}^{r-1},\; \cap_{k=0}^{j-1} \mathcal{E}_k , \; \mathcal{F}};
\]
the important thing to remember is lower bounding $p(\lambda)$ means that we are lower bounding the probability of sampling a bin whose load is at least $\frac{t_0}{n}+\lambda$ (conditioned on the previous $j-1$ layers of ``good'' events).
Recall that since we assume that $\mathcal{F}:=\left\{ \Gamma^{t_0} \leq 6 cn \right\}$ holds, we infer that for any $\lambda \in \mathbb{R}$,
\[
 \left| \left\{ i \in [n] \colon y_{i}^{t_0} < \lambda  \right\} \right| \leq 6cn \cdot e^{\alpha \lambda},
\]
and this implies that for any future step $r \geq t_0$,
\[
 \left| \left\{ i \in [n] \colon x_{i}^{r} < \frac{t_0}{n} + \lambda  \right\} \right| 
 \leq \left| \left\{ i \in [n] \colon x_{i}^{t_0} < \frac{t_0}{n} + \lambda  \right\} \right|
 \leq 6cn \cdot e^{\alpha \lambda}.
\]
Hence for any $r \geq s$,
\begin{align}
 \Pro{ \left. x_{i^{r}}^r \geq \frac{t_0}{n} + \lambda \,\,\right|\,\, \mathfrak{F}^{r-1},\;  \cap_{k=0}^{j-1} \mathcal{E}_k , \mathcal{F}} \geq 1 - 6c \cdot e^{\alpha \lambda}. \label{eq:estimate_one}
\end{align} 
Secondly, we have an alternative estimate based on the fact that $\mathcal{F}$ implies $|B_0| \geq c_3 \cdot n$.  For any $\lambda \leq L$,
\begin{align}
 \Pro{ \left. x_{i^{r}}^r \geq \frac{t_0}{n} + \lambda \,\,\right|\,\, \mathfrak{F}^{r-1}, \; \cap_{k=0}^{j-1} \mathcal{E}_k, \; \mathcal{F}  } &\geq
 \Pro { i^{r} \in B_{0} \, \left| \, \mathfrak{F}^{r-1},\; \cap_{k=0}^{j-1} \mathcal{E}_k, \; \mathcal{F}   \right.} \notag \\ &\geq \frac{|B_{0}|}{n} \geq c_3 > 0. \label{eq:estimate_two}
\end{align} 
Finally, for any $1 \leq \lambda < j$, we have
\begin{align}
 \Pro{ \left. x_{i^{r}}^r \geq \frac{t_0}{n} + L +\lambda \,\,\right|\,\, \mathfrak{F}^{r-1}, \cap_{k=0}^{j-1} \mathcal{E}_k , \mathcal{F} }  
 &\geq \frac{|B_{\lambda}|}{n} \geq (\log n)^{-C^{\lambda}}. \label{eq:estimate_three}
\end{align}
We will now apply  \cref{eq:estimate_one}, \cref{eq:estimate_two} and \cref{eq:estimate_three} in order to lower bound \cref{eq:cache_bound}. For simplicity, let us assume that $L$ is positive (the case where $L$ is negative is similar, and we obtain an even stronger lower bound on the probability). Returning to the product in \cref{eq:cache_bound} where the load threshold ranges from $x_{b_s}^{s}$ to $L+j$, we group the load values into three parts: $(i)$ from $x_{b_s}^{s}$ to $-c_4$; $(ii)$ from $-c_4+1$ to $L$ and $(iii)$ from $L+1$ to $L+j$, where the constant $c_4 > 0$ as the smallest integer such that:
\[
 6c \cdot \sum_{z=c_4}^{\infty}  e^{-\alpha(z-1)} \leq 1/2.
\]
Hence, 
\begin{align*}
\Pro{ \rho \leq \varphi(j) - x_{b^s}^s  \, \left| \, \mathfrak{F}^{s}, \;\cap_{k=0}^{j-1} \mathcal{E}_k , \; \mathcal{F}\right.}
&= \prod_{\lambda=x_{b_s}^r-\frac{t_0}{n}}^{-c_4} p(\lambda) \cdot \prod_{\lambda=-c_4+1}^{L} p(\lambda) \cdot  \prod_{\lambda=L+1}^{L+j-1} p(\lambda)   \\
 &\geq  \prod_{\lambda=x_{b_s}^r-\frac{t_0}{n}}^{-c_4} \left(1 - 6c \cdot e^{\alpha \lambda} \right) \cdot c_3^{L+c_4} \cdot \prod_{\lambda=1}^{j-1} (\log n)^{-C^{\lambda}}, \\
  \intertext{where we have lower bounded the first product by \cref{eq:estimate_one}, the second by \cref{eq:estimate_two} and the third by \cref{eq:estimate_three}. Further estimating this lower bound yields,}
 \Pro{ \rho \leq \varphi(j) - x_{b^s}^s  \, \left| \, \mathfrak{F}^{s}, \;\cap_{k=0}^{j-1} \mathcal{E}_k , \; \mathcal{F}\right.}  &\geq  \left( 1 - 6c \cdot \sum_{z=c_4}^{\infty}   e^{-\alpha (z-1)} \right) \cdot c_3^{L+c_4} \cdot (\log n)^{-\sum_{\lambda=1}^{j-1} C^{\lambda} }, \\
 &\geq  \left( 1 - 6c \cdot \sum_{z=c_4}^{\infty}   e^{-\alpha (z-1)} \right) \cdot c_3^{L+c_4} \cdot (\log n)^{-C^{j+1/3}}, \\ 
       &\geq  4 \cdot (\log n)^{-C^{j+2/3}},
 \end{align*}
 where we have used the fact that $C=100$ and the first and second factors in the penultimate line are both constants that are strictly greater than $0$.

Note that $\varphi(j) - x_{b_s}^s \leq c \log n + L + j -1  \leq 2 c \log n$, and therefore,
\begin{equation}\label{eq:rhobdd}
 \Pro{ \rho \leq 2 c \log n \, \left| \, \mathfrak{F}^t, \mathcal{E}_0 \right.} \geq 4 \cdot (\log n)^{-C^{j+2/3}}  .
 \end{equation}
 Hence from any step $s$ in phase $j$, with probability at least $4 \cdot (\log n)^{-C^{j+2/3}}$ after at most $2 c \log n$ additional steps, we reach the situation where the bin in the cache has load at least $x_{b_r}^{r} \geq \varphi(j) = \frac{t_0}{n} + L + j-1$. 

Recall that by \cref{eq:assumption} we assumed that  we have at least $\frac{n}{2 (\log n)^{C^{j}}}$ bins with load $\frac{t_0}{n} + L + (j-1)$ (see \cref{eq:equal_lower}), this means that with probability at least $\frac{1}{2(\log n)^{C^{j}}}$ the next sampled bin has load  $\frac{t_0}{n} + L + (j-1)$, and thus one bin load reaches
 $\frac{t_0}{n} + L + j$.
 
 Let $X_1,X,\ldots,X_{z}$ with $z:=n/(j_{\max} \cdot (2c \log n+1)) $ be independent Bernoulli random variables with success probability 
 \[
  \Pro{ \rho \leq 2 c \log n \, \left| \, \mathfrak{F}^t, \mathcal{E}_0 \right.}\cdot \frac{1}{2(\log n)^{ C^{j}}} \geq  4 \cdot (\log n)^{-C^{j+2/3}} \cdot \frac{1}{2 (\log n)^{C^{j}}} \geq 2 \cdot (\log n)^{-C^{j+3/4}},
 \]
 where the first inequality holds by \eqref{eq:rhobdd} and the last holds for large $n$ since $C=100$. Let $X=\sum_{i=1}^z X_i$. Then, again since we fixed $C=100$ and can assume that $n$ is large, we have 
 \[
  \Ex{X} := \frac{n}{j_{\max} (2c \log n) } \cdot \frac{2}{(\log n)^{C^{j+3/4}}} \geq 2 \cdot \frac{n}{(\log n)^{C^{j+1}}}.
 \] 
 Since $j \leq j_{\max}$, and $j_{\max} \leq \epsilon \cdot \log \log n$, $\epsilon=1/(2 \log C)$, we have $\Ex{X} = \Omega( \poly(n))$ (due to \cref{eq:is_poly}), and thus by a Chernoff bound,
 \[
  \Pro{X \geq \frac{n}{(\log n)^{C^{j+1}}}} \geq \Pro{ X \geq 1/2 \cdot \Ex{X} } \geq 1 - n^{-\omega(1)}.
 \]
 This means, we will have at least $\frac{n}{(\log n)^{C^{j+1}}}$ bins with load at least $\frac{t_0}{n} + L + j$ during one step in phase $j$, which concludes the induction step. Thus we have established the claim.
\end{proof}

\subsection{Lower Bound for \DWeakMemory}

We also give a simple argument for a lower bound on the gap of \DWeakMemory with runs of a constant length $d$. 
 
\begin{lem}\label{lem:d_weak_memory_lower_bound}
For the \DWeakMemory process with constant $d > 0$ and $m = \frac{1}{400d} n \log n$, we have that
\[
\Pro{\Gap(m) \geq \frac{1}{400d} \cdot \log n} \geq 1 - n^{-2}.
\]
\end{lem}
\begin{proof}
In \DWeakMemory, every $d$ steps the cache is reset and so the ball is allocated using \OneChoice. Hence, in $m$ steps, there are $m/d$ balls allocated using \OneChoice. By e.g.~\cite[Section 4]{PTW15} (see also \cite[Lemma A.9]{LS22Noise}), when $cn\log n$ balls are allocated using \OneChoice, for any constant $c>0$, then with probability at least $1 - n^{-2}$, the max load is at least $(c+\sqrt{c}/10)\log n$. Hence, we get\[
\Pro{\max_{i \in [n]} y_i^m \geq \left( \frac{1}{400d^2} + \frac{1}{200d} \right) \cdot \log n} \geq 1 - n^{-2}.
\]
Therefore, since at step $m$ the average is $(\log n)/(400d)$,
\[
\Pro{\Gap(m) \geq \frac{1}{400d} \cdot \log n} \geq 1 - n^{-2}. \qedhere
\]
\end{proof}

\section{Conclusions}\label{sec:conclusions}
	In this work, we presented an asymptotically tight analysis of \Memory in the heavily loaded case $m \geq n$. We proved that the gap of \Memory is $\Theta(\log \log n)$, matching the performance of \TwoChoice up to constants. In contrast to \TwoChoice, we showed \Memory still works well in a heterogeneous setting where the sampling distribution may be distorted by some arbitrarily large constant factor. 
	We also analyzed other relaxed settings, including one where balls are weighted. In these settings, the cache is reset every constant number of steps~(\DWeakMemory). In those cases, we proved that the gap is still $\mathcal{O}(\log n)$ and remains independent of $m$.
	
	There are several interesting directions. One of them is to consider even more skewed sampling distributions, e.g., heavy-tailed distributions such as Power-Law (similar to \cite{BCM04}, where $m=n$ was studied). One might suspect that \Memory is still able to outperform \DChoice. In particular, we have established that Memory gives a bounded gap on even the most unruly sampling distributions (\cref{thm:arbdist}). It would be interesting to get more precise gap bounds for \Memory on specific sampling distributions of interest such as Power-law distributions.
	
	A second direction is to determine the leading constant in the gap bound. The results by \cite{MPS02} for $m=n$ suggest \Memory might be slightly superior to \TwoChoice also in the heavily loaded case.
	
	Another avenue is to study the impact of $d$ in \DWeakMemory (or $d$-\ResetMemory), when $d$ may be a function of $n$. For instance, what is the smallest value of $d$ that still achieves a gap of $\mathcal{O}(\log \log n)$? So far, we only know that $d$ cannot be a constant, but we do not know whether $d$ needs to be, say logarithmic or polynomial.
	
	Finally, to the best of our knowledge, all implementations of \Memory (including ours) make use of the greedy-rule, i.e., always allocate the ball in the least loaded option among the cache and the sampled bin, and update the cache in the same vein. Note that in the presence of a larger cache size and weighted balls, this may not be optimal, as there are more sophisticated strategies that preemptively try to maintain at least one significantly underloaded bin in the cache.


\appendix

\clearpage

\section{Elementary Tools and Inequalities}

In this section, we state several auxiliary lemmas that we use throughout the paper.

\subsection{Convexity}
For completeness, we define Schur-convexity (see \cite{MRBook}) and state two basic results:
\begin{defi}[{cf.~\cite[Definition A.1]{MRBook}}]
A function $f : \mathbb{R}^n \to \mathbb{R}$ is Schur-convex if for any non-decreasing $x, y \in \mathbb{R}^n$, if $x$ majorizes $y$ then $f(x) \geq f(y)$.
\end{defi}
 
\begin{lem}[{cf.~\cite[Proposition C.1]{MRBook}}] \label{lem:sum_of_convex_is_schur_convex}
Let $g : \mathbb{R} \to \mathbb{R}$ be a convex function. Then $g(x_1, \ldots, x_n) := \sum_{i = 1}^n g(x_i)$ is Schur-convex. 
\end{lem} 

\subsection{Inequalities}

We proceed with two simple inequalities, which we include for the sake of completeness.

\begin{lem} \label{lem:small_z_frac_ineq}
For any $0 < z \leq 1$, we have
\[
\frac{1}{1 + z} \leq 1 - \frac{z}{2}.
\]
\end{lem}
\begin{proof}
For $z \in (0, 1]$, the following chain of implications holds
\[
\frac{1}{1 + z} \leq 1 - \frac{z}{2} \Leftrightarrow 1 \leq 1 + z - \frac{z}{2} - \frac{z^2}{2} \Leftrightarrow z \leq 1. \qedhere
\]
\end{proof}

\begin{lem} \label{lem:small_z_frac_ineq_2}
For any $0 < z \leq 1/2$, we have
\[
\frac{1}{1 - z} \leq 1 + 2z.
\]
\end{lem}
\begin{proof}
For $z \in (0, 1]$, the following chain of implications holds
\[
\frac{1}{1 - z} \leq 1 + 2z \Leftrightarrow 1 \leq 1 + 2z - z - 2z^2 \Leftrightarrow z \leq \frac{1}{2}. \qedhere
\]
\end{proof}

\begin{lem}\label{lem:arthgeo}For any $r<1$ we have $\sum_{i=0}^\infty (i+1)\cdot r^i =\frac{1}{1-r} + \frac{r}{(1-r)^2}$. 
\end{lem}
\begin{proof}For the first sum observe that for any $d$, 
	\[(1-r)\sum_{i=0}^d (i+1)\cdot r^i = \sum_{i=0}^d (i+1)\cdot r^i - \sum_{i=1}^{d+1} i\cdot r^i = 1 + \sum_{i=1}^d r^i + (d+1)r^{d+1}.\] Thus letting $d\rightarrow \infty$ and using the sum for the geometric series $\sum_{i=1}^\infty  r^i = \frac{r}{1-r}$ gives 
	\[(1-r)\sum_{i=0}^\infty (i+1)\cdot r^i = 1 +\frac{r}{1-r}, \] as claimed. 
\end{proof}

\subsection{Probabilistic Inequalities}

For convenience, we state and prove the following well-known result.

\begin{lem} \label{lem:geometric_arithmetic}
Consider a sequence of random variables $(Z_i)_{i \in \mathbb{N}}$ such that there are $0 < a < 1$ and $b > 0$ such that every $i \geq 1$,
\[
\Ex{Z_i \mid Z_{i-1}} \leq Z_{i-1} \cdot a + b.
\]
Then for every $i \geq 1$, 
\[
\Ex{Z_i \mid Z_0}
\leq Z_0 \cdot a^i + \frac{b}{1 - a}.
\]
\end{lem}
\begin{proof}
We will prove by induction that for every $i \in \mathbb{N}$, 
\[
\Ex{Z_i \mid Z_0} \leq Z_0 \cdot a^i + b \cdot \sum_{j = 0}^{i-1} a^j.
\]
For $i = 0$, $\Ex{Z_0 \mid Z_0} \leq Z_0$. Assuming the induction hypothesis holds for some $i \geq 0$, then since $a > 0$,
\begin{align*}
\Ex{Z_{i+1} \mid Z_0} & = \Ex{\Ex{Z_{i+1} \mid Z_i}\mid Z_0} \leq \Ex{Z_{i}\mid Z_0} \cdot a + b \\
 & \leq \Big(Z_0 \cdot a^i + b \cdot \sum_{j = 0}^{i-1} a^j \Big) \cdot a + b \\
 & = Z_0 \cdot a^{i+1} +b \cdot \sum_{j = 0}^i a^j.
\end{align*}
The claims follows using that for $a \in (0,1)$, $\sum_{j=0}^{\infty} a^j = \frac{1}{1-a}$.
\end{proof}

\begin{lem} \label{lem:recurrence_inequality_expectation}
Consider any sequence of random variables $(Z_i)_{i \in \N}$ such that there exists $0 < \alpha < 1$ and $b > 0$
\[
\Ex{Z_i \mid Z_{i-1}} \leq Z_{i - 1} \cdot (1 - \alpha) + b.
\]
Then, assuming that $Z_0 \leq \frac{b}{\alpha}$, then for any $i \geq 0$, 
\[
\Ex{Z_i} \leq \frac{b}{\alpha}.
\]
\end{lem}
\begin{proof}
We will prove the claim by induction. The base case follows by the assumption $Z_0 \leq \frac{b}{\alpha}$. Then, assuming that it holds for $i \geq 0$, then for $i+1$ we have, \begin{align*}
\Ex{Z_{i+1}}
  & = \Ex{\Ex{Z_{i+1} \mid Z_{i}}} \\
  & \leq \Ex{Z_{i} \cdot ( 1 - \alpha ) + b} \\
  & = \Ex{Z_{i}} - \alpha \cdot \Ex{Z_{i}} + b \\
  & \leq \Ex{Z_{i}} - \alpha \cdot \frac{b}{\alpha} + b = \Ex{Z_{i}}. \qedhere 
\end{align*}
\end{proof}

\begin{lem}[Azuma's Inequality for Super-Martingales {\cite[Problem 6.5]{DP09}}] \label{lem:azuma}
Let $X_0, \ldots, X_n$ be a super-martingale satisfying $|X_{i} - X_{i-1}| \leq c_i$ for any $i \in [n]$, then for any $\lambda > 0$,
\[
\Pro{X_n \geq X_0 + \lambda} \leq \exp\left(- \frac{\lambda^2}{2 \cdot \sum_{i=1}^n c_i^2} \right).
\]
\end{lem}

\subsection{Concentration Inequality with a Bad Event}

A central tool in our analysis will be the use of a concentration inequality by Chung and Lu~\cite{CL06} for super-martingales which will be conditional on a bad event not occurring. In this case, the bad event will be $\Phi_1^s = \poly(n)$ for all $s \in [t - T_r, t]$.

We start with the following definitions from~\cite{CL06}. Consider any r.v.~$X$ (in our case it will be the potential function $\Gamma_2$) that can be evaluated by a sequence of decisions $Y_1, Y_2, \ldots ,Y_N$ of finitely many outputs (the chosen bins of the allocated balls). We can describe the process by a \textit{decision tree} $T$, a complete rooted tree with depth $N$ with vertex set $V(T)$. Each edge $(u, v)$ of $T$ is associated with a probability $p_{u, v}$ depending on the decision made from $u$ to $v$. 

We say $f : V (T) \to \mathbb{R}$ satisfies an \textit{admissible condition} $P$ if $P = \{P_v\}$ holds for every vertex $v \in V(T)$. For an admissible condition $P$, the associated bad set $B^i$ over the $X^i$ is defined to be
\[
B^i = \{ v \mid \text{the depth of $v$ is $i$, and $P_u$ does not hold for some ancestor $u$ of $v$} \}.
\]

\begin{thm}[Theorem 8.5 in~\cite{CL06}] \label{thm:orig_chung_lu_theorem_8_5}
For a filter $\mathfrak{F}$, $\{\emptyset, \Omega \} = \mathfrak{F}^{0} \subset \mathfrak{F}^{1} \subset \ldots \subset \mathfrak{F}^{N} = \mathfrak{F}$,
suppose that a random variable $X^{s}$ is $\mathfrak{F}^{s}$-measurable, for $0 \leq s \leq N$. Let $B = B^N$ be the
bad set associated with the following admissible conditions:
\begin{align*}
\Ex{X^{s} \mid \mathfrak{F}^{s-1}} & \leq X^{s - 1}, \\
\Var{X^{s} \mid \mathfrak{F}^{s-1}} & \leq \sigma_s^2, \\
X^{s} - \Ex{X^{s} \mid \mathfrak{F}^{s-1}} & \leq a_s + M,
\end{align*}
for some $\sigma_s > 0$ and $a_s > 0$. Then, we have for any $\lambda >0$,
\[
\Pro{X^{N} \geq X^{0} + \lambda} \leq \exp\left( - \frac{\lambda^2}{2(\sum_{s = 1}^N (\sigma_s^2 + a_s^2)+M \lambda /3) } \right) + \Pro{B}.
\]
\end{thm}

In particular, we always make use of the following simplified version,
\begin{thm}[Corollary of \cref{thm:orig_chung_lu_theorem_8_5}] \label{thm:simplified_chung_lu_theorem_8_5}
For a filter $\mathfrak{F}$, $\{\emptyset, \Omega \} = \mathfrak{F}^{0} \subset \mathfrak{F}^{1} \subset \ldots \subset \mathfrak{F}^{N} = \mathfrak{F}$,
suppose that a random variable $X^{s}$ is $\mathfrak{F}^{s}$-measurable, for $0 \leq s \leq N$. Let $B = B^N$ be the bad set associated with the following admissible conditions:
\begin{align*}
\Ex{X^{s} \mid \mathfrak{F}^{s-1}} & \leq X^{s - 1}, \\
|X^{s} - X^{s-1}| & \leq D,
\end{align*}
for some $D > 0$. Then, we have for any $\lambda >0$,
\[
\Pro{X^{N} \geq X^{0} + \lambda} \leq \exp\left( - \frac{\lambda^2}{10 N D^2 } \right) + \Pro{B}.
\]
\end{thm}
\begin{proof}
We will show that the bounded difference condition $|X^{s} - X^{s-1}| \leq D$, implies the second and third conditions in \cref{thm:simplified_chung_lu_theorem_8_5}. 

For the second condition, we will use Popovicius' inequality for variances (\cref{lem:popovicius_inequality}), which states that for any random variable $Y$ such that $a \leq Y \leq b$,
\begin{align*}  
\Var{Y} \leq \frac{1}{4} \cdot (b - a)^2.
\end{align*}
The bounded difference condition implies that
\[
X^{s-1} - D \leq X^s \leq X^{s-1} + D.
\]
Combining the two, we get that
\begin{align*}
\var{X^{s} \mid \mathfrak{F}^{s-1}}
  \leq \frac{1}{4} \cdot \left( \left. X^{s-1} + D - \bigl(X^{s-1} - D \bigr) \, \right| \, \mathfrak{F}^{s - 1} \right)^2 = D^2.
\end{align*}
For the third condition, 
\[
X^s - \Ex{X^{s} \mid \mathfrak{F}^{s-1}} \leq X^{s-1} + D - (X^{s-1} - D) \leq 2D.
\]
Finally, we get the conclusion by \cref{thm:simplified_chung_lu_theorem_8_5} with $a_s = 2D$, $\sigma_s = D$, and $M = 0$.
\end{proof}

\begin{lem}[Popovicius' Inequality~\cite{P35}] \label{lem:popovicius_inequality}
For any random variable $Y$ satisfying $a \leq Y \leq b$, we have that
\begin{align*}  
\Var{Y} \leq \frac{1}{4} \cdot (b - a)^2.
\end{align*}
\end{lem}


\begin{thebibliography}{33}
	
	
	\ifx \showCODEN    \undefined \def \showCODEN     #1{\unskip}     \fi
	\ifx \showDOI      \undefined \def \showDOI       #1{#1}\fi
	\ifx \showISBNx    \undefined \def \showISBNx     #1{\unskip}     \fi
	\ifx \showISBNxiii \undefined \def \showISBNxiii  #1{\unskip}     \fi
	\ifx \showISSN     \undefined \def \showISSN      #1{\unskip}     \fi
	\ifx \showLCCN     \undefined \def \showLCCN      #1{\unskip}     \fi
	\ifx \shownote     \undefined \def \shownote      #1{#1}          \fi
	\ifx \showarticletitle \undefined \def \showarticletitle #1{#1}   \fi
	\ifx \showURL      \undefined \def \showURL       {\relax}        \fi
	\providecommand\bibfield[2]{#2}
	\providecommand\bibinfo[2]{#2}
	\providecommand\natexlab[1]{#1}
	\providecommand\showeprint[2][]{arXiv:#2}
	
	\bibitem[Alistarh et~al\mbox{.}(2018)]%
	{ABKLN18}
	\bibfield{author}{\bibinfo{person}{Dan Alistarh}, \bibinfo{person}{Trevor
			Brown}, \bibinfo{person}{Justin Kopinsky}, \bibinfo{person}{Jerry~Zheng Li},
		{and} \bibinfo{person}{Giorgi Nadiradze}.} \bibinfo{year}{2018}\natexlab{}.
	\newblock \showarticletitle{Distributionally Linearizable Data Structures}. In
	\bibinfo{booktitle}{\emph{\SPAA{30}{18}}}. \bibinfo{publisher}{ACM},
	\bibinfo{pages}{133--142}.
	\newblock
	\href{https://doi.org/10.1145/3210377.3210411}{\texttt{doi}}
	
	
	\bibitem[Azar et~al\mbox{.}(2020)]%
	{award20}
	\bibfield{author}{\bibinfo{person}{Yossi Azar}, \bibinfo{person}{Andrei~Z.
			Broder}, \bibinfo{person}{Anna Karlin}, \bibinfo{person}{Michael
			Mitzenmacher}, {and} \bibinfo{person}{Eli Upfal}.}
	\bibinfo{year}{2020}\natexlab{}.
	\newblock \bibinfo{title}{{The ACM Paris Kanellakis Theory and Practice
			Award}}.
	\newblock
	\newblock
	\newblock
	\shownote{\url{https://www.acm.org/media-center/2021/may/technical-awards-2020}}.
	
	
	\bibitem[Azar et~al\mbox{.}(1999)]%
	{ABKU99}
	\bibfield{author}{\bibinfo{person}{Yossi Azar}, \bibinfo{person}{Andrei~Z.
			Broder}, \bibinfo{person}{Anna~R. Karlin}, {and} \bibinfo{person}{Eli
			Upfal}.} \bibinfo{year}{1999}\natexlab{}.
	\newblock \showarticletitle{Balanced allocations}.
	\newblock \bibinfo{journal}{\emph{SIAM J. Comput.}} \bibinfo{volume}{29},
	\bibinfo{number}{1} (\bibinfo{year}{1999}), \bibinfo{pages}{180--200}.
	\newblock
	\showISSN{0097-5397}
	\href{https://doi.org/10.1137/S0097539795288490}{\texttt{doi}}
	
	
	\bibitem[Berenbrink et~al\mbox{.}(2012)]%
	{BBFN12}
	\bibfield{author}{\bibinfo{person}{Petra Berenbrink},
		\bibinfo{person}{Andr{\'{e}} Brinkmann}, \bibinfo{person}{Tom Friedetzky},
		{and} \bibinfo{person}{Lars Nagel}.} \bibinfo{year}{2012}\natexlab{}.
	\newblock \showarticletitle{Balls into bins with related random choices}.
	\newblock \bibinfo{journal}{\emph{J. Parallel Distributed Comput.}}
	\bibinfo{volume}{72}, \bibinfo{number}{2} (\bibinfo{year}{2012}),
	\bibinfo{pages}{246--253}.
	\newblock
	\href{https://doi.org/10.1016/j.jpdc.2011.10.006}{\texttt{doi}}
	
	
	\bibitem[Berenbrink et~al\mbox{.}(2014)]%
	{BBFN14}
	\bibfield{author}{\bibinfo{person}{Petra Berenbrink},
		\bibinfo{person}{Andr{\'{e}} Brinkmann}, \bibinfo{person}{Tom Friedetzky},
		{and} \bibinfo{person}{Lars Nagel}.} \bibinfo{year}{2014}\natexlab{}.
	\newblock \showarticletitle{Balls into non-uniform bins}.
	\newblock \bibinfo{journal}{\emph{J. Parallel Distributed Comput.}}
	\bibinfo{volume}{74}, \bibinfo{number}{2} (\bibinfo{year}{2014}),
	\bibinfo{pages}{2065--2076}.
	\newblock
	\href{https://doi.org/10.1016/j.jpdc.2013.10.008}{\texttt{doi}}
	
	
	\bibitem[Berenbrink et~al\mbox{.}(2006)]%
	{BCSV06}
	\bibfield{author}{\bibinfo{person}{Petra Berenbrink}, \bibinfo{person}{Artur
			Czumaj}, \bibinfo{person}{Angelika Steger}, {and} \bibinfo{person}{Berthold
			V\"{o}cking}.} \bibinfo{year}{2006}\natexlab{}.
	\newblock \showarticletitle{Balanced allocations: the heavily loaded case}.
	\newblock \bibinfo{journal}{\emph{SIAM J. Comput.}} \bibinfo{volume}{35},
	\bibinfo{number}{6} (\bibinfo{year}{2006}), \bibinfo{pages}{1350--1385}.
	\newblock
	\showISSN{0097-5397}
	\href{https://doi.org/10.1137/S009753970444435X}{\texttt{doi}}
	
	
	\bibitem[Bohman et~al\mbox{.}(2015)]%
	{BOHMAN2015379}
	\bibfield{author}{\bibinfo{person}{Tom Bohman}, \bibinfo{person}{Alan Frieze},
		{and} \bibinfo{person}{Eyal Lubetzky}.} \bibinfo{year}{2015}\natexlab{}.
	\newblock \showarticletitle{Random triangle removal}.
	\newblock \bibinfo{journal}{\emph{Advances in Mathematics}}
	\bibinfo{volume}{280} (\bibinfo{year}{2015}), \bibinfo{pages}{379--438}.
	\newblock
	\showISSN{0001-8708}
	\href{https://doi.org/10.1016/j.aim.2015.04.015}{\texttt{doi}}
	
	
	\bibitem[Byers et~al\mbox{.}(2004)]%
	{BCM04}
	\bibfield{author}{\bibinfo{person}{John~W. Byers}, \bibinfo{person}{Jeffrey
			Considine}, {and} \bibinfo{person}{Michael Mitzenmacher}.}
	\bibinfo{year}{2004}\natexlab{}.
	\newblock \showarticletitle{Geometric Generalizations of the Power of Two
		Choices}. In \bibinfo{booktitle}{\emph{\SPAA{16}{04}}}.
	\bibinfo{publisher}{ACM}, \bibinfo{pages}{54–63}.
	\newblock
	\showISBNx{1581138407}
	\href{https://doi.org/10.1145/1007912.1007921}{\texttt{doi}}
	
	
	\bibitem[Chung and Lu(2006)]%
	{CL06}
	\bibfield{author}{\bibinfo{person}{Fan Chung} {and} \bibinfo{person}{Linyuan
			Lu}.} \bibinfo{year}{2006}\natexlab{}.
	\newblock \showarticletitle{Concentration inequalities and martingale
		inequalities: a survey}.
	\newblock \bibinfo{journal}{\emph{Internet Math.}} \bibinfo{volume}{3},
	\bibinfo{number}{1} (\bibinfo{year}{2006}), \bibinfo{pages}{79--127}.
	\newblock
	\showISSN{1542-7951}
	\urldef\tempurl%
	\url{http://projecteuclid.org/euclid.im/1175266369}
	\showURL{%
		\tempurl}
	
	
	\bibitem[Dubhashi and Panconesi(2009)]%
	{DP09}
	\bibfield{author}{\bibinfo{person}{Devdatt~P. Dubhashi} {and}
		\bibinfo{person}{Alessandro Panconesi}.} \bibinfo{year}{2009}\natexlab{}.
	\newblock \bibinfo{booktitle}{\emph{Concentration of measure for the analysis
			of randomized algorithms}}.
	\newblock \bibinfo{publisher}{Cambridge University Press, Cambridge}. xvi+196
	pages.
	\newblock
	\showISBNx{978-0-521-88427-3}
	\href{https://doi.org/10.1017/CBO9780511581274}{\texttt{doi}}
	
	
	\bibitem[Giaccone et~al\mbox{.}(2002)]%
	{GPS02}
	\bibfield{author}{\bibinfo{person}{Paolo Giaccone}, \bibinfo{person}{Balaji
			Prabhakar}, {and} \bibinfo{person}{Devavrat Shah}.}
	\bibinfo{year}{2002}\natexlab{}.
	\newblock \showarticletitle{Towards Simple, High-performance Schedulers for
		High-aggregate Bandwidth Switches}. In
	\bibinfo{booktitle}{\emph{\INFOCOM{21}{02}}}. \bibinfo{publisher}{{IEEE}
		Computer Society}, \bibinfo{pages}{1160--1169}.
	\newblock
	
	
	\bibitem[Gibbens et~al\mbox{.}(1988)]%
	{GKK88}
	\bibfield{author}{\bibinfo{person}{Richard~J. Gibbens},
		\bibinfo{person}{Frank~P. Kelly}, {and} \bibinfo{person}{Peter~B. Key}.}
	\bibinfo{year}{1988}\natexlab{}.
	\newblock \showarticletitle{Dynamic alternative routing -- modelling and
		behavior}. In \bibinfo{booktitle}{\emph{12th International Teletraffic
			Congress}}. \bibinfo{publisher}{Elsevier, Amsterdam}.
	\newblock
	
	
	\bibitem[Godfrey(2008)]%
	{G08}
	\bibfield{author}{\bibinfo{person}{Brighten Godfrey}.}
	\bibinfo{year}{2008}\natexlab{}.
	\newblock \showarticletitle{Balls and bins with structure: balanced allocations
		on hypergraphs}. In \bibinfo{booktitle}{\emph{\SODA{19}{08}}}.
	\bibinfo{publisher}{ACM}, \bibinfo{pages}{511--517}.
	\newblock
	
	
	\bibitem[Godfrey and Stoica(2005)]%
	{GS05}
	\bibfield{author}{\bibinfo{person}{Brighten Godfrey} {and} \bibinfo{person}{Ion
			Stoica}.} \bibinfo{year}{2005}\natexlab{}.
	\newblock \showarticletitle{Heterogeneity and load balance in distributed hash
		tables}. In \bibinfo{booktitle}{\emph{\INFOCOM{24}{05}}}.
	\bibinfo{publisher}{{IEEE}}, \bibinfo{pages}{596--606}.
	\newblock
	\href{https://doi.org/10.1109/INFCOM.2005.1497926}{\texttt{doi}}
	
	
	\bibitem[Greenhill et~al\mbox{.}(2020)]%
	{GMP20}
	\bibfield{author}{\bibinfo{person}{Catherine Greenhill},
		\bibinfo{person}{Bernard Mans}, {and} \bibinfo{person}{Ali Pourmiri}.}
	\bibinfo{year}{2020}\natexlab{}.
	\newblock \showarticletitle{{Balanced Allocation on Dynamic Hypergraphs}}. In
	\bibinfo{booktitle}{\emph{\RANDOM{24}{20}}} \emph{(\bibinfo{series}{Leibniz
			International Proceedings in Informatics (LIPIcs)},
		Vol.~\bibinfo{volume}{176})}. \bibinfo{publisher}{Schloss
		Dagstuhl--Leibniz-Zentrum f{\"u}r Informatik}, \bibinfo{pages}{11:1--11:22}.
	\newblock
	\showISBNx{978-3-95977-164-1}
	\showISSN{1868-8969}
	\href{https://doi.org/10.4230/LIPIcs.APPROX/RANDOM.2020.11}{\texttt{doi}}
	
	
	\bibitem[Karp et~al\mbox{.}(1996)]%
	{KLM96}
	\bibfield{author}{\bibinfo{person}{Richard~M. Karp}, \bibinfo{person}{Michael
			Luby}, {and} \bibinfo{person}{Friedhelm Meyer auf~der Heide}.}
	\bibinfo{year}{1996}\natexlab{}.
	\newblock \showarticletitle{Efficient {PRAM} simulation on a distributed memory
		machine}.
	\newblock \bibinfo{journal}{\emph{Algorithmica}} \bibinfo{volume}{16},
	\bibinfo{number}{4-5} (\bibinfo{year}{1996}), \bibinfo{pages}{517--542}.
	\newblock
	\showISSN{0178-4617}
	\href{https://doi.org/10.1007/s004539900063}{\texttt{doi}}
	
	
	\bibitem[Kirsch et~al\mbox{.}(0910)]%
	{KMW09}
	\bibfield{author}{\bibinfo{person}{Adam Kirsch}, \bibinfo{person}{Michael
			Mitzenmacher}, {and} \bibinfo{person}{Udi Wieder}.}
	\bibinfo{year}{2009/10}\natexlab{}.
	\newblock \showarticletitle{More robust hashing: cuckoo hashing with a stash}.
	\newblock \bibinfo{journal}{\emph{SIAM J. Comput.}} \bibinfo{volume}{39},
	\bibinfo{number}{4} (\bibinfo{year}{2009/10}), \bibinfo{pages}{1543--1561}.
	\newblock
	\showISSN{0097-5397}
	\href{https://doi.org/10.1137/080728743}{\texttt{doi}}
	
	
	\bibitem[Los and Sauerwald(2022a)]%
	{LS22Batched}
	\bibfield{author}{\bibinfo{person}{Dimitrios Los} {and} \bibinfo{person}{Thomas
			Sauerwald}.} \bibinfo{year}{2022}\natexlab{a}.
	\newblock \showarticletitle{Balanced Allocations in Batches: Simplified and
		Generalized}. In \bibinfo{booktitle}{\emph{\SPAA{34}{22}}}.
	\bibinfo{publisher}{ACM}, \bibinfo{pages}{389–399}.
	\newblock
	\showISBNx{9781450391467}
	\href{https://doi.org/10.1145/3490148.3538593}{\texttt{doi}}
	
	
	\bibitem[Los and Sauerwald(2022b)]%
	{LS22Queries}
	\bibfield{author}{\bibinfo{person}{Dimitrios Los} {and} \bibinfo{person}{Thomas
			Sauerwald}.} \bibinfo{year}{2022}\natexlab{b}.
	\newblock \showarticletitle{{Balanced Allocations with Incomplete Information:
			The Power of Two Queries}}. In \bibinfo{booktitle}{\emph{\ITCS{13}{22}}}
	\emph{(\bibinfo{series}{Leibniz International Proceedings in Informatics
			(LIPIcs)}, Vol.~\bibinfo{volume}{215})}. \bibinfo{publisher}{Schloss Dagstuhl
		-- Leibniz-Zentrum f{\"u}r Informatik}, \bibinfo{pages}{103:1--103:23}.
	\newblock
	\showISBNx{978-3-95977-217-4}
	\showISSN{1868-8969}
	\href{https://doi.org/10.4230/LIPIcs.ITCS.2022.103}{\texttt{doi}}
	
	
	\bibitem[Los and Sauerwald(2022c)]%
	{LS22Noise}
	\bibfield{author}{\bibinfo{person}{Dimitrios Los} {and} \bibinfo{person}{Thomas
			Sauerwald}.} \bibinfo{year}{2022}\natexlab{c}.
	\newblock \showarticletitle{Balanced Allocations with the Choice of Noise}. In
	\bibinfo{booktitle}{\emph{\PODC{41}{22}}} \emph{(\bibinfo{series}{PODC'22})}.
	\bibinfo{publisher}{ACM}, \bibinfo{pages}{164–175}.
	\newblock
	\showISBNx{9781450392624}
	\href{https://doi.org/10.1145/3519270.3538428}{\texttt{doi}}
	
	
	\bibitem[Los et~al\mbox{.}(2022)]%
	{LSS22}
	\bibfield{author}{\bibinfo{person}{Dimitrios Los}, \bibinfo{person}{Thomas
			Sauerwald}, {and} \bibinfo{person}{John Sylvester}.}
	\bibinfo{year}{2022}\natexlab{}.
	\newblock \showarticletitle{{Balanced Allocations: Caching and Packing,
			Twinning and Thinning}}. In \bibinfo{booktitle}{\emph{\SODA{33}{22}}}.
	\bibinfo{publisher}{{SIAM}}, \bibinfo{pages}{1847--1874}.
	\newblock
	\href{https://doi.org/10.1137/1.9781611977073.74}{\texttt{doi}}
	
	
	\bibitem[Los et~al\mbox{.}(2023)]%
	{LSS23}
	\bibfield{author}{\bibinfo{person}{Dimitrios Los}, \bibinfo{person}{Thomas
			Sauerwald}, {and} \bibinfo{person}{John Sylvester}.}
	\bibinfo{year}{2023}\natexlab{}.
	\newblock \showarticletitle{Balanced Allocations with Heterogeneous Bins: The
		Power of Memory}. In \bibinfo{booktitle}{\emph{Proceedings of the 2023
			{ACM-SIAM} Symposium on Discrete Algorithms, {SODA} 2023}}.
	\bibinfo{publisher}{{SIAM}}, \bibinfo{pages}{4448--4477}.
	\newblock
	\href{https://doi.org/10.1137/1.9781611977554.ch169}{\texttt{doi}}
	
	
	\bibitem[Luczak and Norris(2013)]%
	{LN13}
	\bibfield{author}{\bibinfo{person}{Malwina~J. Luczak} {and}
		\bibinfo{person}{James~R. Norris}.} \bibinfo{year}{2013}\natexlab{}.
	\newblock \showarticletitle{{Averaging over fast variables in the fluid limit
			for Markov chains: Application to the supermarket model with memory}}.
	\newblock \bibinfo{journal}{\emph{The Annals of Applied Probability}}
	\bibinfo{volume}{23}, \bibinfo{number}{3} (\bibinfo{year}{2013}),
	\bibinfo{pages}{957 -- 986}.
	\newblock
	
	
	\bibitem[Marshall et~al\mbox{.}(2011)]%
	{MRBook}
	\bibfield{author}{\bibinfo{person}{Albert~W. Marshall}, \bibinfo{person}{Ingram
			Olkin}, {and} \bibinfo{person}{Barry~C. Arnold}.}
	\bibinfo{year}{2011}\natexlab{}.
	\newblock \bibinfo{booktitle}{\emph{Inequalities: theory of majorization and
			its applications} (\bibinfo{edition}{second} ed.)}.
	\newblock \bibinfo{publisher}{Springer, New York}. xxviii+909 pages.
	\newblock
	\showISBNx{978-0-387-40087-7}
	\href{https://doi.org/10.1007/978-0-387-68276-1}{\texttt{doi}}
	
	
	\bibitem[Mitzenmacher et~al\mbox{.}(2002)]%
	{MPS02}
	\bibfield{author}{\bibinfo{person}{Michael Mitzenmacher},
		\bibinfo{person}{Balaji Prabhakar}, {and} \bibinfo{person}{Devavrat Shah}.}
	\bibinfo{year}{2002}\natexlab{}.
	\newblock \showarticletitle{Load Balancing with Memory}. In
	\bibinfo{booktitle}{\emph{\FOCS{43}{02}}}. \bibinfo{publisher}{{IEEE}},
	\bibinfo{pages}{799--808}.
	\newblock
	\href{https://doi.org/10.1109/SFCS.2002.1182005}{\texttt{doi}}
	
	
	\bibitem[Mitzenmacher et~al\mbox{.}(2001)]%
	{MRS01}
	\bibfield{author}{\bibinfo{person}{Michael Mitzenmacher},
		\bibinfo{person}{Andr\'{e}a~W. Richa}, {and} \bibinfo{person}{Ramesh
			Sitaraman}.} \bibinfo{year}{2001}\natexlab{}.
	\newblock \showarticletitle{The power of two random choices: a survey of
		techniques and results}.
	\newblock In \bibinfo{booktitle}{\emph{Handbook of randomized computing, {V}ol.
			{I}, {II}}}. \bibinfo{series}{Comb. Optim.}, Vol.~\bibinfo{volume}{9}.
	\bibinfo{publisher}{Kluwer Acad. Publ., Dordrecht},
	\bibinfo{pages}{255--312}.
	\newblock
	\href{https://doi.org/10.1007/978-1-4615-0013-1\_9}{\texttt{doi}}
	
	
	\bibitem[Peres et~al\mbox{.}(2015)]%
	{PTW15}
	\bibfield{author}{\bibinfo{person}{Yuval Peres}, \bibinfo{person}{Kunal
			Talwar}, {and} \bibinfo{person}{Udi Wieder}.}
	\bibinfo{year}{2015}\natexlab{}.
	\newblock \showarticletitle{Graphical balanced allocations and the
		{$(1+\beta)$}-choice process}.
	\newblock \bibinfo{journal}{\emph{Random Structures \& Algorithms}}
	\bibinfo{volume}{47}, \bibinfo{number}{4} (\bibinfo{year}{2015}),
	\bibinfo{pages}{760--775}.
	\newblock
	\showISSN{1042-9832}
	\href{https://doi.org/10.1002/rsa.20558}{\texttt{doi}}
	
	
	\bibitem[Popoviciu(1935)]%
	{P35}
	\bibfield{author}{\bibinfo{person}{Tiberiu Popoviciu}.}
	\bibinfo{year}{1935}\natexlab{}.
	\newblock \showarticletitle{Sur les équations algébriques ayant toutes leurs
		racines réelles}.
	\newblock \bibinfo{journal}{\emph{Mathematica (Cluj)}}  \bibinfo{volume}{9}
	(\bibinfo{year}{1935}), \bibinfo{pages}{129–145}.
	\newblock
	
	
	\bibitem[Shah and Prabhakar(2002)]%
	{SP02}
	\bibfield{author}{\bibinfo{person}{Devavrat Shah} {and} \bibinfo{person}{Balaji
			Prabhakar}.} \bibinfo{year}{2002}\natexlab{}.
	\newblock \showarticletitle{The use of memory in randomized load balancing}. In
	\bibinfo{booktitle}{\emph{\ISIT{02}}}. \bibinfo{pages}{125}.
	\newblock
	
	
	\bibitem[Talwar and Wieder(2014)]%
	{TW14}
	\bibfield{author}{\bibinfo{person}{Kunal Talwar} {and} \bibinfo{person}{Udi
			Wieder}.} \bibinfo{year}{2014}\natexlab{}.
	\newblock \showarticletitle{Balanced Allocations: {A} Simple Proof for the
		Heavily Loaded Case}. In \bibinfo{booktitle}{\emph{\ICALP{41}{14}}}
	\emph{(\bibinfo{series}{Lecture Notes in Computer Science},
		Vol.~\bibinfo{volume}{8572})}. \bibinfo{publisher}{Springer},
	\bibinfo{pages}{979--990}.
	\newblock
	\href{https://doi.org/10.1007/978-3-662-43948-7\_81}{\texttt{doi}}
	
	
	\bibitem[V{\"{o}}cking(1999)]%
	{V99}
	\bibfield{author}{\bibinfo{person}{Berthold V{\"{o}}cking}.}
	\bibinfo{year}{1999}\natexlab{}.
	\newblock \showarticletitle{How Asymmetry Helps Load Balancing}. In
	\bibinfo{booktitle}{\emph{\FOCS{40}{99}}}. \bibinfo{publisher}{{IEEE}},
	\bibinfo{pages}{131--141}.
	\newblock
	\href{https://doi.org/10.1109/SFFCS.1999.814585}{\texttt{doi}}
	
	
	\bibitem[Wieder(2007)]%
	{W07}
	\bibfield{author}{\bibinfo{person}{Udi Wieder}.}
	\bibinfo{year}{2007}\natexlab{}.
	\newblock \showarticletitle{Balanced allocations with heterogenous bins}. In
	\bibinfo{booktitle}{\emph{\SPAA{19}{07}}}. \bibinfo{publisher}{{ACM}},
	\bibinfo{pages}{188--193}.
	\newblock
	\href{https://doi.org/10.1145/1248377.1248407}{\texttt{doi}}
	
	
	\bibitem[Wieder(2017)]%
	{W17}
	\bibfield{author}{\bibinfo{person}{Udi Wieder}.}
	\bibinfo{year}{2017}\natexlab{}.
	\newblock \showarticletitle{Hashing, Load Balancing and Multiple Choice}.
	\newblock \bibinfo{journal}{\emph{Found. Trends Theor. Comput. Sci.}}
	\bibinfo{volume}{12}, \bibinfo{number}{3-4} (\bibinfo{year}{2017}),
	\bibinfo{pages}{275--379}.
	\newblock
	\href{https://doi.org/10.1561/0400000070}{\texttt{doi}}
	
	
\end{thebibliography}
\end{document}